\documentclass[journal,onecolumn,11pt]{IEEEtran}
\IEEEoverridecommandlockouts
\usepackage{amsmath}
\usepackage{amsfonts}
\usepackage{graphicx}
\usepackage{url}
\usepackage{amssymb}
\usepackage{graphicx}
\usepackage{caption}
\usepackage{subcaption}
\usepackage{spverbatim}
\usepackage{listings,wasysym}
\usepackage{hyperref}
\usepackage{listings}
\usepackage{enumerate}
\ifCLASSINFOpdf
\else
\fi
\newtheorem{theorem}{Theorem}[section]
\newtheorem{lemma}[theorem]{Lemma}
\newtheorem{proposition}[theorem]{Proposition}
\newtheorem{corollary}[theorem]{Corollary}
\newtheorem{remark}[theorem]{Remark}

\newcommand{\ba}{\boldsymbol{a}}
\newcommand{\bb}{\boldsymbol{b}}

\newcommand{\be}{\boldsymbol{e}}
\newcommand{\bff}{\boldsymbol{f}}
\newcommand{\bg}{\boldsymbol{g}}
\newcommand{\bh}{\boldsymbol{h}}

\newcommand{\bp}{\boldsymbol{p}}
\newcommand{\bq}{\boldsymbol{q}}

\newcommand{\bu}{\boldsymbol{u}}
\newcommand{\bv}{\boldsymbol{v}}
\newcommand{\bw}{\boldsymbol{w}}
\newcommand{\bx}{\boldsymbol{x}}
\newcommand{\by}{\boldsymbol{y}}
\newcommand{\bz}{\boldsymbol{z}}

\newcommand{\blambda}{\boldsymbol{\lambda}}

\newcommand{\BA}{\boldsymbol{A}}
\newcommand{\BB}{\boldsymbol{B}}

\newcommand{\BD}{\boldsymbol{D}}
\newcommand{\BE}{\boldsymbol{E}}
\newcommand{\BF}{\boldsymbol{F}}
\newcommand{\BI}{\boldsymbol{I}}
\newcommand{\BJ}{\boldsymbol{J}}
\newcommand{\BH}{\boldsymbol{H}}
\newcommand{\BM}{\boldsymbol{M}}
\newcommand{\BP}{\boldsymbol{P}}
\newcommand{\BS}{\boldsymbol{S}}
\newcommand{\BT}{\boldsymbol{T}}
\newcommand{\BU}{\boldsymbol{U}}
\newcommand{\BV}{\boldsymbol{V}}
\newcommand{\BW}{\boldsymbol{W}}
\newcommand{\BX}{\boldsymbol{X}}
\newcommand{\BY}{\boldsymbol{Y}}
\newcommand{\BZ}{\boldsymbol{Z}}

\newcommand{\BPHI}{\boldsymbol{\Phi}}
\newcommand{\BPHIB}{\boldsymbol{\Phi}^{\bot}}

\newcommand{\BPhi}{\boldsymbol{\Phi}}
\newcommand{\bzero}{\boldsymbol{0}}

\newcommand{\beps}{\boldsymbol{\eps}}
\newcommand{\btheta}{\boldsymbol{\theta}}

\newcommand{\hby}{\hat{\boldsymbol{y}}}
\newcommand{\hBX}{\hat{\boldsymbol{X}}}

\newcommand{\tbu}{\widetilde{\boldsymbol{u}}}
\newcommand{\tbv}{\widetilde{\boldsymbol{v}}}

\newcommand{\A}{\mathcal{A}}

\newcommand{\PP}{\mathcal{P}}

\newcommand{\CM}{\mathcal{M}}

\renewcommand{\Pr}{\mathbb{P}}

\newcommand{\CC}{\mathbb{C}}

\newcommand{\CZ}{\mathcal{Z}}

\newcommand{\muh}{\mu_h}
\newcommand{\mum}{\mu_{\max}}
\newcommand{\mumin}{\mu_{\min}}

\newcommand{\I}{\boldsymbol{I}}
\newcommand{\RR}{\mathbb{R}}
\newcommand{\lag}{\langle}
\newcommand{\rag}{\rangle}
\newcommand{\ol}{\overline}

\newcommand{\lp}{\left(} 
\newcommand{\rp}{\right)} 
\newcommand{\ls}{\left[} 
\newcommand{\rs}{\right]} 

\newcommand{\Tr}{\text{Tr}}
\newcommand{\eps}{\epsilon}
\newcommand{\TB}{T^{\bot}}

\DeclareMathOperator{\VEC}{vec}
\DeclareMathOperator{\E}{\mathbb E}

\DeclareMathOperator{\diag}{diag}
\DeclareMathOperator{\Null}{Null}
\DeclareMathOperator{\SNR}{SNR}

\DeclareMathOperator{\Ran}{Ran}
\DeclareMathOperator{\mix}{mix}
\DeclareMathOperator{\subjectto}{\text{subject to}}

\begin{document}
%
\title{ {Blind Deconvolution Meets Blind Demixing: \\ Algorithms and Performance Bounds} \thanks{This work was supported by the National Science Foundation under  grant~DTRA-DMS 1322393 and~DMS 1620455. This paper was also presented at  the 2016 50th Asilomar Conference on the Signals, Systems and Computers, see~\cite{LS16b}.} }


\author{Shuyang~Ling, 
        Thomas~Strohmer
\thanks{S.~Ling and T.~Strohmer are with the Department of Mathematics, University of California at Davis, Davis, CA 95616, USA (E-mail: syling@math.ucdavis.edu; strohmer@math.ucdavis.edu).}}


%

\markboth{IEEE Transactions on Information Theory}{}


\maketitle



%

\begin{abstract}
Suppose that we have $r$ sensors and each one intends to send a function $\boldsymbol{g}_i$ (e.g.\ a signal or an image)  to a
receiver common to all $r$ sensors.
During transmission, each  $\boldsymbol{g}_i$  gets convolved with a function $\boldsymbol{f}_i$.
The receiver records the function $\by$, given by the sum of all these convolved signals.
When and under which conditions is it possible to recover the individual signals $\boldsymbol{g}_i$ and the blurring functions 
$\boldsymbol{f}_i$
from just one received signal $\boldsymbol{y}$? This challenging problem, which intertwines blind deconvolution with blind demixing, appears in a variety of applications, such as audio processing, image
processing, neuroscience, spectroscopy, and astronomy. It is also expected to play a central role in
connection with the future Internet-of-Things. We will prove that under reasonable and practical assumptions, it is possible
to solve this otherwise highly ill-posed problem and recover the $r$ transmitted functions $\boldsymbol{g}_i$  and the impulse responses $\boldsymbol{f}_i$ in a robust, reliable, and efficient
manner from just one single received function $\by$ by solving a semidefinite program. We derive explicit bounds on the number
of measurements needed for successful recovery and prove that our method is robust in the presence of noise. Our theory is actually sub-optimal, since numerical experiments demonstrate that, quite remarkably, recovery is still possible if the number of measurements is close to the number of degrees of freedom.
\end{abstract}


\IEEEpeerreviewmaketitle

\section{Introduction\label{s:intro}}

Suppose we are given $r$ sensors and each one sends a function $\bg_i$ (e.g.\ a signal or an image)  to a receiver common to all $r$ sensors. 
During transmission each  $\bg_i$  gets convolved with a function $\bff_i$ where $\{\bff_i\}_{i=1}^r$ may all differ from each other.
The receiver records the function $\by$, given by the sum of all these convolved signals. More precisely,
\begin{equation}
\by = \sum_{i=1}^{r}  \bff_i \ast \bg_i + \bw,
\label{convmodel}
\end{equation}
where $\bw$ is additive noise. Assume that the receiver  knows neither $\{\bff_i\}_{i=1}^r$ nor $\{\bg_i\}_{i=1}^r$. 
When and under which conditions is it possible to recover all the individual signals $\bff_i$ and  $\bg_i$ from just one received signal $\by$? 

Blind deconvolution (when $r = 1$) by itself is already a hard problem to solve. Here we deal with the even more difficult situation -- \emph{a mixture of blind deconvolution problems}. 
Thus we need to correctly blindly deconvolve and demix at the same time.
This challenging problem appears in a variety of applications, such as audio processing~\cite{LXQZ09}, image processing~\cite{SSZ05,SZY11}, 
neuroscience~\cite{SGP15}, spectroscopy~\cite{Toumi14}, astronomy~\cite{Cardoso2002}. It also arises in wireless communications\footnote{In  wireless communications this is also known as ``multiuser joint 
channel estimation and equalization.''}~\cite{WP98} and is expected to play a central role in 
connection with the future Internet-of-Things~\cite{WBSJ14}.
Common to almost all approaches to tackle this problem is the assumption that
we have multiple received signals at our disposal, often at least as many received signals as there are transmitted signals. Indeed, many of the 
existing methods fail if the assumption of multiple received signals is not fulfilled.

In this paper, we consider the rather difficult case, where only one received signal is given, as  shown in~\eqref{convmodel}.
Of course, without further assumptions, this problem is highly underdetermined and not solvable.
We will prove that under reasonable and practical conditions, it is indeed
possible to recover the $r$ transmitted signals and the associated channels in a robust, reliable, and efficient manner from just one single received signal.
Our theory has important implications for applications, such as the Internet-of-Things, since it paves the way for an efficient multi-sensor communication strategy
with minimal signaling overhead.

\vskip0.25cm
To provide a glimpse of the kind of results we will prove, let us assume that each of the $\bg_i \in \RR^L$ lies in a known subspace of dimension $N$ ($L\geq N$), i.e., there exist matrices $\BA_i$ of size $L \times N$ such that $\bg_i = \BA_i \bx_i$. In addition the matrices $\BA_i$ need to satisfy a certain ``local'' mutual incoherence condition 
described in detail in~\eqref{def:mu}. This condition can be satisfied if the $\BA_i$ are e.g.\ Gaussian random matrices.
We will prove a formal and slightly more general version (see Theorem~\ref{thm:main} and Theorem~\ref{thm:noise}) of the following informal theorem.
For simplicity for the moment we consider a noiseless scenario, that is, $\bw=0$. Below and throughout the paper ``$\ast$" denotes circular convolution.

\begin{theorem}{\bf [Informal version]}\label{thm:informal}
Let $\bx_i \in \RR^N$ and let the $\BA_i$ be $L \times N$ i.i.d.\ Gaussian random matrices. Furthermore, assume that the impulse responses $\bff_i \in \CC^L$ have {\em maximum delay spread} $K$, i.e., for each $\bff_i$ there holds $\bff_i(k) = 0$ if $k > K$.  Let  $\mu^2_h$ be a certain  ``incoherence parameter'' related to the measurement matrices, 
defined in~\eqref{def:muh}.
Suppose we are given
\begin{equation}
\by = \sum_{i=1}^{r}    \bff_i \ast (\BA_i \bx_i).
\label{convmodel1}
\end{equation}
Then, as long as the number of measurements $L$ satisfies
\begin{equation*}
L \gtrsim C r^2 \max\{  K,\mu^2_h N\},
\end{equation*}
(where $C$ is a numerical constant), all $\bx_i$ (and thus $\bg_i = \BA_i \bx_i$) as well as all $\bff_i$ can be recovered from $\by$ with high probability by solving a semidefinite program.
\end{theorem}

Recovering $\{\bff_i\}_{i=1}^r$ and $\{\bx_i\}_{i=1}^r$ is only possible up to a constant, since we can always multiply each
$\bx_i$ with $c_i\neq 0$ and each $\bff_i$ with $1/c_i$ and still get the same result. Hence, here and throughout the paper, recovery
of the vectors $\bff_i$ and $\bx_i$ always means recovery modulo constants $c_i$.

\bigskip
We point out that the emphasis of this paper is on developing a theoretical and algorithmic framework for joint blind deconvolution and blind demixing. A detailed discussion of applications
is beyond the scope of this paper. There are several aspects, such as time synchronization, that do play a role in some applications and need further attention. We postpone
such details to a forthcoming paper, in which we plan to elaborate on the proposed framework in connection with specific applications.

\subsection{Related work}

Problems of the type~\eqref{convmodel} or~\eqref{convmodel1} are ubiquitous in many applied scientific disciplines and in applications,  
see e.g~\cite{gesbert1999blind,WP98,LXQZ09,SSZ05,shin2007blind,li2001direct,SZY11,SGP15,Toumi14,Cardoso2002,WBSJ14}. 
Thus, there is a large body of works to solve different versions of these problems. Most of the existing works however require the availability of
multiple received signals $\by_1,\dots,\by_m$. And indeed, it is not hard to imagine that for instance an SVD-based
approach will succeed if $m \ge r$ (and must fail if $m=1$).  A sparsity-based approach can be found in~\cite{sudhakar2010double}.
However, in this paper we are interested in the case where we have only one single received signal $\by$ --
a single snapshot, in the jargon of array processing. Hence, there is little overlap between these methods heavily relying on multiple
snapshots (many of which do not come with any theory) and the work presented here.

The setup in~\eqref{convmodel1} is reminiscent of a single-antenna multi-user spread spectrum communication scenario~\cite{Ver98}. 
There, the matrix $\BA_i$ represents the spreading matrix assigned to the $i$-th user and $\bff_i$  models the associated multipath channel. 
There are numerous papers on blind channel estimation in connection with CDMA, including the previously cited articles~\cite{gesbert1999blind,WP98,li2001direct}.
Our work differs from the existing literature on this topic in several ways: As mentioned before, we do not require that we have multiple received signals, we allow all 
multipath channels $\bff_i$ to differ from each other, and do not impose a particular channel model. Moreover, 
we provide a rigorous mathematical theory, instead of just empirical observations.

The special case $r=1$ (one unknown signal and one unknown convolving function) reduces~\eqref{convmodel} to the standard blind deconvolution problem, which has been heavily studied in the literature, cf.~\cite{chaudhuri2014blind} and the references therein. Many of the techniques 
for ``ordinary'' blind deconvolution do not extend (at least not in any obvious manner) to the case $r>1$. Hence, there is essentially no overlap
with this work -- with one notable exception. The pioneering paper~\cite{RR12} has definitely inspired our work and also informed many of the proof techniques used in this paper. Hence, our paper can and should be seen as an extension of the ``single-user'' ($r=1$) 
results in~\cite{RR12}  to the multi-user setting ($r>1$). However, it will not come as a big surprise to the reader familiar with~\cite{RR12}, that there is no simple way to extend the results in~\cite{RR12} to the multi-user setting unless  we assume that we have multiple received signals $\by_1,\dots,\by_m$. Indeed, as may be obvious from the length of the proofs in our paper, there are substantial differences in the theoretical derivations between this manuscript and~\cite{RR12}. In particular, the sufficient condition for exact recovery in this paper is more complicated since $r$ ($r >1$) users are considered and the ``incoherence" between users needs to be introduced properly. Moreover,  the construction of approximate dual certificate is nontrivial as well (See Section~\ref{s:dual}) in the ``multi-user" scenario.

The paper~\cite{ACD15} considers the following generalization of~\cite{RR12}\footnote{Since the main result in~\cite{ACD15} relies on Lemma~4 of~\cite{RR12}, the issues raised in Remark~\ref{remark_lemma} apply
to~\cite{ACD15}  as well.}. 
Assume that we are given signals 
$\by_i = \bff \ast \bg_i, i=1,\dots, r$, the goal is to recover the $\bff$ and $\bg_i$ from $\by_1,\dots,\by_r$. This setting is somewhat in the spirit of~\eqref{convmodel}, but
it is significantly less challenging, since (i) it assumes the same convolution function $\bff$ for each signal $\bg_i$ and (ii) there are as many output signals $\by_i$  as we have input signals $\bg_i$.

Non-blind versions of~\eqref{convmodel} or~\eqref{convmodel1} can be found for instance in~\cite{WGMM13,mccoy2014sharp,mccoy2014convexity,ALMT14}.
In the very interesting paper~\cite{WGMM13}, the authors analyze various problems of decomposing a given observation into multiple incoherent components, which can be expressed as
\begin{equation}
\label{wright}
\text{minimize}  \,\,\, \sum_i \lambda_i \|\BZ_i \|_{(i)}   \qquad \text{subject to} \quad \sum_i \BZ_i = \BM.
\end{equation}
Here $\| \cdot \|_{(i)}$ are (decomposable) norms that encourage various types of low-complexity structure. However, as mentioned before, there is no ``blind'' component in the problems analyzed in~\cite{WGMM13}. Moreover, while~\eqref{wright} is formally somewhat similar to the semidefinite program that we derive to solve the blind deconvolution-blind demixing problem (see~\eqref{cvxprog}), the dissimilarity of the right-hand sides in~\eqref{wright} 
and~\eqref{cvxprog} makes all the differences when theoretically analyzing these two problems.

The current manuscript can as well be seen as an extension of our work on self-calibration~\cite{LS15} to the multi-sensor  case.
In this context, we also refer to related (single-input-single-output) analysis in~\cite{LLB16ID,Chi16}.

\subsection{Organization of this manuscript}

In  Section~\ref{s:prelim} we describe in detail the setup and the problem we are solving. We also introduce some notations and key concepts used throughout the manuscript.
The main results for the noiseless as well as the noisy case are stated in Section~\ref{s:maintheorem}. 
Numerical experiments can be found in Section~\ref{s:numerics}. 
Section~\ref{s:proofs} is devoted to the proofs of these results.
We conclude in Section~\ref{s:conclusion} and present some auxiliary results in the Appendix.

\section{Preliminaries and Basic Setup\label{s:prelim}}

\subsection{Notation}
Before moving to the basic model, we introduce notation which will be used throughout the paper. Matrices and vectors are denoted in boldface such as $\BZ$ and $\bz$. The individual entries of a matrix or a vector are denoted in normal font such as $Z_{ij}$ or $z_i.$
For any matrix $\BZ$, $\|\BZ\|_*$ denotes nuclear norm, i.e., the sum of its singular values; $\|\BZ\|$ denotes operator norm, i.e., its largest singular value, and $\|\BZ\|_F$ denotes the Frobenius norm, i.e., $\|\BZ\|_F =\sqrt{\sum_{ij} |Z_{ij}|^2 }$. For any vector $\bz$, $\|\bz\|$ denotes its Euclidean norm. For both matrices and vectors, $\BZ^T$ and $\bz^T$ stand for the transpose of $\BZ$ and $\bz$ respectively while $\BZ^*$ and $\bz^*$ denote their complex conjugate transpose.  $\bar{z}$ and $\bar{\bz}$  denote the complex conjugate of $z$ and $\bz$ respectively. We equip the matrix space $\CC^{K\times N}$ with the inner product defined as $\lag \BU, \BV\rag : = \Tr(\BU\BV^*).$ A special case is the inner product of two vectors, i.e., $\lag \bu, \bv\rag = \Tr(\bu\bv^*) = \bv^*\bu = (\bu^*\bv)^*.$ The identity matrix of size $n$ is denoted by $\BI_n$. For a given vector $\bv$, $\diag(\bv)$ represents the diagonal matrix whose diagonal entries are given by the vector $\bv$. 

Throughout the paper, $C$ stands for a constant and $C_{\alpha}$ is a constant which depends linearly on $\alpha$ (and on no other numbers). For the two linear subspaces $T_i$ and $\TB_i$ defined in~\eqref{cond:proj} and~\eqref{cond:projc}, we denote the projection of $\BZ$ on $T_i$ and $\TB_i$ as $\BZ_{T_i} : = \PP_{T_i}(\BZ)$ and $\BZ_{\TB_i} : = \PP_{\TB_i}(\BZ)$ respectively. $\PP_{T_i}$ and $\PP_{\TB_i}$ are the corresponding projection operators onto $T_i$ and $\TB_i.$ 
\subsection{The basic model}

We develop our theory for  a more general model than the blind deconvolution/blind demixing model discussed in Section~\ref{s:intro}.
Our framework also covers certain self-calibration scenarios~\cite{LS15} involving multiple sensors. We consider the following setup\footnote{ In~\eqref{eq:model} we assume
a {\em common clock} among the different sources. For sources whose distance to the receiver differs greatly, his assumption would require additional synchronization.
A detailed discussion of this timing aspect is beyond the scope of this paper, as it is application dependent.}
\begin{equation}\label{eq:model}
\by = \sum_{i=1}^r\diag(\BB_i\bh_i) \BA_i \bx_i,
\end{equation}
where $\by \in \CC^L$, $\BB_i \in \CC^{L\times K_i}$, $\BA_i \in \RR^{L\times N_i}$, $\bh_i \in \RR^{K_i}$ and $\bx_i \in \RR^{N_i}$.
We assume that all the matrices $\BB_i$ and  $\BA_i$ are given, but none of the  $\bx_i$ and  $\bh_i$ are known.
Note that all $\bh_i$ and $\bx_i$ can be of different lengths. We point out that the total number of measurements is given by the
length of $\by$, i.e., by $L$. Moreover, we let $K : = \max K_i$ and $N: = \max  N_i$ throughout our presentation. 

This model includes the blind deconvolution-blind demixing problem~\eqref{convmodel} as a special case, as we will explain in Section~\ref{s:maintheorem}.
But it also includes other cases as well. Consider for instance a linear system $\by = \sum_{i=1}^r \BA_i(\btheta_i) \bx_i$, where the measurement matrices
$\BA_i$ are not fully known due to lack of calibration~\cite{FW91,BPGD13,LS15} and $\btheta_i$ represents the unknown calibration parameters associated with
$\BA_i$. An important special situation that arises e.g.\ in array calibration~\cite{FW91}
is the case where we only know the direction of the rows of  $\BA_i$.  In other words, the norms of each of the rows of $\BA_i$ are unknown.
 If in addition each of the $\btheta_i$ belongs to a known subspace represented by
$\BB_i$, i.e., $\btheta_i = \BB_i \bh_i$, then we can write  such an $\BA_i(\btheta_i)$ as $\BA_i(\btheta_i) = \diag(\BB_i \bh_i) \BA_i$.

\bigskip

Let $\bb_{i,l}$ denote the $l$-th column of $\BB_i^*$ and $\ba_{i,l}$ the $l$-th column of $\BA_{i}^*$. A simple application of linear algebra gives
\begin{equation}\label{def:A}
y_l= \sum_{i=1}^r (\BB_i \bh_i )_l \bx_i^*\ba_{i,l} = \sum_{i=1}^r \bb_{i,l}^*\bh_i \bx_i^*\ba_{i,l}
\end{equation}
where $y_l$ is the $l$-th entry of $\by.$ One may find an obvious difficulty of this problem as the nonlinear relation between the measurement vectors $(\bb_{i,l}, \ba_{i,l})$ and the unknowns $(\bh_i, \bx_i).$
Proceeding with the meanwhile well-established lifting trick~\cite{CSV11}, we let $\BX_i := \bh_i \bx_i^* \in\RR^{K_i\times N_i}$ and define 
 the {\em linear} mapping $\A_i: \CC^{K_i\times N_i}\rightarrow \CC^L$ for $i=1,\dots, r$ by
\begin{equation*}
 \A_i(\BZ) := \{ \bb_{i,l}^* \BZ \ba_{i,l} \}_{l=1}^L.
\end{equation*}
Note that the adjoint operator of $\A_i$ is 
\begin{equation}\label{def:A}
\A_i^*: \CC^L \rightarrow \CC^{K_i\times N_i}, \qquad  \A_i^*(\bz) = \sum_{l=1}^L z_l \bb_{i,l}\ba_{i,l}^*
\end{equation}
since $\CC^{K_i\times N_i}$ is equipped with the inner product $\lag \BU, \BV\rag = \Tr(\BU\BV^*)$ for any $\BU$ and $\BV\in\CC^{K_i\times N_i}$. $\A_i^*(\bz)$ can be also written into simple matrix form, i.e., $\A_i^*(\bz) = \BB_i^*\diag(\bz)\BA_i$, which is easily verified by definition. 
Thus we have lifted the {\em non-linear vector-valued} equations~\eqref{eq:model} to {\em linear matrix-valued} equations given by
\begin{equation}\label{lineq}
\by = \sum_{i=1}^r\A_i(\BX_i).
\end{equation}
Alas, the set of linear equations~\eqref{lineq}  will be highly underdetermined, unless we make the number of measurements $L$ very large, which may not be desirable or feasible in practice. Moreover, finding such $r$ rank-1 matrices  satisfying~\eqref{lineq} is generally an NP-hard problem~\cite{RechtSIAM,Fazel02}. Hence, to combat this underdeterminedness, we attempt to recover $(\bh_i,\bx_i)_{i=1}^{r}$ 
by solving the following nuclear norm minimization problem,
\begin{equation}\label{cvxprog}
\min \,\, \sum_{i=1}^r\|\BZ_i\|_* \quad \subjectto \quad \sum_{i=1}^r\A_i(\BZ_i) = \by.
\end{equation}
If the solutions (or the minimizers to~\eqref{cvxprog}) $\hBX_1,\dots,\hBX_r$ are all rank-one, we can easily extract  $\bh_i$ and $\bx_i$ from $\hBX_i$ via a simple matrix factorization.
In case of noisy data, the $\hBX_i$ will not be exactly rank-one, in which case we set $\bh_i$ and $\bx_i$  to be the left and right singular vector
respectively, associated with the largest singular value of $\hBX_i$.
Naturally, the question arises if and when the solution to~\eqref{cvxprog} coincides with the true solution $(\bh_i,\bx_i)_{i=1}^{r}$.
It is the main purpose of this paper to shed light on this question.

\subsection{Incoherence conditions on the matrices $\BB_i$ \label{ss:incoherence}}

Analogous to matrix completion, where one needs to impose certain incoherence conditions on the singular vectors (see e.g.~\cite{CR08}), 
we introduce two quantities that describe a notion of incoherence of the matrices $\BB_i$. 
We require $\BB_i^*\BB_i = \I_{K_i}$  and define
\begin{align}\label{def:mumax}
\begin{split}
\mum^2 & := \max_{1\leq l\leq L, 1\leq i\leq r} \frac{L}{K_i} \|\bb_{i,l}\|^2, \\ \mumin^2 & := \min_{1\leq l\leq L, 1\leq i\leq r} \frac{L}{K_i} \|\bb_{i,l}\|^2.
\end{split}
\end{align}
$\BB_i^*\BB_i = \I_{K_i}$ implies that
$1\leq \mum^2 \leq \frac{L}{K_i}$ and $0\leq \mumin^2\leq 1$. In particular,  if each $\BB_i$ is a partial DFT matrix then $\mum^2 = \mumin^2 = 1$. The quantity
$\mumin^2$ will be useful to establish Theorem~\ref{thm:noise}, while the main purpose of  introducing $\mu^2_{\max}$ is to quantify 
a ``joint incoherence pattern'' on all $\BB_i$. Namely, there is a {\em common} partition $\{\Gamma_p\}_{p=1}^P$
of the index set $\{1,\cdots, L\}$ with $|\Gamma_p| = Q$ and $L = PQ$ such that for each pair of $(i,p)$ with $1\leq i\leq r$ and $1\leq p\leq P$, we
have

\begin{equation}\label{cond:iso}
\max_{1\leq i\leq r, 1\leq p\leq P}\left\|\BT_{i,p} - \frac{Q}{L}\I_{K_i}\right\| \leq \frac{Q}{4L},\quad \BT_{i,p} : = \sum_{l\in\Gamma_p} \bb_{i,l}\bb_{i,l}^*,
\end{equation}
which says that each $\BT_{i,p}$ does not deviate   too much from $\I_{K_i}$. The  key question here is whether  such a {\em common} partition exists. It is hard to answer it in general. To the best of our knowledge, it is known that for each $\BB_i$, there exists a partition $\{\Gamma_{i,p}\}_{p=1}^P$ (where $\Gamma_{i,p}$ depends on $i$) such that 
\begin{equation*}
\max_{1\leq p\leq P}\left\| \sum_{l\in\Gamma_{i,p}} \bb_{i,l}\bb_{i,l}^*- \frac{Q}{L}\I_{K_i}\right\| \leq \frac{Q}{4L}, \quad \forall 1\leq i\leq r,
\end{equation*}
if $Q \geq C\mu^2_{\max}K_i \log L$ where this argument is shown to be true in~\cite{RR12} by Theorem 1.2 in~\cite{candes07sparsity}.
Based on this observation, at least we have following several special cases which satisfy~\eqref{cond:iso} for a common partition $\{\Gamma_p\}_{p=1}^P$.
\begin{enumerate}
\item All $\BB_i$ are the same. Then the common partition $\{\Gamma_p\}_{p=1}^P$ can be chosen the same as $\{\Gamma_{i,p}\}_{p=1}^P$ for any particular $i.$
\item If each $\BB_i, i\neq j$ is a submatrix of $\BB_j$, then we simply let $\Gamma_p = \Gamma_{j,p}$ such that~\eqref{cond:iso} holds.

\item If all $\BB_i$ are ``low-frequency'' DFT matrices, i.e., the first $K_i$ columns of an $L\times L$ DFT matrix with $\BB_i^*\BB_i = \I_{K_i}$, an {\em explicit} partition of $\Gamma_{p}$ can be constructed such that 
\begin{equation}\label{eq:subfourier}
\BT_{i,p} = \sum_{l\in\Gamma_p}\bb_{i,l}\bb_{i,l}^* = \frac{Q}{L}\I_{K_i}.
\end{equation} 
For example, suppose $L =PQ$ and $Q\geq K_{i}$, we can achieve $\BT_{i,p} = \frac{Q}{L}\I_{K_i}$ and $|\Gamma_p| = Q$ by letting $\Gamma_p = \{p , P +  p, \cdots, (Q - 1)P + p\}$.  A short proof will be provided in Section~\ref{sub:fourier}. 

\end{enumerate}
\begin{remark}
The existence of the common partition $\{\Gamma_p\}_{p=1}^P$ satisfying~\eqref{cond:iso} is an extremely important ingredient of constructing the dual certificate by golfing scheme. 

\end{remark}

Some direct implications of~\eqref{cond:iso}
are
\begin{equation}\label{cond:ST}
\|\BT_{i,p}\| \leq \frac{5Q}{4L}, \quad \|\BS_{i,p}\| \leq \frac{4L}{3Q}, \quad \forall 1\leq i\leq r, 1\leq p\leq P.
\end{equation}
where $\BS_{i,p} := \BT_{i,p}^{-1}$.
Now let us introduce the second incoherence quantity, which is also crucial in the proof of Theorem~\ref{thm:main},
\begin{align}\label{def:muh}
\begin{split}
\mu^2_h & := \max\Big\{ \frac{Q^2}{L} \max_{l\in \Gamma_p, 1\leq p\leq P,1\leq i\leq r} \frac{|\lag \BS_{i,p}\bh_i,
\bb_{i,l}\rag|^2}{\|\bh_i\|^2},  L \max_{1\leq l\leq L, 1\leq i\leq r}   \frac{|\lag \bh_i, \bb_{i,l}\rag|^2}{\|\bh_i\|^2} \Big\}.
\end{split}
\end{align}
The range of $\mu^2_h$ is given in Proposition~\ref{prop:muh}.
\begin{remark}\label{remark_lemma}
The attentive reader may have noticed that the definition of $\mu_h^2$ is a bit more intricate than the one in~\cite{RR12}, where  $\mu_h^2$ only depends on $
|\lag \bh_i, \bb_{i,l}\rag|^2$.
The reason is that we need to establish a result similar to Lemma~4 in~\cite{RR12}, but the proof of Lemma 4
as stated is not entirely accurate,
and a fairly simple way to fix this issue is to slightly modify the definition of $\mu_h^2$. 
Another easy way to fix the issue is to consider all $\BB_i$ as low-frequency Fourier matrices.  If so,  $\muh^2$ in~\eqref{def:muh} reduces to a simpler form of $\mu^2_h$, i.e., $\mu^2_h = L\max_{1\leq l\leq L}  |\lag \bb_l, \bh\rag|^2/\|\bh\|^2$ in~\cite{RR12} because the explicit partition of low-frequency DFT matrices allows $\BT_{i,p} = \frac{Q}{L}\I_{K_i}$ and $\BS_{i,p} = \frac{L}{Q}\I_{K_i}$.
\end{remark}

Both $\mu^2_{\max}$ and $\mu^2_h$ measure the incoherence of $\BB_i$ and the latter one, depending $\bh_i$, also describes the interplay between $\bh_i$ and $\BB_i$. To sum up,  for all  $1\leq l\leq L$ and $1 \leq i \leq  r,$
\begin{eqnarray}
\begin{split}\label{cond:i}
& \|\bb_{i,l}\|^2 \leq \frac{\mum^2K_i}{L}, \quad |\lag \bh_i, \bb_{i,l}\rag|^2 \leq \frac{\mu^2_h}{L}\|\bh_i\|^2, \\
& |\lag \BS_{i,p}\bh_i, \bb_{i,l}\rag|^2 \leq \frac{L\mu^2_h}{Q^2}\|\bh_i\|^2.
\end{split}
\end{eqnarray}
\begin{proposition}{\bf [Range of $\mu^2_h$]}\label{prop:muh}
Under the condition of~\eqref{cond:iso} and~\eqref{cond:ST},
\begin{equation*}
1 \leq  \mu^2_h \leq \frac{16}{9}\mum^2K_i, \quad \forall 1\leq i\leq r.
\end{equation*}
\end{proposition}
\begin{proof}
We start with~\eqref{def:muh} and~\eqref{cond:i} to find the lower bound of $\mu_h^2$ first. Without loss of generality, all $\bh_i$ are of unit norm. The definition of $\mu_h^2$ and $|\Gamma_p| = Q$ immediately imply that
\begin{eqnarray*}
\mu^2_h & \geq &\max_{i,p}\left\{  \frac{Q}{L} \sum_{l\in\Gamma_p} |\lag \BS_{i,p}\bh_i, \bb_{i,l}\rag |^2,
\sum_{l=1}^L |\lag \bh_i, \bb_{i,l}\rag|^2 \right\} \\
& = & \max_{i,p}\left\{  \frac{Q}{L} \sum_{l\in\Gamma_p}\bh_i^*\BS_{i,p} \bb_{i,l}\bb_{i,l}^*\BS_{i,p}\bh_i,
\sum_{l=1}^L \bh_i^*\bb_{i,l}\bb_{i,l}^*\bh_i \right\} \\
& = & \max_{i,p}\left\{ \frac{Q}{L} \bh_i^*\BS_{i,p}\bh_i, 1\right\}.
\end{eqnarray*}
Note that 
\begin{equation*}
1\leq \max_{i,p}\left\{ \frac{Q}{L} \bh_i^*\BS_{i,p}\bh_i, 1\right\} \leq \frac{4}{3},
\end{equation*}
which follows from $\|\BS_{i,p}\| \leq \frac{4L}{3Q}$ and thus we can conclude the lower bound of $\mu^2_h$ is between $1$ and $\frac{4}{3}.$

We proceed to derive the upper bound for $\mu^2_h.$ Applying Cauchy-Schwarz inequality to~\eqref{def:muh} gives
\begin{align*}
\mu^2_h 
& \leq \max_{p,i,l} \left\{ \frac{Q^2}{L} \|\BS_{i,p}\|^2 \|\bb_{i,l}\|^2, L\|\bb_{i,l}\|^2 \right\} \leq \frac{Q^2}{L} \frac{16L^2}{9Q^2} \cdot \frac{\mum^2K_i}{L} \leq \frac{16}{9} \mum^2K_i
\end{align*}
where $\|\BS_{i,p}\| \leq \frac{4L}{3Q}$ and $\|\bb_{i,l}\|^2 \leq \frac{\mum^2K_i}{L}.$
\end{proof}

\subsection{Is  the incoherence parameter $\mu_h^2$ necessary?}
This subsection is devoted to a further discussion of the role of $\mu_h^2.$ In order to provide a clearer explanation of the significance of $\mu_h^2$, we first reformulate the recovery of $\{\BX_i\}_{i=1}^r$ subject to~\eqref{lineq} as a rank-$r$ matrix recovery problem. Each entry of $\by$ is actually the inner product of two rank-$r$ block-diagonal matrices, i.e.,  
\begin{equation*}
y_l = 
\left\lag
\begin{bmatrix}
\bh_1\bx_1^* & \cdots & \bzero \\
\vdots  & \ddots & \vdots \\
\bzero & \cdots & \bh_r\bx_r^* \\
\end{bmatrix}
,
\begin{bmatrix}
\bb_{1,l}\ba_{1,l}^* & \cdots & \bzero \\
\vdots  & \ddots & \vdots \\
\bzero & \cdots & \bb_{r,l}\ba_{r,l}^* \\
\end{bmatrix}
\right\rag.
\end{equation*}
Recall that in matrix completion~\cite{CR08,Recht11MC},  the left and right singular vectors of the true matrix cannot be too aligned with those of the test matrix. A similar spirit applies to this problem as well, i.e., both
\begin{equation}
\label{mu_explained}
\max_{1\leq l\leq L, 1\leq i\leq r}\frac{L| \lag\bb_{i,l}, \bh_i \rag |^2}{\|\bh_i\|^2}, \quad \max_{1\leq l\leq L, 1\leq i\leq r}\frac{| \lag\ba_{i,l}, \bx_i \rag |^2}{\|\bx_i\|^2}
\end{equation}
are required to be small. We can ensure that the second term in~\eqref{mu_explained} is small since each $\ba_{i,l}$ is a Gaussian random vector and randomness contributes a lot to making the quantity small (with high probability). However, the first term is deterministic and could in principle be very large for certain $\bh_i$ (more precisely, the worst case could be ${\cal O}(K)$), hence we need to put a constraint on $\mu^2_h$ in order to control its size.   As numerical simulations presented in Section~\ref{s:numerics} show,  the relevance of $\mu^2_h$ goes beyond ``proof-technical reasons''. The required number of measurements for successful recovery does indeed depend on $\mu^2_h$, see Figure~\ref{fig:Lvsmuh}, at least when using the suggested approach via semidefinite programing.

\subsection{Conditions on the matrices $\BA_i$}
Throughout the proof of main theorem, we also need to be able to control a certain ``mutual incoherence'' of the matrices $\A_i$ on the subspaces $T_i$, cf.~\eqref{def:mu}.
This condition involves the quantity
\begin{equation}
\max_{j\neq k}\|\PP_{T_j}\A_j^*\A_k\PP_{T_k}\|.
\end{equation}
This quantity is formulated in terms of the matrices $\A_i$ (and not the $\BA_i$), but in order to get a grip
on this quantity, it will be convenient and necessary to impose some conditions on the matrices $\BA_i$. For instance we may assume that the $\BA_i$ are i.i.d.~Gaussian random matrices, which we will do henceforth. Thus, we  require that the $l$-th column of $\BA_{i}^*$, $\ba_{i,l} \sim\mathcal{N}(0,\I_{N_i})$, i.e., $\ba_{i,l}$ is an
$N_i\times 1$ standard real Gaussian random vector. In that case the expectation of $\A_i^*\A_i(\BZ_i) = \sum_{l=1}^L
\bb_{i,l}\bb_{i,l}^* \BZ_i \ba_{i,l}\ba_{i,l}^*$ can be computed
\begin{equation*}
\E(\A_i^*\A_i(\BZ_i)) = \sum_{l=1}^L \bb_{i,l}\bb_{i,l}^*\BZ_i \E(\ba_{i,l}\ba_{i,l}^*) = \BZ_i
\end{equation*}
for all $\BZ_i \in\CC^{K_i\times N_i},$
which says that the expectation of $\A_i^*\A_i$ is the identity.
In the proof, we also need to examine $\A^*_{i,p}\A_{i,p}$ to construct the so-called~\emph{dual certificate} via golfing scheme. Considering the common partition $\{\Gamma_p\}_{p=1}^P$ satisfying~\eqref{cond:iso}, we define $\A_{i,p}:
\CC^{K_i\times N_i}\rightarrow \CC^Q$ and $\A_{i,p}^*: \CC^Q\rightarrow \CC^{K_i\times N_i}$  correspondingly by
\begin{equation}\label{def:AP}
\A_{i,p}(\BZ_i) = \{\bb_{i,l}^* \BZ_i \ba_{i,l}\}_{l\in\Gamma_p},
\quad \A_{i,p}^*(\bz) = \sum_{l\in\Gamma_p} z_l \bb_{i,l}\ba_{i,l}^*.
\end{equation}
The definition of $\A_{i,p}$ is the same as that of $\A_i$ except that $\A_{i,p}$ only uses a subset of all measurements.
However, the expectation of $\A_{i,p}^*\A_{i,p}$ is no longer the identity in general (except the case when all $\BB_i$ are low-frequency DFT matrices and satisfy~\eqref{eq:subfourier}), i.e.,
\begin{equation*}
\A_{i,p}^*\A_{i,p}(\BZ_i) = \sum_{l\in\Gamma_p} \bb_{i,l}\bb_{i,l}^*\BZ_i \ba_{i,l}\ba_{i,l}^*,
\end{equation*}
and
\begin{equation}\label{def:T}
\E(\A_{i,p}^*\A_{i,p}(\BZ_i)) = \BT_{i,p}\BZ_i, \quad \BT_{i,p} := \sum_{l\in\Gamma_p}\bb_{i,l}\bb_{i,l}^*.
\end{equation}
The non-identity expectation of $\A_{i,p}^*\A_{i,p}$ will create challenges throughout the proof. However, there is
an easy trick to fix this issue. By properly assuming $Q > K_i$, $\BT_{i,p}$ is actually invertible. Consider $\A_{i,p}^*\A_{i,p} (\BS_{i,p}\BZ_i)$ and its expectation now yields
\begin{equation}\label{def:S}
\E(\A_{i,p}^*\A_{i,p} (\BS_{i,p}\BZ_i)) = \BT_{i,p}\BS_{i,p} \BZ_i = \BZ_i, \quad \BS_{i,p} := \BT_{i,p}^{-1}.
\end{equation}
This trick, i.e., making the expectation of $\A_{i,p}^*\A_{i,p}\BS_{i,p}$ equal to the identity, plays an important
role in the proof.

\section{Main Results\label{s:maintheorem}}

\subsection{The noiseless case}

Our main finding is that solving~\eqref{cvxprog} enables demixing and blind deconvolution simultaneously. Moreover, our method is also robust to noise.
\begin{theorem}{\bf [Main Theorem]}\label{thm:main}
Consider the model in~\eqref{eq:model} and assume that each $\BB_i\in\CC^{L\times K_i}$ with $\BB_i^*\BB_i = \I_{K_i}$ and each $\BA_i$ is a Gaussian random matrix, i.e., each entry in $\BA_i\stackrel{\text{i.i.d}}{\sim} \mathcal{N}(0,1)$. Let $\mum^2$ and $\mu^2_h$ be as defined in~\eqref{def:mumax} and~\eqref{def:muh} respectively, and denote $K := \max_{1\leq i\leq r}K_i$ and $N := \max_{1\leq i\leq r}N_i$. If
\begin{equation*}
L \geq C_{\alpha+\log(r)} r^2 \max\{\mum^2K, \mu^2_h N\}\log^2 L\log\gamma,
\end{equation*}
where $\gamma$ is the upper bound of $\|\A_i\|$ and obeys $\gamma \leq \sqrt{N\log(NL/2) + \alpha \log L}$,
then solving~\eqref{cvxprog} recovers $\{\BX_i = \bh_i\bx_i^*, 1\leq i\leq r\}$ exactly
with probability at least $1 - {\cal O}(L^{-\alpha + 1})$.
\end{theorem}

Even though the proof of Theorem~\ref{thm:main} follows a meanwhile well established route, the details of the proof itself are nevertheless quite involved and technical. Hence, for convenience we give a brief overview of the proof architecture. In Section~\ref{s:sufficient} we derive a sufficient condition and an approximate dual certificate condition for the minimizer of~\eqref{cvxprog} to be the unique solution to~\eqref{eq:model}. These conditions stipulate that the  matrices $\A_i$ need to satisfy two key properties. The first property, proved in Section~\ref{s:local}, can be considered as a modification of the celebrated \emph{Restricted Isometry Property (RIP)}~\cite{candes2006}, as it requires the $\A_i$ to act in a certain sense as ``local'' approximate isometries~\cite{candes_plan2011,CSV11}. The second property, proved in Section~\ref{s:incoherence}, requires the two operators $\A_i$ and $\A_j$ to satisfy a \emph{``local'' mutual
incoherence property}. With these two key properties in place, we can now construct an approximate dual certificate that fulfills the conditions derived in 
Section~\ref{s:sufficient}. We use the golfing scheme~\cite{gross11recovering} for this purpose, the constructing of which can be found in Section~\ref{s:dual}. 
With all these tools in place, we assemble the proof of Theorem~\ref{thm:main} in Section~\ref{mainproof}.

The theorem assumes for convenience that the $\bh_i$ and the $\bx_i$ are real-valued, but it is easy to modify the proof for complex-valued $\bh_i$ and  $\bx_i$. We leave this modification to the reader.
While Theorem~\ref{thm:informal} is the first of its kind, the derived condition on the number of measurements in~\eqref{convmodel1} is not optimal. Numerical experiments suggest (see e.g.\ Figure~\ref{fig:Lr} in Section~\ref{s:numerics}), that the number of measurements required for a successful solution of the blind deconvolution-blind demixing problem  scales with $r$ and not with $r^2$. Indeed, the simulations indicate that successful recovery via semidefinite programming is possible with a number of measurements close to the theoretical minimum, i.e., with  $L \gtrsim r (K+N)$, see Section~\ref{s:numerics}. This is a good news from a viewpoint of application and means that there is room for improvement in our theory. Nevertheless, this brings up the question whether we can improve upon our bound. A closer inspection of the proof shows that the $r^2$-bottleneck comes from the 
requirement $ \max_{j\neq k}\|\PP_{T_j}\A_j^*\A_k\PP_{T_k}\| \le \frac{1}{4r}$, see condition~\eqref{cond:suffcond}. In order to achieve this we need that $L$, the number of measurements, scales essentially like $r^2 \max\{\mum^2K, \mu^2_h N\}$ (up to $\log$-factors),  see Section~\ref{s:incoherence}.  Is it possible, perhaps with a different condition that does not rely on mutual incoherence between the $\A_j$, to reduce this requirement
on $L$ to one that scales like $ r \max\{\mum^2K, \mu^2_h N\}$?

Now we take a little detour to revisit the blind deconvolution problem described in the introduction and in the informal Theorem~\ref{thm:informal}, 
which is actually contained in our proposed framework as a special case.  Recall the model in~\eqref{eq:model} that $\by$ is actually the sum of Hadamard products of $\BB_i \bh_i$ and $\BA_i \bx_i$. 
Let $\BF$ be the Discrete Fourier Transform matrix of size $L \times L$ with $\BF^*\BF = \I_L$ and let the $L \times K_i$ matrix $\BB_i$ consist of the first $K_i$ columns of $\BF$ (then $\BB_i^{\ast} \BB_i = \BI_{K_i}$).
Now we can express~\eqref{eq:model}  equivalently as the sum of circular convolutions of  $ \BF^{-1} ( \BB_i \bh_i)$ and $\BF^{-1} (\BA_i \bx_i)$, i.e.,
\begin{align}
\label{diag2conv}
\begin{split}
\frac{1}{\sqrt{L}}\BF^{-1} {\by} & = \sum_{i=1}^r \BF^{-1} (\BB_i \bh_i) \ast \BF^{-1} (\BA_i \bx_i ) = \sum_{i=1}^r ( \BF^{-1}\BB_i)  \bh_i \ast (\BF^{-1} \BA_i) \bx_i.  
\end{split}
\end{align}
Set
\begin{equation*}
\bff_i: = 
\begin{bmatrix}
\bh_i \\ {\bf 0}_{L-K_i}
\end{bmatrix}.
\end{equation*}
Then there holds 
\begin{equation*}
\BF^{-1} \BB_i \bh_i = \BF^{-1}
\begin{bmatrix}
\BB_i & {\bf 0}_{L\times (L-K_i)} 
\end{bmatrix}
\begin{bmatrix}
\bh_i \\ {\bf 0}_{L-K_i}
\end{bmatrix}  = \bff_i.
\end{equation*}
Hence with a slight abuse of notation (replacing $\frac{1}{\sqrt{L}}\BF^{-1} {\by}$ in~\eqref{diag2conv} by $\by$ and  $\BF^{-1} \BA_i$ by $\BA_i$, 
using the fact that the Fourier transform of a Gaussian random matrix is again a Gaussian random matrix), we can express~\eqref{eq:model}
equivalently as
\begin{equation*}
\by = \sum_{i=1}^r \bff_i \ast (\BA_i \bx_i),
\end{equation*}
which is exactly~\eqref{convmodel} up to a normalization factor.

Thus we can easily derive the following corollary from Theorem~\ref{thm:main} (using the fact that  $\mu_{\max} = 1$ for the particular
choice of $\BB_i$  above). This corollary is the precise version of the informal Theorem~\ref{thm:informal}.

\begin{corollary}\label{cor:bd}
Consider the model in~\eqref{eq:model}, i.e.,
\begin{equation*}
\by = \sum_{i=1}^{r}    \bff_i \ast (\BA_i \bx_i),
\end{equation*}
where we assume that $\bff_i(k) = 0$ for $k > K_i$. Suppose that each $\BA_i$ is a Gaussian random matrix, i.e., each entry in $\BA_i\stackrel{\text{i.i.d}}{\sim} \mathcal{N}(0,1)$. Let 
$\mu^2_h$ be as defined in~\eqref{def:muh}. If
\begin{equation*}
L \geq C_{\alpha+ \log(r)} r^2\max\{K, \mu^2_h N\}\log^2 L\log\gamma,
\end{equation*}
where $\max_{i}\|\A_i\| \leq \gamma \leq \sqrt{N\log(NL/2) + \alpha \log L},$
then solving~\eqref{cvxprog} recovers $\{\BX_i: = \bh_i\bx_i^*, 1\leq i\leq r\}$ exactly 
with probability at least $1 - {\cal O}(L^{-\alpha + 1})$.
\end{corollary}

For the special case $r=1$, Corollary~\ref{cor:bd} becomes Theorem 1 in~\cite{RR12} with the proviso that in principle our $\mu_h^2$ is defined slightly differently than in~\cite{RR12}, see Remark~\ref{remark_lemma}. Yet, if we choose the partition of the matrix $\BB$ as suggested in the third example in Subsection~\ref{ss:incoherence}, then the difference between the two definitions of $\mu_h^2$ vanishes.)

\subsection{Noisy data}

In reality, measurements are always noisy. Hence, suppose  $\hby = \by + \beps$ where $\beps$ is noise with  $\|\beps\| \le \eta$. In this case we solve the following optimization program to recover $\{\BX_i\}_{i=1}^r,$
\begin{equation}\label{cvx:noise}
\min \sum_{i=1}^r\|\BZ_i\|_* \quad \text{subject to} \quad \left\| \sum_{i=1}^r\A_i(\BZ_i) - \hby \right\| \leq \eta.
\end{equation}
We should choose $\eta$ properly in order to make $\BX_{i}$ inside the feasible set and $\|\hby\| > \eta.$
Let $\{\hBX_i\}_{i=1}^r$ be the minimizer to~\eqref{cvx:noise}. We immediately know 
\begin{equation}\label{inq:noiseobj}
\sum_{i=1}^r\|\hBX_i\|_*  \leq \sum_{i=1}^r \|\BX_{i}\|_*.
\end{equation}
Our goal is to see how $\sqrt{\sum_{i=1}^r \| \hBX_i - \BX_{i}\|^2_F}$ varies with respect to the noise level $\eta$.

\begin{theorem}{\bf [Stability Theory]}\label{thm:noise}
Assume we observe $\hat{\by} = \by + \beps = \sum_{i=1}^r \A_i(\BX_{i}) + \beps$ with $\|\beps\| \leq \eta$. 
Then, under the same conditions as in Theorem~\ref{thm:main}, the minimizer $\{\hat{\BX}_i\}_{i=1}^r$  to~\eqref{cvx:noise} satisfies
\begin{equation*}
\sqrt{ \sum_{i=1}^r \|\hat{\BX}_i - \BX_{i}\|_F^2 } \leq  C  \frac{\lambda_{\max}}{\lambda_{\min}} \, r  \sqrt{\max\{K, N\}} \eta
\end{equation*}
with probability at least $1 - {\cal O}(L^{-\alpha + 1})$. Here, $\lambda^2_{\max}$ and $\lambda^2_{\min}$ are the largest and the smallest eigenvalue of $\sum_{i=1}^r \A_i\A_i^*$, respectively. 
\end{theorem}
Note that with a little modification of Lemma 2 in~\cite{RR12}, it can be shown that $\frac{\lambda_{\max}}{\lambda_{\min}} \sim \frac{\mum}{\mumin}.$ The proof of Theorem~\ref{thm:noise} will be given in Section~\ref{s:stability}.

With Theorem~\ref{thm:noise} and Wedin's $\sin(\theta)$ theorem~\cite{Wedin72,Stewart90} for singular value/vector perturbation theory, we immediately have the performance guarantees of recovering individual $(\bh_i, \bx_i)_{i=1}^r$ by applying SVD to $\hBX_i.$ 

\begin{corollary}\label{colvec}
Let $\hat{\bh}_i = \sqrt{\hat{\sigma}_{i1}} \hat{\bu}_{i1}$ and $\hat{\bx}_i = \sqrt{\hat{\sigma}_{i1}} \hat{\bv}_{i1}$ where $\sigma_{i1}$, $\hat{\bu}_{i1}$ and $\hat{\bv}_{i1}$ are the leading singular value, left and right singular vectors of $\hBX_i$ respectively. Then there exist  $\{c_i\}_{i=1}^r$ and a constant $c_0$ such that
\begin{align*}
\| \bh_i - c_i \hat{\bh}_i\| & \leq c_0\min(\eps/\|\bh_i\|, \|\bh_i\|), \\
\| \bx_i - c_i^{-1} \hat{\bx}_i\| & \leq c_0\min(\eps/\|\bx_i\|, \|\bx_i\|)
\end{align*} 
where $\eps = \sqrt{\sum_{i=1}^r\|\hat{\BX}_i - \BX_i\|_F^2}.$
\end{corollary}

\section{Numerical Simulations \label{s:numerics}}

In this section we present a range of simulations that illustrate and complement the theoretical results of the previous section.

\subsection{Number of measurements vs.\ number of sources \label{s:minL}}

We investigate empirically the minimal $L$ required to simultaneously demix and deconvolve $r$ sources. Here are the parameters and settings used in the simulations: the number of sources $r$ varies from $1$ to $7$ and $L = 50,100,\cdots,750$ and $800$.  For each $1\leq i\leq r$, $K_i = 30$ and $N_i = 25$ are fixed. Each $\BB_i$ is the first $K_i$ columns of an $L\times L$ DFT matrices with $\BB_i^*\BB_i = \I_{K_i}$ and each $\BA_i$ is an $L\times N_i$ Gaussian random matrix. $\bh_i$ and $\bx_i$ satisfy $\mathcal{N}(\bzero, \I_{K_i})$ and $\mathcal{N}(\bzero, \I_{N_i})$ respectively. We denote $\BX_{i} = \bh_i\bx_i^*$, the ``lifted" matrix and solve~\eqref{cvxprog} to recover $\BX_{i}.$ For each pair of $(L,r)$, 10 experiments are performed and the recovery is regarded as  a success if
\begin{equation}\label{eq:success}
\frac{\sqrt{\sum_{i=1}^r \|\hat{\BX}_{i} - \BX_{i}\|_F^2}}{\sqrt{\sum_{i=1}^r \|\BX_{i}\|_F^2}} < 10^{-3}
\end{equation}
where each $\hat{\BX}_i$, given by solving~\eqref{cvxprog} via the CVX package~\cite{cvx}, serves as an  approximation of $\BX_{i}$. 
Theorem~\ref{thm:main} implies that the minimal required $L$ scales with $r^2$, which is not optimal in terms of number of degrees of freedom. Figure~\ref{fig:Lr} validates the non-optimality of our theory. Figure~\ref{fig:Lr} shows a sharp phase transition boundary between success and failure and furthermore the minimal $L$ for exact recovery seems to have a strongly linear correlation with number of sources $r$. Note that if $L$ is approximately greater than $80r$, solving~\eqref{cvxprog} gives the exact recovery of $\BX_{i}$ numerically, which is quite close to the theoretical limit $(K_i + N_i)r = 55r.$ 

\bigskip

Moreover, our method extends to other types of settings  although we do not have theories for them yet. In wireless communication, it is particularly interesting to see the recovery performance if $\BA_i = \BD_i \BH_i$ where $\BD_i$ is a diagonal matrix with Bernoulli random variables (taking value $\pm 1$ with equal probabilities) on the diagonal and $\BH_i$ is fixed as the first $N_i$ columns of a non-random Hadamard matrix. In other words, the only randomness of $\BA_i$ comes from $\BD_i$. Both $\BH_i$ and $\BD_i$ are matrices of $\pm 1$ entries and can be easily generated in many applications. 
By using the same settings on $L$, $r$, $\bh_i$ and $\bx_i$ as before and $K_i = N_i = 15$, we apply~\eqref{cvxprog} to recover $(\bh_i, \bx_i)_{i=1}^r$. Since the existence of Hadamard matrices of order $4k$ with positive integer $k$ is still an open problem~\cite{HW78HM}, we only test $L = 2^{s}$ with $s = 6, 7,8$ and $9$. Surprisingly,  Figure~\ref{fig:Lr} (the bottom one) also demonstrates that the minimal $L$ scales linearly with $r$ and our algorithm almost reaches the information theoretic optimum even if all $\BA_i$ are partial Hadamard matrices. 

\begin{figure}[h!]
\centering
\includegraphics[width = 80mm]{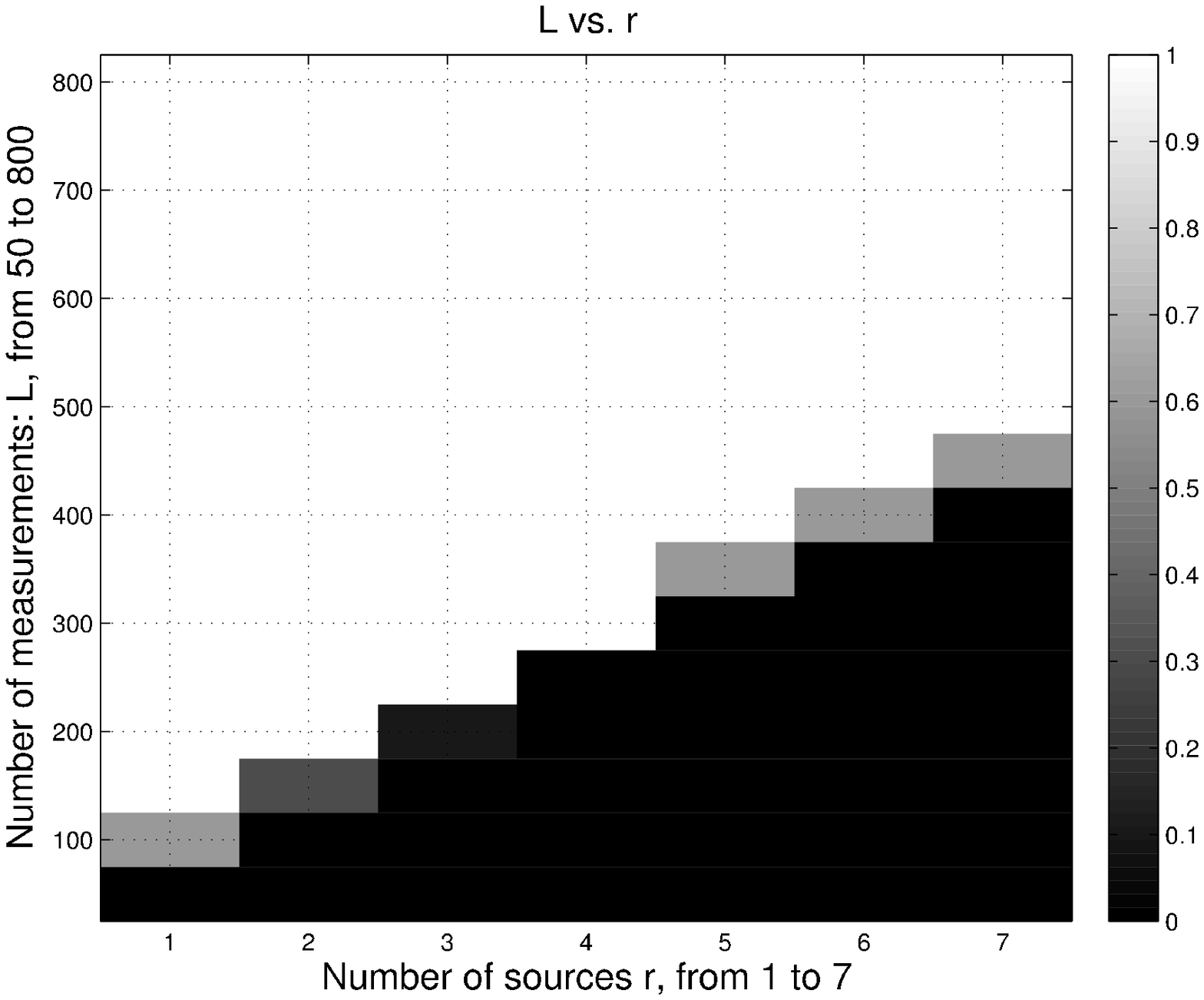}
\vfill
\vskip0.5cm
\centering
\includegraphics[width = 80mm]{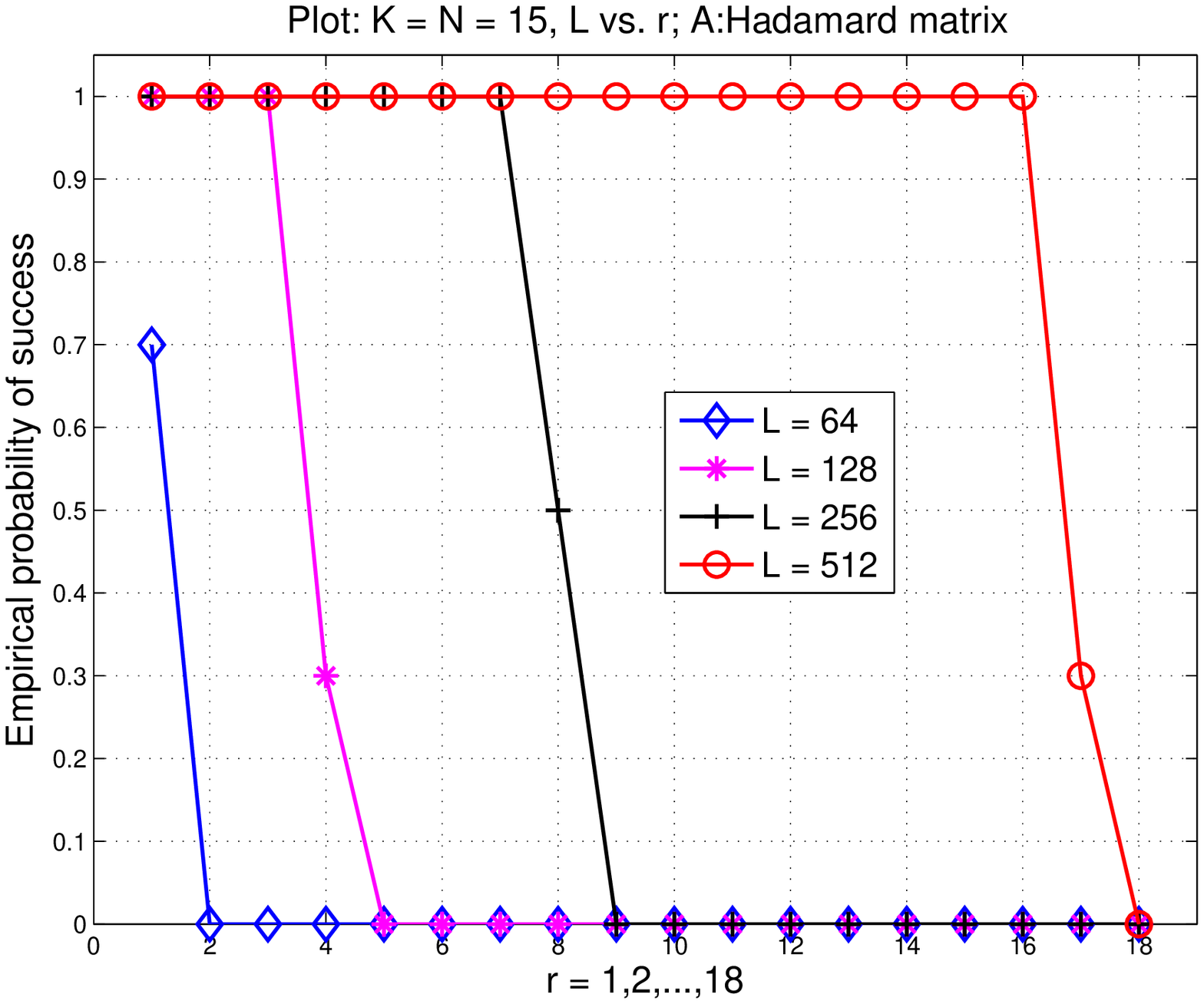}
\caption{Phase transition plot: performance of~\eqref{cvxprog} for different pairs of $(L,r)$. White: 100$\%$ success and black: $0\%$ success. Top: $\BA_i:$ $L\times N_i$ Gaussian random matrices. $K_i = 30$ and $N_i = 25$.  $1\leq r\leq 7$ and $L = 50, 100,\cdots, 800$; Bottom: $\BA_i = \BD_i \BH_i$ where $\BH_i$ is the first $N_i$ columns of an $L\times L$ Hadamard matrix and $\BD_i$ is a diagonal matrix with i.i.d. random entries taking $\pm 1$ with equal probability.  $K_i = N_i = 15$ with $r = 1,\cdots, 18$ and $L = 64, 128, 256, 512$. }
\label{fig:Lr}
\end{figure}


\bigskip
Figure~\ref{fixedL} shows the performance of recovery via solving~\eqref{cvxprog} under the assumption that $L$ is fixed and $K_i$ and $N_i$ are changing. The results are presented for two cases: (i) the $\BA_i$ are Gaussian random matrices, and (ii) the $\BA_i$ are Hadamard matrices premultiplied by a binary diagonal matrix as explained above.
In the simulations, we assume there exist two sources $(r = 2)$ with $K_1 = K_2$ and $N_1 = N_2$. We fix $L = 128$ and let $K_i$ and $N_i$ vary from 5 to 50. $\BB_i$ consists of the first $K_i$ columns of an $L\times L$ DFT matrix.
Both $\bh_i$ and $\bx_i$ are random Gaussian vectors. 
The boundary between success and failure in the phase transition plot is well approximated by a line, which matches the relationship between $L, K_i$, and $N_i$ stated in Theorem~\ref{thm:main}. More precisely, the probability of success is quite satisfactory if $L = 128 \geq 1.5 r(K_i + N_i)$ in this case.

\begin{figure}[h!]
\begin{minipage}{0.48\textwidth}
\centering
\includegraphics[width = 75mm]{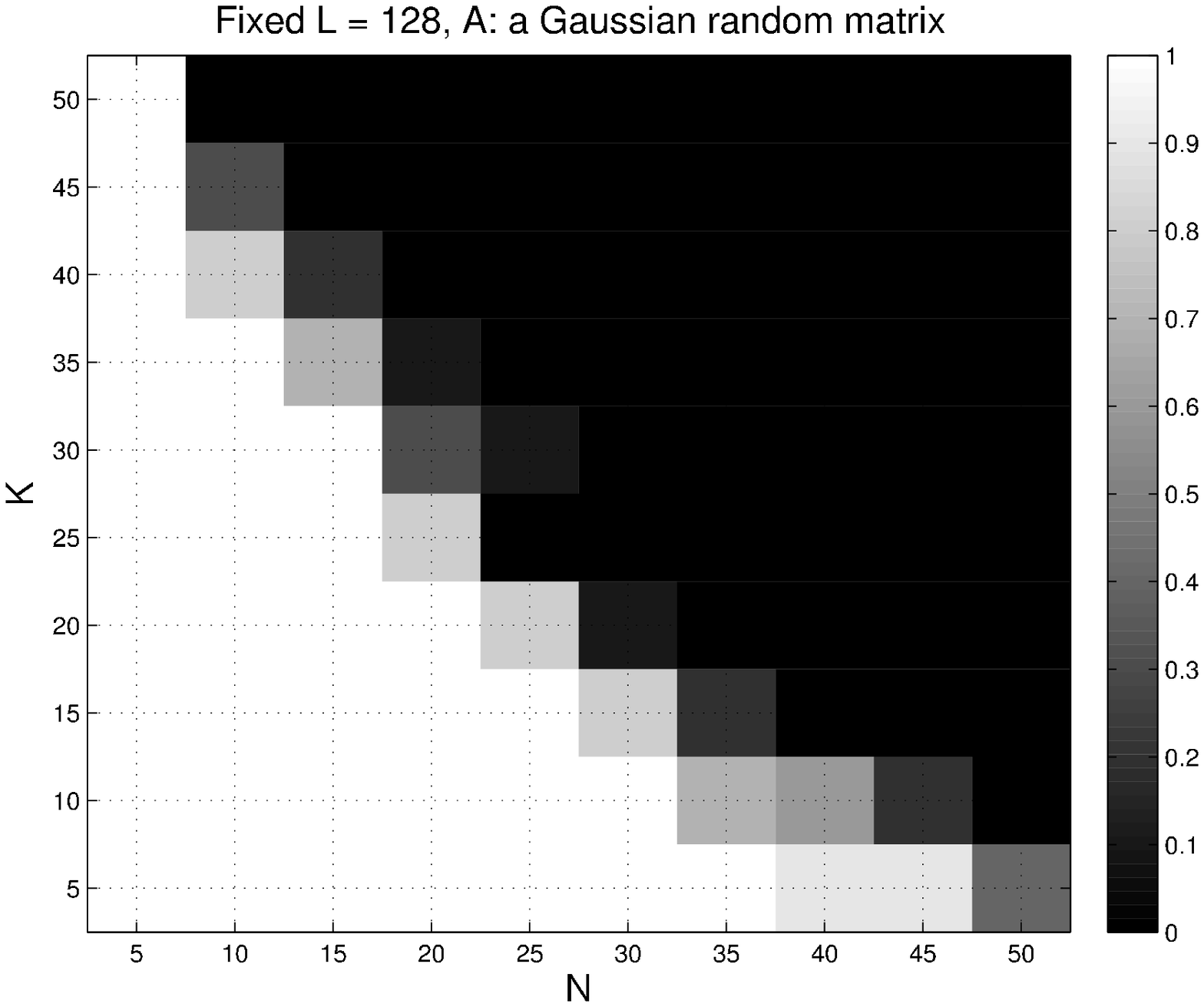}
\end{minipage}
\hfill
\begin{minipage}{0.48\textwidth}
\centering
\includegraphics[width = 75mm]{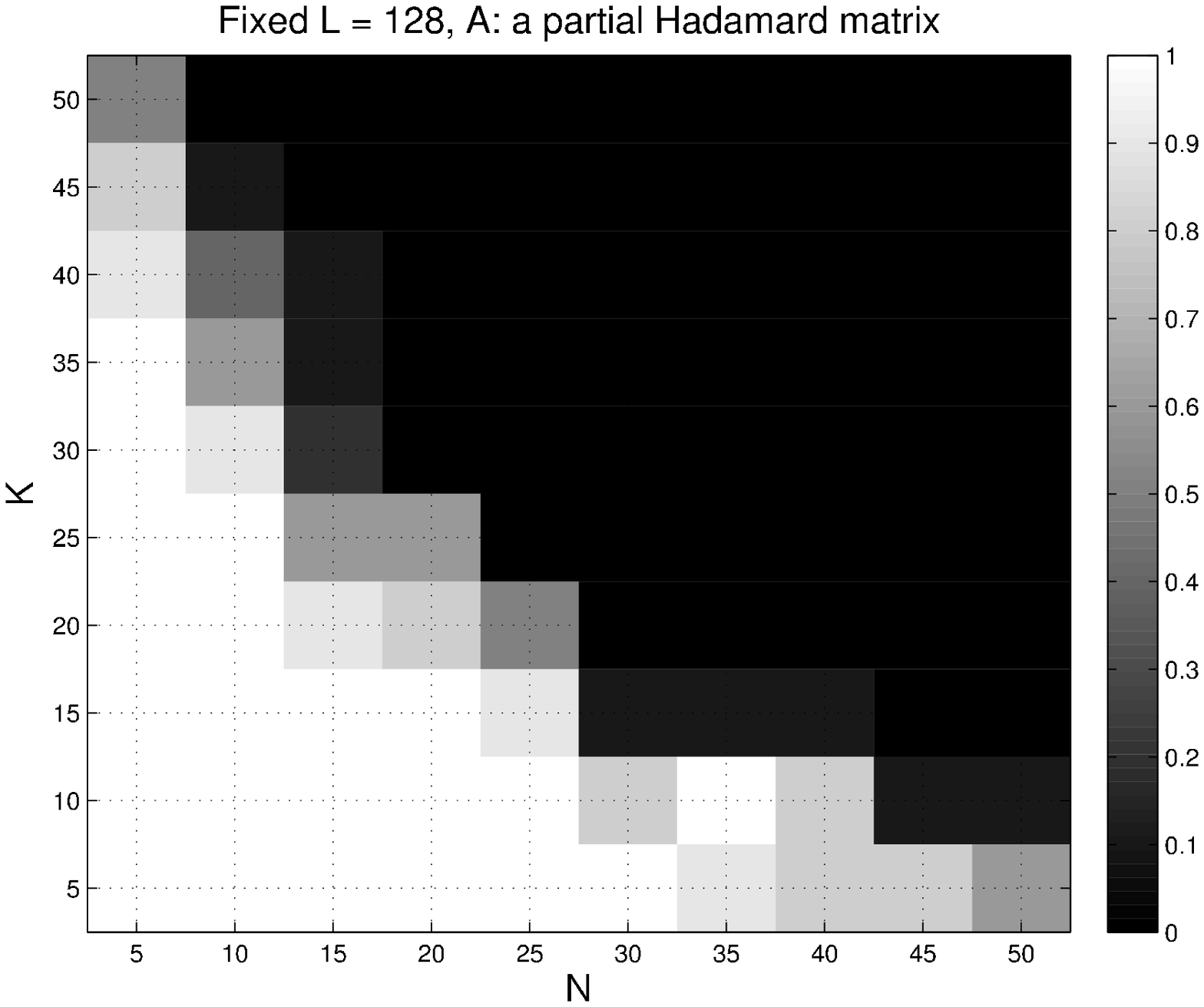}
\end{minipage}
\caption{Phase transition plot: empirical probability of recovery success for $(K_i,N_i)$ where $K_i$ and $N_i$ both vary from 5 to 50 and $L = 128$ is fixed. White: 100$\%$ success and black: $0\%$ success. Left: each $\BA_i$ is a $L\times N_i$ Gaussian random matrix; Right: $\BA_i = \BD_i\BH_i$ with $\BH_i$ being the first $N_i$ columns of the $L\times L$ Hadamard matrix and $\BD_i$ a diagonal matrix with entries taking value on $\pm 1$ with equal probabilities. }
\label{fixedL}
\end{figure}

\subsection{Number of measurements vs.\  incoherence parameter}
Theorem~\ref{thm:main} indicates that $L$ scales with $\mu^2_h$ defined in~\eqref{def:muh} and $\mu^2_h$ also plays an important role in the proof. Moreover, Figure~\ref{fig:Lvsmuh} implies that $\mu_h^2$ is not only necessary for ``technical reasons" but also related to the numerical performance. In the experiment, we fix $r=1$ and $K = N = 30$. $\BA$ is a Gaussian random matrix, and $\BB$ is a low-frequency Fourier matrix, while $L$ and $\mu^2_h$ vary. 
Thanks to the properties of low-frequency Fourier matrices, we are able to construct a vector $\bh$ whose associated incoherence parameter $\mu^2_h$ in~\eqref{def:muh} is equal to a particular number. 
 In particular, we choose $\bh$ to be one of those vectors whose first $3,6,\cdots, 27, 30$ entries are $1$ and the others are zero. The advantage of those choices is that $\max_{1\leq l\leq L} L|\lag  \bb_l, \bh\rag |^2/\|\bh\|^2$ will not change with $L$ and can be computed explicitly.  
We can see in Figure~\ref{fig:Lvsmuh} that the minimal $L$ required for exact recovery seems strongly linearly associated with $\mu^2_h = L \max |\lag \bb_l, \bh\rag|^2/\|\bh\|^2$. 

\begin{figure}[h!]
\centering
\includegraphics[width = 75mm]{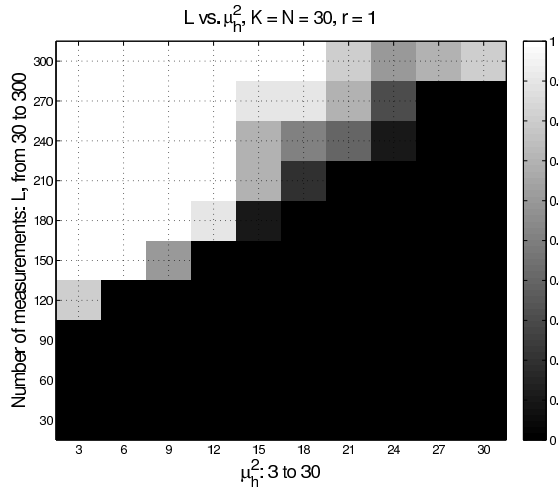}
\caption{Phase transition plot: Empirical probability of recovery success for $(L, \max L|\lag \bb_l, \bh\rag |^2/\|\bh\|^2)$ where $r = 1$, $K = N = 30$. White: 100$\%$ success and black: $0\%$ success.  }
\label{fig:Lvsmuh}
\end{figure}

\subsection{Robustness}

In order to illustrate the robustness of our algorithm with respect to noise as stated in Theorem~\ref{thm:noise}, we conduct two simulations to study how the relative error $\frac{\sqrt{\sum_{i=1}^r \|\hat{\BX}_{i} - \BX_{i}\|_F^2}}{\sqrt{\sum_{i=1}^r \|\BX_{i}\|_F^2}}$ behaves under different levels of noise. In the first experiment we  choose $r = 3$, i.e., there are totally 3 sources. They are  of different sizes, i.e., $(K_1, N_1) = (20, 20)$, $(K_2, N_2) = (25, 25)$ and $(K_3, N_3) = (20, 20).$
$L$ is fixed to be 256, the $\BB_i$ are as outlined in Section~\ref{s:minL} and the $\BA_i$  are Gaussian random matrices. In the simulation, we choose $\beps_i$ to be a normalized Gaussian  
random vector. Namely, we first sample $\beps_i$ from a multivariate Gaussian distribution and then normalize $\| \beps_i\|_F = \sigma \sqrt{\sum_{i=1}^r \|\BX_{i}\|_F^2}$ where $\sigma = 1, 0.5, 0.1, 0.05, 0.01, \cdots$ and $0.0001$. For each $\sigma,$ we run 10 experiments and compute the average relative error in the scale of dB, i.e., $10\log_{10}(\text{Avg.RelErr})$. 

\bigskip

We run a similar experiment, this time with $r=15$ sources (all  $N_i$ are equal to 10, and all $K_i$ are equal to 15) and the $\BA_i$ are the ``random'' Hadamard matrices described above.
For both experiments, Figure~\ref{noise} indicates that the average relative error (dB) is linearly correlated with $\SNR = 10\log_{10}(\sum_{i=1}^r \|\BX_{i}\|_F^2/ \|\beps\|^2_F)$, as one would wish.
\begin{figure}[h!]
\begin{minipage}{0.48\textwidth}
\centering
\includegraphics[width = 80mm]{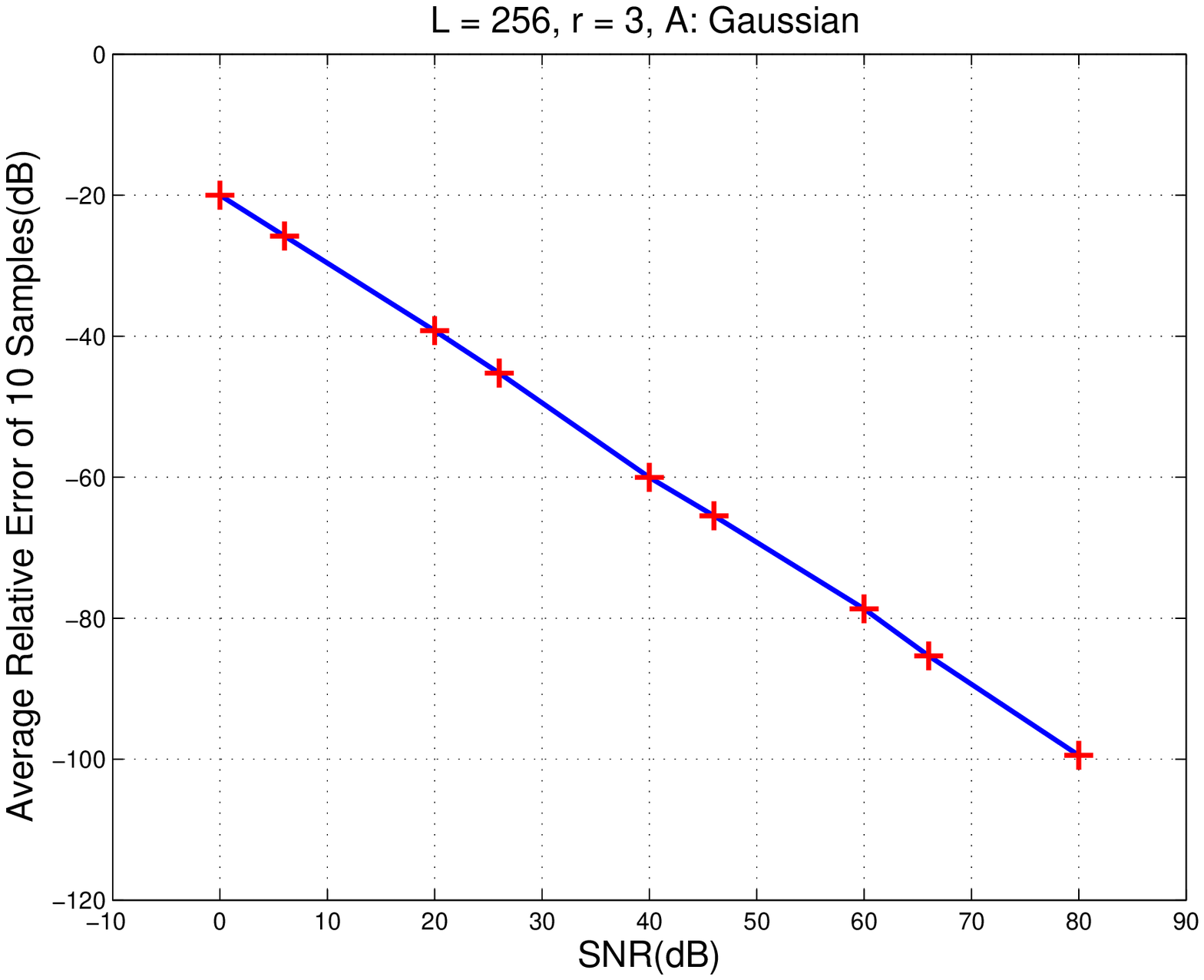}
\end{minipage}
\hfill
\begin{minipage}{0.48\textwidth}
\centering
\includegraphics[width = 80mm]{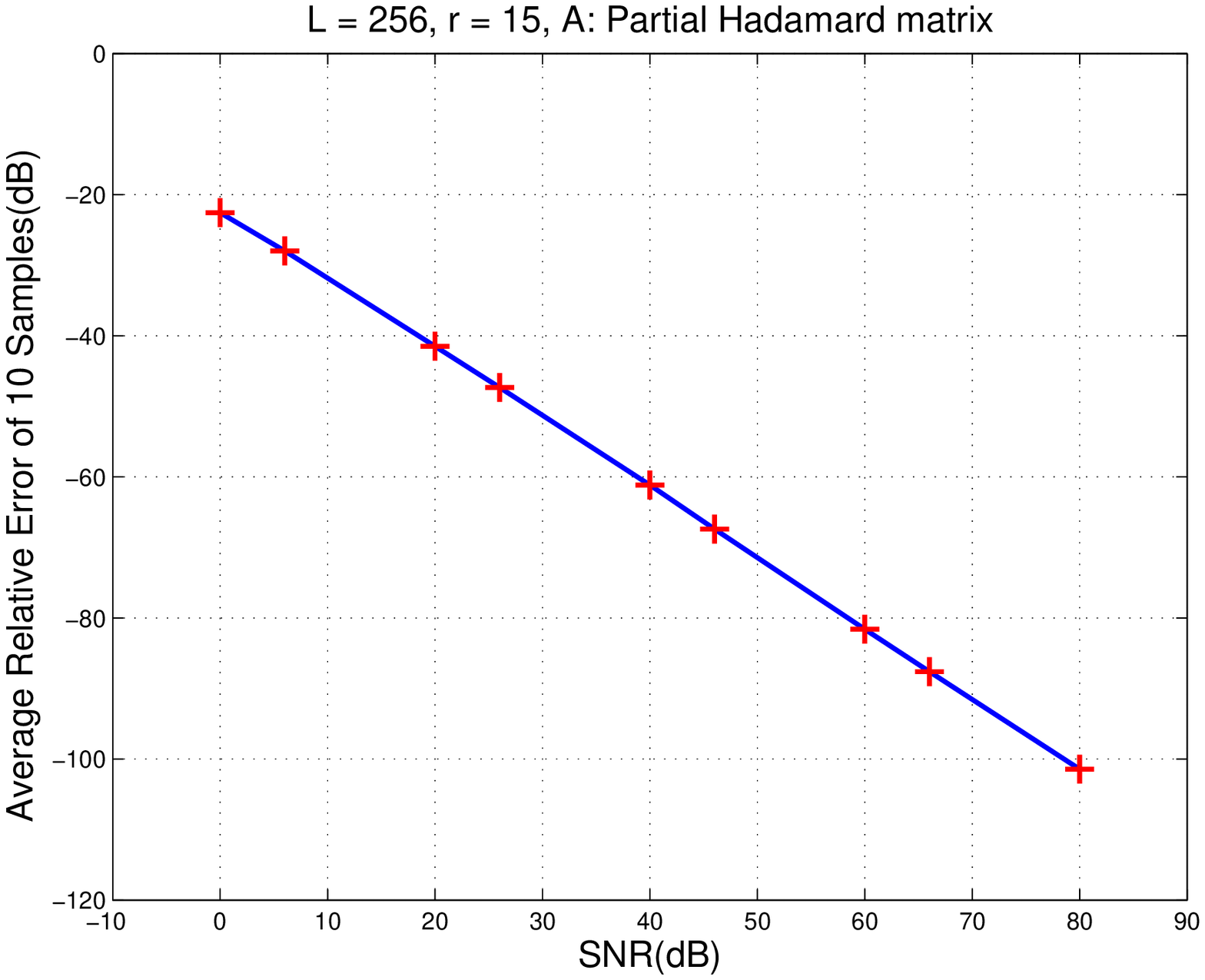}
\end{minipage}
\caption{Performance of~\eqref{cvx:noise} under different SNR. Left: $\{\BA_i\}$ are Gaussian and there are 3 sources and $L = 256$; Right: $\BA_i = \BD_i\BH_i$ where $\BH_i$ is a partial Hadamard matrix and $\BD_i$ is a diagonal matrix with random $\pm 1$ entries. Here there are $15$ sources in total and $L = 512$.}
\label{noise}
\end{figure}

\section{Proofs}\label{s:proofs}

This section is dedicated to proving the theorems presented in Section~\ref{s:maintheorem}. Since the proofs are rather involved and technical, we have arranged the arguments into individual subsections. We first start out with stating sufficient conditions under which the main theorems will hold, see
Subsection~\ref{s:sufficient}. In Subsection~\ref{s:local} we state and analyze a certain form of \emph{local restricted isometry property} and a specific \emph{local incoherence property} is established in
Subsection~\ref{s:incoherence}. Both of them are
associated with the assumptions of the sufficient conditions in Subsection~\ref{s:sufficient}. We construct a dual certificate in Subsection~\ref{s:dual}. With all these results in place, Subsections~\ref{mainproof} and
\ref{s:stability} are devoted to assembling the proofs of the main theorems.

\subsection{Sufficient conditions\label{s:sufficient}}

Without loss of generality, we assume that the lifted matrix $\BX_i = \alpha_i\bh_i\bx_i^*$, where  $\bh_i$ and $\bx_i$ are all real and of unit norm and $\alpha_i$ is a  scalar  for all $1\leq i\leq r$ throughout Section~\ref{s:sufficient}--\ref{s:stability}.  
We also define a linear space which $\bh_i\bx_i^*$ lies in and which will be useful in the further analysis:
\begin{equation}\label{cond:proj}
T_i = \{ \bh_i\bh_i^*\BZ_i + (\I_{K_i} - \bh_i\bh_i^*)\BZ_i\bx_i\bx_i^* | \BZ_i\in\CC^{K_i\times N_i}  \}
\end{equation} 
and its corresponding complement is defined as
\begin{equation}\label{cond:projc}
\TB_i = \{ (\I_{K_i} - \bh_i\bh_i^*)\BZ_i (\I_{N_i} - \bx_i\bx_i^*) | \BZ_i\in\CC^{K_i\times N_i}  \}.
\end{equation} 

Now we present the first sufficient condition, under which $\{\alpha_i\bh_i\bx_i^*\}_{i=1}^r$ is the unique minimizer. However, Lemma~\ref{lem:suff-1st} is not easy to use in reality and therefore, we derive Lemma~\ref{lemma:suffcond}, a more useful condition, from Lemma~\ref{lem:suff-1st}.
\begin{lemma}\label{lem:suff-1st}
Assume that 
\begin{equation}\label{eq:suff-1st}
\sum_{i=1}^r \lag \BH_i, \bh_i\bx_i^* \rag + \|\BH_{i,\TB_i}\|_* > 0.
\end{equation}
for any real $\{\BH_i\}_{i=1}^r$ satisfying $\sum_{i=1}^r \A_i(\BH_i) = 0$ and at least one of $\BH_i$ is nonzero.  Then $\{\alpha_i\bh_i\bx_i^*\}_{i=1}^r$ is the unique minimizer to the convex program~\eqref{cvxprog}. 
\end{lemma}

\begin{proof}
For any feasible element of the convex program~\eqref{cvxprog}, it must have the form
of $\{\alpha_i\bh_i\bx_i^* + \BH_i\}_{i=1}^r.$
It suffices to show that the $\sum_{i=1}^r\|\alpha_i\bh_i\bx_i^* + \BH_i\|_* > \sum_{i=1}^r \|\alpha_i\bh_i\bx_i^*\|_*$ for any nontrivial set of $\{\BH_i\}_{i=1}^r$, i.e., at least one of $\BH_i$ is nonzero. For each $\BH_i$, there exists a $\BV_i\in \TB_i$ such that
\begin{equation*}
\lag\BH_i, \BV_i \rag = \lag\BH_{i,\TB_i}, \BV_i  \rag = \|\BH_{i,\TB_i}\|_*
\end{equation*}
where $\BH_{i,\TB_i}$ is the projection of $\BH_i$ on $\TB_i$ and $\| \BV_i \| = 1$. Thus $\bh_i\bx_i^* + \BV_i$ belongs to the \footnote{See~\cite{CR11} and~\cite{FR13} for detailed discussions of the sub-differential and the sub-gradient.}{sub-differential} of $\|\cdot\|_*$ at $\BX_i = \alpha_i\bh_i\bx_i^*.$
\begin{align*}
\sum_{i=1}^r\|\alpha_i\bh_i\bx_i^* + \BH_i\|_* & \geq  \sum_{i=1}^r \left( \| \alpha_i\bh_i\bx_i^*  \|_* + \lag  \bh_i\bx_i^* +  \BV_i, \BH_i\rag \right)\\
 & = \sum_{i=1}^r \left(\|\alpha_i\bh_i\bx_i^*\|_* + \lag \BH_{i, T_i}, \bh_i\bx_i^* \rag  + \|\BH_{i,\TB_i}\|_* \right)\\
& > \sum_{i=1}^r \|\alpha_i\bh_i\bx_i^*\|_*
\end{align*}
where the first inequality follows from the definition of subgradient and~\eqref{eq:suff-1st}. 
\end{proof}
Now we consider under what condition on $\A_i$, the unique minimizer is $\{\alpha_i \bh_i\bx_i^*\}.$ 
Define $\mu$ by
\begin{equation}\label{def:mu}
\mu : = \max_{j\neq k}\|\PP_{T_j}\A_j^*\A_k\PP_{T_k}\| 
\end{equation}
as a measure of \emph{incoherence} between any pairs of linear operators. $\A_i\PP_{T_i}$ is the restriction of $\A_i$ onto $T_i.$

\begin{lemma}{\bf [Key Lemma]}\label{lemma:suffcond}
Assume that
\begin{equation}\label{cond:suffcond}
\|\PP_{T_i}\A_i^*\A_i \PP_{T_i} - \PP_{T_i}\| \leq \frac{1}{4}, \quad \mu \leq \frac{1}{4r}, \quad \|\A_i\| \leq \gamma
\end{equation}
and also there exists a $\blambda\in\CC^L$ such that
\begin{equation}\label{cond:suff2}
\|\bh_i\bx_i^* - \PP_{T_i}(\A_i^*(\blambda))\|_F \leq \alpha, \quad \|\PP_{\TB_i}(\A_i^*(\blambda))\| \leq \beta
\end{equation}
for all $1\leq i\leq r$ and $ (1 - \beta)-2r\gamma \alpha> 0$,
then $\{\alpha_i\bh_i\bx_i^*\}_{i=1}^r$ is the unique minimizer to~\eqref{cvxprog}. In particular, we can choose $\alpha = (5r \gamma )^{-1}$ and $\beta = \frac{1}{2}$. Here $\|\A_i\| := \sup_{\BZ\neq \bzero} \|\A_i(\BZ)\|_F/\|\BZ\|_F.$
\end{lemma}
\begin{remark}
This sufficient condition is quite standard and useful to prove that convex relaxation is able to recover the desired ground truth, see~\cite{Recht11MC,RR12,gross11recovering,CSV11,LS15} for more examples. 

With~\eqref{cond:suffcond}, we will show a variant of null space property $($where the null space refers to $(\{\{\BH_i\}_{i=1}^r: \sum_{i=1}^r \A_i(\BH_i) = \bzero\})$ under Frobenius norm, as shown in~\eqref{nsp}. The key component in~\eqref{cond:suff2} is the existence of $\blambda$. In fact, $\blambda$ is  an approximate dual feasible solution which certifies that $\{ \bh_i\bx_i^*\}_{i=1}^r$ is the unique minimum of the nuclear norm minimization program~\eqref{cvxprog}.
\end{remark}

\begin{proof}
It suffices to show that for any nonzero $\{\BH_i\}_{i=1}^r$ with $\sum_{i=1}^r\A_i(\BH_i) = 0$, there holds
\begin{equation*}
I_1 : = \sum_{i=1}^r \left\lag \BH_i, \bh_i\bx_i^* - \A^*_i (\blambda) \right\rag + \|\BH_{i,\TB_i}\|_* > 0
\end{equation*}
under~\eqref{cond:suffcond} and~\eqref{cond:suff2}.
By decomposing the inner product in $I_1$ on $T_i$ and $\TB_i$ for each $i$, we have
\begin{align*}
 I_1 & =\sum_{i=1}^r \Big(\left\lag \BH_{i, T_i}, \bh_i\bx_i^* - \PP_{T_i}(\A_i^*(\blambda))\right\rag \\
 & \qquad \qquad - \left\lag\BH_{i,\TB_i},  \PP_{\TB_i}(\A_i^*(\blambda))\right\rag + \|\BH_{i,\TB_i}\|_*\Big).
\end{align*}
Then, by applying  the Cauchy-Schwarz inequality and $|\lag \BU, \BV\rag| \leq \|\BU\|_* \|\BV\|$ for any matrices $\BU, \BV$ of the same size, we only need to show that a lower bound of $I_1$ is positive, i.e., $I_1 \geq I_2 > 0$:
\begin{align}\label{cond:a}
\begin{split}
I_2 & : = \sum_{i=1}^r \Big(- \| \BH_{i,T_i} \|_F \|\bh_i\bx_i^* - \PP_{T_i}(\A_i^*(\blambda)) \|_F  + \|\BH_{i,\TB_i}\|_* \left(1- \|\PP_{\TB_i}(\A_i^*(\blambda))\|\right)\Big) > 0.
\end{split}
\end{align}
From now on, we aim to show 
\begin{align*}
\begin{split}
\frac{1}{2r} \left(\sum_{i=1}^r \|\BH_{i,T_i}\|_F\right) & \leq \gamma \left(\sum_{i=1}^r \|\BH_{i,\TB_i}\|_F\right)  \leq \gamma \left(\sum_{i=1}^r \|\BH_{i, \TB_i}\|_*\right)
\end{split}
\end{align*}
in order to achieve~\eqref{cond:a}.
We start with $\sum_{i=1}^r\A_i(\BH_i)  = 0.$ Note that $\BH_{i} = \BH_{i, T_i} + \BH_{i, \TB_i}$, there holds
\begin{equation*}
\left\| \sum_{i=1}^r\A_i (\BH_{i,T_i})  \right\|_F  =  \left\|  \sum_{i=1}^r\A_i (\BH_{i,\TB_i})  \right\|_F.
\end{equation*}
It is easy to bound the quantity on the right hand side by $\|\A_i\| \leq \gamma$,
\begin{equation}\label{upperbd}
\left\|  \sum_{i=1}^r\A_i (\BH_{i,\TB_i})  \right\|_F \leq \gamma \left(\sum_{i=1}^r\|\BH_{i,\TB_i}\|_F\right).
\end{equation}
The difficulty is to establish the lower bound. 
\begin{align}
\left\| \sum_{i=1}^r\A_i(\BH_{i,T_i}) \right\|^2_F  & \quad \geq  \sum_{i=1}^r \| \A_i(\BH_{i,T_i})\|^2 + 2\sum_{j\neq k}\left\lag  \A_j(\BH_{j,T_j}),  \A_k(\BH_{k,T_k})\right\rag \nonumber \\
& \quad \geq  \frac{3}{4} \sum_{i=1}^r \|\BH_{i,T_i}\|_F^2 - 2\mu \sum_{j\neq k} \|\BH_{j,T_j}\|_F \|\BH_{k,T_k}\|_F \label{eq:low-2} \\
& \quad = 
\begin{bmatrix}
\|\BH_{1, T_1}\|_F \\
\vdots \\
\|\BH_{r, T_r}\|_F
\end{bmatrix}^*
\underbrace{\begin{bmatrix}
\frac{3}{4} & \cdots & -\mu \\
\vdots & & \ddots & \vdots \\
-\mu & \cdots & \frac{3}{4}
\end{bmatrix}}_{r\times r \text{ matrix}}
\begin{bmatrix}
\|\BH_{1, T_1}\|_F \\
\vdots \\
\|\BH_{r, T_r}\|_F
\end{bmatrix}, \nonumber 
\end{align}
where~\eqref{eq:low-2}  follows from~\eqref{cond:suffcond}. 
It is easy to see that the coefficient matrix inside the quadratic form has its smallest eigenvalue $\frac{3}{4} - (r-1)\mu \geq \frac{3}{4} - \frac{r-1}{4r} > \frac{1}{2}$ and all the other eigenvalues are $\mu+ \frac{3}{4}$.
Hence
\begin{equation}\label{lowbd}
\left\|\sum_{i=1}^r\A_i(\BH_{i,T_i})\right\|_F 
\geq \sqrt{ \frac{1}{2}\sum_{i=1}^r \left\|\BH_{i, T_i}\right\|_F^2 }  
\geq \frac{1}{2r} \sum_{i=1}^r \left\|\BH_{i,T_i}\right\|_F.
\end{equation}
Combining~\eqref{lowbd} with~\eqref{upperbd} leads to
\begin{equation}\label{nsp}
\frac{1}{2r} \sum_{i=1}^r \|\BH_{i,T_i}\|_F  \leq \gamma\sum_{i=1}^r\|\BH_{i,\TB_i}\|_F.
\end{equation}
$I_2$ in~\eqref{cond:a} has its lower bound as follows:
\begin{align}
I_2 & \geq \sum_{i=1}^r \Big(- \| \BH_{i,T_i} \|_F \|\bh_i\bx_i^* - \PP_{T_i}(\A_i^*(\blambda)) \|_F + \|\BH_{i,\TB_i}\|_F \left(1- \|\PP_{\TB_i}(\A_i^*(\blambda))\|\right)\Big) \label{eq:cond-a1} \\
& \geq -\alpha  \sum_{i=1}^r  \|\BH_{i, T_i}\|_F +  (1 - \beta) \sum_{i=1}^r \|\BH_{i,\TB_i}\|_F \label{eq:cond-a2} \\
& \geq -2r\gamma \alpha  \sum_{i=1}^r \|\BH_{i,\TB_i}\|_F  + (1 - \beta) \sum_{i=1}^r \|\BH_{i,\TB_i}\|_F \label{eq:cond-a3} \\
& \geq (-2r\gamma \alpha + (1 - \beta)) \sum_{i=1}^r \|\BH_{i,\TB_i}\|_F \geq 0, \nonumber
\end{align}
where~\eqref{eq:cond-a1} uses $\|\cdot\|_* \geq \|\cdot\|_F$, ~\eqref{eq:cond-a2}  follows from the assumption~\eqref{cond:suff2}, and ~\eqref{eq:cond-a3}  follows from~\eqref{nsp}.
Under the condition $-2r\gamma \alpha + (1 - \beta) > 0$, ~\eqref{cond:a} holds if at least one of  the terms $\|\BH_{i,\TB_i}\|_F$ is nonzero. If $\BH_{i,\TB_i} = \bzero$ for all $1\leq i\leq r$, then $\BH_i = \bzero$ via~\eqref{nsp}.
\end{proof}

\subsection{Local Isometry Property\label{s:local}}

Our goal in this subsection is to prove that the first assumption in~\eqref{cond:suffcond} of Lemma~\ref{lemma:suffcond} holds with high probability if $L$ is large enough. 
Instead of studying $\|\PP_{T_i}\A_i^*\A_i\PP_{T_i} - \PP_{T_i}\|$ directly, we will focus on the more general expression $\|\PP_{T_i}\A_{i,p}^*\A_{i,p}\BS_{i,p}\PP_{T_i} - \PP_{T_i}\|$, where $\A_{i,p}$ and $\BS_{i,p}$ are defined in~\eqref{def:AP} and~\eqref{def:S} respectively.

\subsubsection{An explicit formula for $\PP_{T_{i}} \A_{i,p}^{*}\A_{i,p} \BS_{i,p} \PP_{T_{i}} $}

For each fixed pair of $(i, p)$ where $1\leq i\leq r$ and $1\leq p\leq P$, the proof of $\|\PP_{T_i}\A_{i,p}^*\A_{i,p}\BS_{i,p}\PP_{T_i} - \PP_{T_i} \| \leq \frac{1}{4}$ is actually the same. 
Therefore, for simplicity of notation, we omit the subscript $i$ and denote $\PP_{T_i}\A_{i,p}^*\A_{i,p}\BS_{i,p}\PP_{T_i}$ by $\PP_{T}\A_{p}^*\A_{p}\BS_p\PP_{T}$ 
throughout the proof of Proposition~\ref{prop:lcisop}.
By definition, $\A_{p}\BS_{p} \PP_{T}(\BZ) = \{ \bb_{l}^*\BS_{p}\PP_T(\BZ)\ba_{l} \}_{l\in\Gamma_p}$ for any $\BZ\in\CC^{K\times N}$. Using~\eqref{cond:proj} gives us an explicit expression of $\bb_{l}^*\BS_{p}\PP_{T}(\BZ)\ba_{l}$, i.e.,
\begin{align*}
\bb_{l}^*\BS_{p}\PP_{T}(\BZ)\ba_{l}  
& =  \bb_{l}^*\BS_p \left[ \bh\bh^* \BZ + (\I_K - \bh\bh^*)\BZ \bx\bx^*\right]\ba_{l} \\
& =  \lag \BS_p \bh, \bb_{l} \rag  \bh^*\BZ\ba_{l} + \lag\ba_{l}, \bx \rag \bb_{l}^*\BS_p (\I_K - \bh\bh^*)\BZ\bx \\
& =  \bh^* \BZ\tbv_{l} + \tbu_{l}^* \BZ \bx, \quad l\in\Gamma_p,
\end{align*}
where $\PP_T(\BZ) = \bh\bh^* \BZ + (\I_K - \bh\bh^*)\BZ\bx\bx^*$ and both $\bh$ and $\bx$ are assumed to be real and of unit norm. Similarly,
\begin{equation*}
\bb_l^*\PP_T(\BZ)\ba_l = \bh^* \BZ\bv_{l} + \bu_{l}^* \BZ \bx, \quad l\in\Gamma_p
\end{equation*}
where 
\begin{eqnarray}\label{def:uv}
\begin{split}
\bv_{l} & :=  \lag \bh, \bb_{l}\rag \ba_{l}, \\
\bu_{l} & :=  \lag \ba_{l}, \bx\rag (\I_K - \bh\bh^*)\bb_{l}, \\
\tbv_{l} & :=  \lag \BS_p \bh, \bb_{l}\rag \ba_{l}, \\
\tbu_{l} & :=  \lag \ba_{l}, \bx\rag (\I_K - \bh\bh^*)\BS_p\bb_{l}.
\end{split}
\end{eqnarray}
Now we have
\begin{align*}
\A_p \PP_{T}(\BS_p\BZ)  = \{  \lag \BZ, \bh\tbv_{l}^* + \tbu_{l}\bx^*\rag \}_{l\in\Gamma_p}, \quad \PP_{T}\A_p^*(\bz)  = \sum_{l\in\Gamma_p} z_l (\bh\bv_{l}^* + \bu_{l}\bx^*).
\end{align*}
By combining the terms we arrive at
\begin{align}\label{eq:paap}
\begin{split}
\PP_{T}\A_{p}^*\A_{p}\BS_p \PP_{T}(\BZ) 
 & =  \sum_{l\in\Gamma_p} \Big( \bh\bh^* \BZ\tbv_{l}\bv_{l}^* 
+ \bh\tbu_{l}^* \BZ \bx \bv_{l}^* 
+ \bu_{l}\bh^*\BZ\tbv_{l} \bx^*
+ \bu_{l}\tbu_{l}^* \BZ \bx \bx^* \Big).
\end{split}
\end{align}
The explicit form of each component in this summation is
\begin{align*}
\bh\bh^*\BZ\tbv_l\bv_l^* & = \ol{\lag\bh, \bb_l \rag} \lag \BS_p\bh, \bb_l \rag \bh\bh^* \BZ \ba_l\ba_l^*, \\
\bh\tbu_l^*\BZ\bx\bv_l^* & = \ol{\lag \bh, \bb_l\rag} \bh \bb^*_l\BS_p(\I_K - \bh\bh^*)\BZ \bx \bx^*\ba_l \ba_l^*, \\
\bu_l\bh^*\BZ\tbv_l\bx^* & =  \lag \BS_p\bh, \bb_l\rag (\I_K - \bh\bh^*)\bb_l \bh^* \BZ \ba_l \ba_l^*\bx \bx^*, \\
\bu_l\tbu_l^* \BZ\bx\bx^* & = |\lag \ba_l,\bx\rag|^2 (\I_K - \bh\bh^*)\bb_l \bb_l^* \BS_p (\I_K - \bh\bh^*) \BZ\bx\bx^*.
\end{align*}

It is easy to compute the expectation of those random matrices by using $\E(\ba_l\ba_l^*) = \I_N$ and $\E |\lag \ba_l,\bx\rag|^2 = \|\bx\|^2 = 1.$ Our goal here is to estimate the operator norm of $\PP_{T}\A_{p}^*\A_{p}\BS_p \PP_{T}  - \PP_T$ which is the sum of four components, i.e., 
\begin{equation*}
\PP_{T}\A_{p}^*\A_{p}\BS_p \PP_{T}  - \PP_T = \sum_{s=1}^4\CM_s
\end{equation*}
where each $\CM_i$ is a random linear operator with zero mean. More precisely, each of $\CM_s$ is given by
\begin{align*}
&\CM_1(\BZ) := \sum_{l\in\Gamma_p} \ol{\lag\bh, \bb_l \rag} \lag \BS_p\bh, \bb_l \rag \bh\bh^* \BZ (\ba_l\ba_l^* - \I_N),  \\ 
&\CM_2(\BZ) := \sum_{l\in\Gamma_p}  \ol{\lag \bh, \bb_l\rag} \bh \bb^*_l\BS_p(\I_K - \bh\bh^*)\BZ \bx \bx^*( \ba_l \ba_l^* - \I_N),  \\
&\CM_3(\BZ) := \sum_{l\in\Gamma_p}  \lag \BS_p\bh, \bb_l\rag (\I_K - \bh\bh^*)\bb_l \bh^* \BZ (\ba_l \ba_l^* - \I_N)\bx \bx^*,  \\
&\CM_4(\BZ) :=  \sum_{l\in\Gamma_p}  (|\lag \ba_l,\bx\rag|^2 - 1)(\I_K - \bh\bh^*)\bb_l \bb_l^* \BS_p (\I_K - \bh\bh^*) \BZ\bx\bx^*.  
\end{align*}
Each $\CM_{s}$ can be treated as a $KN\times KN$ matrix because 
it is a linear operator from $\CC^{K\times N}$ to $\CC^{K\times N}$.

\begin{proposition}\label{prop:lcisop}
Under the assumption of~\eqref{cond:i} and~\eqref{cond:iso} and that $\{\ba_{i,l}\}$ are standard Gaussian random vectors of length $N_i,$
\begin{equation}\label{eq:lcisop}
\|\PP_{T_i}\A_{i,p}^*\A_{i,p}\BS_{i,p}\PP_{T_i} - \PP_{T_i}\| \leq \frac{ 1}{4},\quad 1\leq i\leq r, 1\leq p\leq P
\end{equation}
holds simultaneously  with probability at least $1 - L^{-\alpha + 1}$ if $Q\geq C_{\alpha+\log(r)}\max\{\mum^2K, \mu^2_h N\}\log^2 L)$ where $K := \max K_i$ and $N := \max N_i.$ 
\end{proposition}
The following corollary, which is a special case of Proposition~\ref{prop:lcisop} (simply set $Q=L$ and $\BS_{i,p} = \I_{K_i}$), indicates the first condition in~\eqref{cond:suffcond} holds with high probability. 
\begin{corollary}\label{eq:lciso}
Under the assumption of~\eqref{cond:i} and~\eqref{cond:iso} and that $\{\ba_{i,l}\}$ are standard Gaussian random vectors of length $N_i,$
\begin{equation}
\|\PP_{T_i}\A_{i}^*\A_{i}\PP_{T_i} - \PP_{T_i}\| \leq \frac{ 1}{4},\quad 1\leq i\leq r
\end{equation}
holds  with probability at least $1 - L^{-\alpha+1}$ if $L\geq C_{\alpha+\log(r)}\max\{\mum^2K, \mu^2_h N\}\log^2 L$ where $K := \max K_i$ and $N := \max N_i.$ 
\end{corollary}

\begin{remark}
Although Proposition~\ref{prop:lcisop} and Corollary~\ref{eq:lciso} are quite similar to Lemma 3 in~\cite{RR12} at the first glance, we include $\BS_{i,p}$ and the new definition of $\muh^2$ in our result. The purpose is to resolve the issue mentioned in Remark~\ref{remark_lemma} by making $\E(\PP_{T_i}\A_{i,p}^*\A_{i,p}\BS_{i,p}\PP_{T_i}) = \PP_{T_i}.$ Therefore we would prefer to rewrite the proof for the sake of completeness in our presentation, although the main tools are quite alike. 
\end{remark}

The proof of Proposition~\ref{prop:lcisop} is given as follows.
\begin{proof}
To prove Proposition~\ref{prop:lcisop}, it suffices to show that  $\|\CM_s\| \leq \frac{1}{16}$ for $1\leq s\leq 4$ and then take the union bound over all $1\leq p\leq P$ and $1\leq i\leq r$. For each fixed pair of $(p,i)$, it is shown in Lemma~\ref{lem:CM1-ISO} that
\begin{equation*}
\|\PP_{T_i}\A_{i,p}^*\A_{i,p}\BS_{i,p}\PP_{T_i} - \PP_{T_i}\| \leq \frac{ 1}{4}
\end{equation*}
with probability at least $1 - 4L^{-\alpha}$ if $Q \geq C_{\alpha}\max\{\mum^2K_i, \mu^2_h N_i\}\log^2 L$.  Now  we simply take the union bound over all $1\leq p\leq P$ and $1\leq i\leq r$ and obtain
\begin{align*}
& \Pr\left( \|\PP_{T_i}\A_{i,p}^*\A_{i,p}\BS_{i,p}\PP_{T_i} - \PP_{T_i}\| \leq \frac{ 1}{4}, \forall 1\leq i\leq r, 1\leq p\leq P \right) \geq 1 - 4P r L^{-\alpha} \geq 1 - 4rL^{-\alpha + 1}
\end{align*}
if $Q \geq C_{\alpha}\max\{\mum^2K, \mu^2_h N\}\log^2 L$ where there are $Pr$ events and $L = PQ$. In order to compensate for the loss of probability due to the union bound and to make the probability of success at least $1 - L^{-\alpha +1}$, we can just choose $\alpha' = \alpha + \log r$, or equivalently,
$Q \geq C_{\alpha+\log(r)}\max\{\mum^2K, \mu^2_h N\}\log^2 L$. 
\end{proof}

\subsubsection{Estimation of $\|\CM_s\|$}

\begin{lemma}\label{lem:CM1-ISO}
Under the assumptions of~\eqref{cond:i},~\eqref{cond:iso} and~\eqref{cond:ST} and that
$\ba_l\sim\mathcal{N}(\bzero, \I_N)$ independently, the estimate
\begin{equation*}
\|\CM_s\| \leq \frac{1}{16} \qquad \text{for $s =1,2,3,4$,}
\end{equation*}
holds with probability at least $1 - L^{-\alpha}$ if $Q\geq C_{\alpha} \max\{\mu^2_{\max}K, \mu^2_h N\} \log^2 (L)$.
\end{lemma}

\begin{proof}
We will only prove the bound for $\|\CM_2\|$. The estimation of $\|\CM_1\|$, $\|\CM_2\|$ and $\|\CM_4\|$ can be done similarly.

By definition of $\CM_2$, 
\begin{equation*}
\CM_2(\BZ) :=  \sum_{l\in\Gamma_p} \CZ_l(\BZ), 
\end{equation*}
where 
\begin{equation*}
\CZ_l(\BZ) = \ol{\lag \bh, \bb_l\rag} \bh \bb^*_l\BS_p(\I_K - \bh\bh^*)\BZ \bx \bx^*(\ba_l \ba_l^* - \I_N).
\end{equation*}
Immediately, we have $\|\CZ_{l}\| \leq  \| \ol{\lag \bh, \bb_l\rag} \bh \bb^*_l\BS_p \| \cdot \|(\ba_l \ba_l^* - \I_N)\|$ and $\CZ_l$ is actually a $KN\times KN$ matrix. Then we  estimate $\|\CZ_l\|_{\psi_1}$ as follows:
\begin{eqnarray*}
\|\CZ_l\|_{\psi_1} & \leq  & |\lag \bh, \bb_l\rag | \| \bh\bb_l^*\BS_p \| \cdot \| (\ba_l\ba_l^*- \I_N)\bx\|_{\psi_1} \\
& \leq & \frac{\mu_h}{\sqrt{L}} \cdot \frac{4L}{3Q} \frac{\mu_{\max}\sqrt{K}}{\sqrt{L}} \cdot \| (\ba_l\ba_l^*- \I_N)\bx\|_{\psi_1} \\
& \leq & C\frac{\mu_{\max}\mu_h\sqrt{KN}}{Q} \\
& \leq & C\frac{\max\{ \mum^2K, \mu^2_h N\}}{Q},
\end{eqnarray*}
where $\| (\ba_l\ba_l^* - \I_N)\bx \|_{\psi_1} \leq C\sqrt{N}$ follows from~\eqref{lem:11JB}. Therefore we have $R:=\max_{l\in\Gamma_p} \|\CZ_l\|_{\psi_1} \leq C\frac{ \max\{ \mu^2_{\max}K, \mu^2_hN\} }{Q}.$ 
Now we proceed to estimate $\sigma^2.$ By definition, the adjoint of $\CZ_l$ is in the form of
\begin{equation*}
\CZ_l^*(\BZ) = \lag \bh, \bb_l\rag(\I_K - \bh\bh^*)\BS_p \bb_l\bh^*\BZ (\ba_l \ba_l^* - \I_N)\bx \bx^*.
\end{equation*}
Then $\CZ^*\CZ$ and $\CZ\CZ^*$ are easily obtained by definition, 
\begin{align*}
\CZ_l^*\CZ_l(\BZ) & = |\lag \bh, \bb_l\rag|^2 (\I_K - \bh\bh^*)\BS_p \bb_l\bb_l^*\BS_p  (\I_K - \bh\bh^*)\BZ \bx\bx^*(\ba_l \ba_l^* - \I_N)^2\bx\bx^*
\end{align*}
and 
\begin{align*}
\CZ_l\CZ_l^*(\BZ) 
& = |\lag \bh, \bb_l\rag|^2\bb_l^*\BS_p(\I_K - \bh\bh^*)\BS_p \bb_l \bh\bh^*\BZ (\ba_l \ba_l^* - \I_N)\bx \bx^* (\ba_l\ba_l^* - \I_N).
\end{align*}
The expectation of $\CZ_l^*\CZ_l$ and $\CZ_l\CZ_l^*$ are computed via
\begin{align*}
\E(\CZ_l^*\CZ_l(\BZ)) 
& =  (N +1) |\lag \bh, \bb_l\rag|^2  (\I_K - \bh\bh^*)\BS_p \bb_l\bb_l^*\BS_p(\I_K - \bh\bh^*)\BZ \bx\bx^*
\end{align*}
where $\E(\ba_l\ba_l^* - \I_N)^2 = (N+1)\I_N$ follows from~\eqref{lem:9JB}. Similarly,
\begin{align*}
\E(\CZ_l\CZ_l^*(\BZ)) 
& =  |\lag \bh, \bb_l\rag|^2 \bb_l^*\BS_p(\I_K - \bh\bh^*)\BS_p \bb_l  \bh\bh^*\BZ (\I_N + \bx\bx^*) 
\end{align*}
where $\E[(\ba_l\ba_l^* - \I_N)\bx\bx^*(\ba_l\ba_l^* - \I_N)] = \|\bx\|^2 \I_N + \bx\bx^*$ from~\eqref{lem:12JB} and the fact that $\bx$ is real.
Taking the sum of $\E(\CZ^*_l\CZ_l)$ and $\E(\CZ_l\CZ_l^*)$ over $l\in\Gamma_p$ gives
\begin{align*}
\left\| \sum_{l\in\Gamma_p} \E(\CZ_l^*\CZ_l) \right\| 
& = (N + 1) \left\| \sum_{l\in\Gamma_p} |\lag \bh, \bb_l\rag|^2 \BS_p \bb_l\bb_l^*\BS_p\right\| \\
& \leq \frac{2\mu^2_h N}{L} \| \BS_p \|  \leq \frac{2\mu^2_h N}{L} \cdot \frac{4L}{3Q} =  \frac{8\mu^2_h N}{3Q}
\end{align*}
and
\begin{align*}
 \left\|\sum_{l\in\Gamma_p} \E(\CZ_l\CZ_l^*)\right\| & \leq  2\left\| \sum_{l\in\Gamma_p} |\lag \bh, \bb_l\rag|^2 \bb_l^*\BS_p(\I_K - \bh\bh^*)\BS_p \bb_l \right\| \\
& \leq  2\max_{l\in\Gamma_p}\{ \bb_l^*\BS_p(\I_K - \bh\bh^*)\BS_p \bb_l  \} \sum_{l\in\Gamma_p} |\lag \bh, \bb_l\rag|^2 \\ 
& \leq  2\left[\|\BS_{p}\|^{2} \max_{l\in\Gamma_{p}}\|\bb_{l}\|^{2}\right] \cdot \|\BT_p\| \\
& \leq  \frac{32L^{2}}{9Q^{2}} \cdot \frac{\mum^2K}{L} \cdot  \frac{5Q}{4L}\\
&  \leq  \frac{40 \mum^2K}{9Q}.
\end{align*}
Thus the variance $\sigma^2$ is bounded above by
\begin{equation*}
\sigma^2 \leq C \frac{\max\{ \mum^2K, \mu^2_h N\}}{Q}
\end{equation*}
and $\log\left(\frac{\sqrt{Q}R}{\sigma}\right) \leq C_1\log L$ for some constant $C_1$.
Then we just use~\eqref{thm:bern} to estimate the deviation of $\CM_2$ from $\bzero$ by choosing $t = \alpha\log L$. Setting $Q\geq C_{\alpha} \max\{\mu^2_{\max}K, \mu^2_h N\} \log^2L/\delta^2$ gives us
\begin{align*}
\|\CM_2\| & \leq C\max\Big\{ \sqrt{ \frac{\max\{\mum^2K, \mu^2_h N\}}{Q} (\alpha+2)\log L },  \frac{\max\{\mum^2K, \mu^2_h N\}}{Q} (\alpha+2)\log^2 L  \Big\}  \leq \delta
\end{align*}
where $K$ and $N$ are properly assumed to be smaller than $L.$
In particular, we take $\delta = \frac{1}{16}$  and have
\begin{equation*}
\|\CM_2\| \leq \frac{1}{16}
\end{equation*}
with the probability at least $1 - L^{-\alpha}.$
\end{proof}

\subsection{Proof of $\mu \leq \frac{1}{4r}$ \label{s:incoherence}}

In this section, we aim to show that $\mu \leq \frac{1}{4r}$, where $\mu$ is defined in~\eqref{def:mu}, i.e., the second condition in~\eqref{cond:suffcond} holds with high probability.  The main idea here is first to show that a more general and stronger version of incoherent property, 
\begin{equation*}
\| \PP_{T_j}\A_{j,p}^*\A_{k,p}\BS_{k,p}\PP_{T_k} \| \leq \frac{1}{4r}
\end{equation*}
 holds with high probability for any $1\leq p\leq P$ and $j\neq k$. Since  the derivation is essentially the same for all different pairs of $(j,k)$ with $j\neq k$, without loss of generality, we take $j = 1$ and $k=2$ as an example throughout this section. We finish the proof by taking the union bound over all possible sets of $(j,k,p).$

\subsubsection{An explicit expression for $\PP_{T_2}\A_{2,p}^*\A_{1,p}\BS_{1,p}\PP_{T_1}$}

Following the same procedures as the previous section,  we have explicit expressions for $\A_{1,p}\PP_{T_1}$ and $\PP_{T_2}\A_{2,p}^*$,
\begin{align*}
\A_{1,p}\BS_{1,p}\PP_{T_1}(\BZ) & = \{  \lag \BZ, \bh_1\tbv_{1,l}^* + \tbu_{1,l}\bx_1^*\rag \}_{l\in\Gamma_p}, \\
 \PP_{T_2}\A_{2,p}^*(\bz) & = \sum_{l\in\Gamma_p} z_l (\bh_2\bv_{2,l}^* + \bu_{2,l}\bx_2^*).
\end{align*}
where $\tbu_{1,l}$, $\tbv_{1,l}$, $\bu_{2,l}$ and $\bv_{2,l}$ are defined in~\eqref{def:uv} except the notation, where we omit subscript $i$ in the previous section. 
By combining $\PP_{T_2}\A_{2,p}^*$ and $\A_{1,p}\BS_{1,p}\PP_{T_1}$, we arrive at
\begin{align}\label{eq:paasp}
\begin{split}
\PP_{T_2}\A_{2,p}^*\A_{1,p}\BS_{1,p}\PP_{T_1}(\BZ) =    \sum_{l\in\Gamma_p} \Big( \bh_2\bh_1^*\BZ\tbv_{1,l}\bv_{2,l}^* + \bh_2\tbu_{1,l}^* \BZ \bx_1 \bv_{2,l}^* + \bu_{2,l}\bh_1^*\BZ\tbv_{1,l} \bx_{2}^*
+ \bu_{2,l}\tbu_{1,l}^* \BZ \bx_1 \bx_2^* \Big).
\end{split}
\end{align}
Note that the expectations of all terms are equal to $\bzero$ because $\{ \bu_{1,l}, \bv_{1,l} \}$ is independent
of $\{ \bu_{2,l}, \bv_{2,l} \}$ and both $\bu_{i,l}$ and $\bv_{i,l}$ have zero mean. Define $\CM_{s, \mix}$ as
\begin{eqnarray*}
\begin{split}
\CM_{1,\mix}(\BZ) & := \sum_{l\in\Gamma_p}\bh_2\bh_1^*\BZ \tbv_{1,l}\bv_{2,l}^* =  \sum_{l\in\Gamma_p}\lag \BS_{1,p}\bh_1, \bb_{1,l}\rag \ol{\lag \bh_2, \bb_{2,l}\rag} \bh_2\bh_1^*\BZ
\ba_{1,l}\ba_{2,l}^*, \\ 
\CM_{2,\mix}(\BZ) & := \sum_{l\in\Gamma_p} \bh_2\tbu_{1,l}^*\BZ\bx_1\bv_{2,l}^* =   \sum_{l\in\Gamma_p}\lag \ba_{1,l}, \bx_1\rag \ol{\lag \bh_2, \bb_{2,l}\rag} \bh_2 \bb_{1,l}^* \BS_{1,p}(\I_{K_1} -
\bh_1\bh_1^*)\BZ\bx_1 \ba_{2,l}^*, \\ 
 \CM_{3,\mix}(\BZ) & := \sum_{l\in\Gamma_p} \bu_{2,l}\bh_1^*\BZ\tbv_{1,l}\bx_2^* =   \sum_{l\in\Gamma_p}\lag \ba_{2,l}, \bx_2 \rag \lag \BS_{1,p}\bh_1, \bb_{1,l}\rag (\I_{K_2} - \bh_2\bh_2^*)\bb_{2,l}
\bh_1^* \BZ \ba_{1,l} \bx_2^*, \\ 
\CM_{4,\mix}(\BZ) &:= \sum_{l\in\Gamma_p} \bu_{2,l}\tbu_{1,l}^* \BZ \bx_1 \bx_2^* =   \sum_{l\in\Gamma_p}\Big(\lag \ba_{1,l}, \bx_1\rag \lag\ba_{2,l}, \bx_2 \rag (\I_{K_2} -
\bh_2\bh_2^*)\bb_{2,l}\bb_{1,l}^* \BS_{1,p}(\I_{K_1} - \bh_1\bh_1^*) \BZ \bx_1\bx_2^*\Big), 
\end{split}
\end{eqnarray*}
and there holds
\begin{equation*}
\PP_{T_2}\A_{2,p}^*\A_{1,p}\BS_{1,p}\PP_{T_1}  = \sum_{s=1}^4 \CM_{s, \mix}.
\end{equation*}
Each $\CM_{s,\mix}$ can be treated as a $K_2N_2\times K_1N_1$ matrix because 
it is a linear operator from $\CC^{K_1\times N_1}$ to $\CC^{K_2\times N_2}$.

\begin{proposition}\label{prop:mix}
Under the assumption of~\eqref{cond:i} and~\eqref{cond:iso} and that $\{\ba_{i,l}\}$ are standard Gaussian random vectors of length $N_i,$
\begin{equation}\label{eq:mix}
\|\PP_{T_j}\A_{j,p}^*\A_{k,p}\BS_{k,p}\PP_{T_k} \| \leq \frac{ 1}{4r },\quad 1\leq j\neq k\leq r, 1\leq p\leq P
\end{equation}
holds  with probability at least $1 - L^{-\alpha + 1}$ if $Q\geq C_{\alpha+ \log(r)}r^2 \max\{\mum^2K, \mu^2_h N\}\log^2 L$ where $K := \max K_i$ and $N := \max N_i.$ 
\end{proposition}
By setting $Q = L$, we immediately have $\mu \leq \frac{1}{4r}$, which is written into the following corollary. 
\begin{corollary}\label{for:mix}
Under the assumption of~\eqref{cond:i} and~\eqref{cond:iso} and that $\{\ba_{i,l}\}$ are standard Gaussian random vectors of length $N_i,$ there holds
\begin{equation}\label{eq:mix}
\|\PP_{T_j}\A_{j}^*\A_{k}\PP_{T_k} \| \leq \frac{1}{4r},\quad 1\leq j\neq k\leq r,
\end{equation}
with probability at least $1 - L^{-\alpha + 1}$ if $L\geq C_{\alpha+2\log(r)}r^2 \max\{\mum^2K, \mu^2_h N\}\log^2 L$ where $K := \max K_i$ and $N := \max N_i.$ In other words, $\mu \leq \frac{1}{4r}.$
\end{corollary}

The proof of Proposition~\ref{prop:mix} follows two steps. First we will show each $\| \CM_{s,\mix} \| \leq \frac{1}{16r}$ holds with high probability, followed by taking the union bound over all $j\neq k$ and $1\leq p\leq P$.
\begin{proof}
For any fixed set of $(j,k,p)$ with $j\neq k$, it has been shown, in Lemma~\ref{lem:CM1-mix}, that
\begin{equation*}
\|\PP_{T_j}\A_{j,p}^*\A_{k,p}\BS_{k,p}\PP_{T_k} \| \leq \frac{ 1}{4r}
\end{equation*}
with probability at least $1 - 4L^{-\alpha}$ if $Q \geq C_{\alpha} r^2 \max\{\mum^2K, \mu^2_h N\}\log^2L$.  Then we simply take the union bound over all $1\leq p\leq P$ and $1\leq j\neq k\leq r$ and it leads to
\begin{align*}
\Pr\left( \|\PP_{T_j}\A_{j,p}^*\A_{k,p}\BS_{k,p}\PP_{T_k} \| \leq \frac{1}{4r}, \quad \forall j\neq k, 1\leq p\leq P \right) \geq 1 - 4L^{-\alpha} P r^2/2 \geq 1 - 2L^{-\alpha + 1}r^2
\end{align*}
if $Q \geq C_{\alpha}r^2 \max\{\mum^2K, \mu^2_h N\}\log^2L$ where there are at most $Pr^2/2$ events and $L = PQ$. In order to make the probability of success at least $1 - L^{-\alpha +1}$, we can just choose $\alpha' = \alpha + 2\log r$, or equivalently,
$Q \geq C_{\alpha + 2\log(r)}r^2 \max\{\mum^2K, \mu^2_h N\}\log^2 L$. 
\end{proof}

\subsubsection{Estimation of $\|\CM_{s,\mix}\|$}

\begin{lemma}\label{lem:CM1-mix}
Under the assumptions of~\eqref{cond:i},~\eqref{cond:iso} and~\eqref{cond:ST} and that
$\ba_{i,l}\sim\mathcal{N}(\bzero, \I_{N_i})$ independently for $i=1,2$ and $l\in\Gamma_p$, then
\begin{equation*}
\|\CM_{s,\mix}\| \leq \frac{1}{16r} \qquad \text{for $s=1,2,3,4$},
\end{equation*}
holds with probability at least $1 - L^{-\alpha}$ if $Q\geq C_{\alpha} r^2 \max\{\mum^2K, \mu^2_h N\} \log^2 L$.
\end{lemma}

\begin{proof}
We only prove the bound for $\CM_{2,\mix}$, the proofs of the bounds for $\CM_{1,\mix}, \CM_{3,\mix}$, and
$\CM_{4,\mix}$ use similar arguments and are left to the reader.

Following from the definition of $\CM_{2,\mix}$, 
\begin{equation*}
\CM_{2,\mix} :=  \sum_{l\in\Gamma_p} \CZ_l(\BZ),
\end{equation*}
where
\begin{equation*}
 \CZ_l(\BZ) =
 \lag \ba_{1,l}, \bx_1\rag \ol{\lag \bh_2, \bb_{2,l}\rag} \bh_2 \bb_{1,l}^* \BS_{1,p}(\I_{K_1} - \bh_1\bh_1^*)\BZ\bx_1 \ba_{2,l}^*.
\end{equation*}
Note that $\|\CZ_l\| \leq  |\lag \ba_{1,l}, \bx_1\rag \ol{\lag \bh_2, \bb_{2,l}\rag}| \| \bh_2 \bb_{1,l}^* \BS_{1,p} \| \| \bx_1 \ba_{2,l}^* \| $, and by using Lemma~\ref{lem:6JB} and~\ref{lemma:psi}, we have
\begin{align*}
\|\CZ_l\|_{\psi_1} 
 \leq &  |\lag \bh_2, \bb_{2,l}\rag|  \|\BS_{1,p}\bb_{1,l}\| \cdot \| (|\lag \ba_{1,l}, \bx_1\rag| \cdot\|\ba_{2,l}\|) \|_{\psi_1} \\
 \leq & C\frac{\mum\mu_h \sqrt{K_1}}{Q} \| (|\lag \ba_{1,l}, \bx_1\rag| \cdot\|\ba_{2,l}\|) \|_{\psi_1} \\
 \leq & C\frac{\mum\mu_h \sqrt{K_1N_2}}{Q} \\
 \leq & C\frac{\max\{\mum^2 K_1, \mu^2_h N_2\}}{Q}
\end{align*}
where 
$\| (|\lag \ba_{1,l}, \bx_1\rag| \cdot\|\ba_{2,l}\|) \|_{\psi_1} \leq C\sqrt{N_2}$ follows from Lemma~\ref{lemma:psi}.
We proceed to estimate $\sigma^2$ by first finding $\CZ_l^*(\BZ),$
\begin{equation*}
\CZ^*_l(\BZ) = \ol{\lag \ba_{1,l}, \bx_1\rag} \lag \bh_2, \bb_{2,l}\rag   (\I_{K_1} - \bh_1\bh_1^*)\BS_{1,p}\bb_{1,l} \bh_2^*  \BZ\ba_{2,l}\bx_1^*.
\end{equation*}

Both $\CZ^*_l\CZ_l(\BZ)$ and $\CZ_l\CZ^*_l(\BZ)$ have the following forms:
\begin{align*}
\CZ_l^*\CZ_l(\BZ) & =  |\lag \ba_{1,l}, \bx_1\rag \lag \bh_2, \bb_{2,l}\rag|^2 \|\ba_{2,l}\|^2  (\I_{K_1} - \bh_1\bh_1^*) \BS_{1,p}\bb_{1,l} \bb_{1,l}^*\BS_{1,p} (\I_{K_1} - \bh_1\bh_1^*)  \BZ\bx_1\bx_1^*
\end{align*}
and 
\begin{align*}
\CZ_l\CZ_l^*(\BZ) & =  |\lag \ba_{1,l}, \bx_1\rag \lag \bh_2, \bb_{2,l}\rag|^2   \bb_{1,l}^*\BS_{1,p} (\I_{K_1} - \bh_1\bh_1^*) \BS_{1,p}\bb_{1,l} \bh_2\bh_2^*  \BZ\ba_{2,l}\ba_{2,l}^*.
\end{align*}
The expectations of $\CZ_l^*\CZ_l$ and $\CZ_l\CZ_l^*$ are 
\begin{align*}
\E(\CZ_l^*\CZ_l(\BZ)) & = N_2| \lag \bh_2, \bb_{2,l}\rag|^2   (\I_{K_1} - \bh_1\bh_1^*) \BS_{1,p}\bb_{1,l} \bb_{1,l}^*\BS_{1,p} (\I_{K_1} - \bh_1\bh_1^*)  \BZ\bx_1\bx_1^*
\end{align*}
and 
\begin{align*}
\E(\CZ_l\CZ_l^*(\BZ)) & =  |\lag \bh_2, \bb_{2,l}\rag|^2 \bb_{1,l}^*\BS_{1,p}  (\I_{K_1} - \bh_1\bh_1^*)\BS_{1,p}\bb_{1,l} \bh_2\bh_2^*  \BZ.
\end{align*}
Taking the sum over $l\in\Gamma_p$ results in 
\begin{align*}
\left\| \sum_{l\in\Gamma_p} \E(\CZ_l^*\CZ_l) \right\| & = N_2 \left\| \sum_{l\in\Gamma_p} |\lag \bh_2, \bb_{2,l}\rag|^2  \BS_{1,p}\bb_{1,l}\bb_{1,l}^*\BS_{1,p}  \right\| \\
& \quad\leq  \frac{\mu^2_hN_2}{L} \left\| \sum_{l\in\Gamma_p} \BS_{1,p}\bb_{1,l}\bb_{1,l}^* \BS_{1,p}\right\| \\
& \quad \leq \frac{\mu^2_h N_2}{L} \cdot \|\BS_{1,p}\| \leq \frac{\mu^2_hN_2}{L} \cdot \frac{4L}{3Q} = \frac{4\mu^2_hN_2}{3Q}
\end{align*}
and
\begin{align*}
 \left\|\sum_{l\in\Gamma_p} \E(\CZ_l\CZ_l^*)\right\| & = \sum_{l\in\Gamma_p} |\lag \bh_2, \bb_{2,l}\rag|^2 \bb_{1,l}^*\BS_{1,p} (\I_{K_1} - \bh_1\bh_1^*)\BS_{1,p}\bb_{1,l} \\
&  = \sum_{l\in\Gamma_p} \|(\I_{K_1} - \bh_1\bh_1^*) \BS_{1,p}\bb_{1,l}\|^2 |\lag \bh_2, \bb_{2,l}\rag|^2  \\
& \leq \max_{l\in\Gamma_p} \{ \|\BS_{1,p}\|^2 \|\bb_{1,l}\|^2 \} \cdot \| \BT_{2,p} \|\\
& \leq \frac{16L^2}{9Q^2} \cdot \frac{\mum^2K_1}{L} \cdot \frac{5Q}{4L} \leq\frac{20\mum^2K_1}{9Q}.
\end{align*}
Thus the variance $\sigma^2$ is bounded by
\begin{equation*}
\sigma^2 \leq C\frac{\max\{\mum^2K_1, \mu^2_{h}N_2\}}{Q} \leq C\frac{\max\{\mum^2K, \mu^2_{h}N\}}{Q}.
\end{equation*}
Then we just apply~\eqref{thm:bern} to estimate the deviation of $\CM_{2,\mix}$ from $\bzero$
by choosing $t = \alpha\log L$. 

Letting $Q\geq C_{\alpha}\max\{\mum^2K,\muh^2N\} \log^2 L/\delta^2$ gives us
\begin{align*}
\|\CM_{2,\mix}\| & \leq  C\max\Big\{ \sqrt{ \frac{\max\{\mum^2K, \muh^2N\}}{Q} (\alpha + 2)\log L}, \frac{\max\{\mum^2K, \muh^2N\}}{Q} (\alpha + 2)\log^2 L  \Big\}  \leq \delta,
\end{align*}
with probability at least $1 - L^{-\alpha}$ where $K$ and $N$ are properly assumed to be smaller than $L.$
Let $\delta = \frac{1}{16r}$ and $Q\geq C_{\alpha} r^2 \max\{\mum^2K, \mu^2_hN\} \log^2L$, 
\begin{equation*}
\|\CM_{2,\mix}\| \leq \frac{1}{16r}
\end{equation*}
with the probability at least $1 - L^{-\alpha}.$
\end{proof}

\subsection{Constructing a dual certificate\label{s:dual}}

In this section, we will construct a $\blambda$ such that
\begin{equation}\label{eq:suffcond2}
\|\bh_i\bx_i^* - \PP_{T_i}(\A_i^*(\blambda))\|_F \leq (5r\gamma)^{-1}, \quad \|\PP_{\TB_i}(\A^*_i(\blambda))\| \leq \frac{1}{2}
\end{equation}
holds simultaneously for all $1\leq i\leq r$.
If such a $\blambda$ exists, then solving~\eqref{cvxprog} yields exact recovery according to Lemma~\ref{lemma:suffcond}. The difficulty of this mission is obvious since we require all $\A_i^*(\blambda)$ to be close to $\bh_i\bx_i^*$ and ``small" on $\TB_i$. However, it becomes possible with help of the incoherence between $\A_i$ and $\A_j$. The method to achieve this goal is to apply a well-known and widely used technique called \emph{golfing scheme}, developed by Gross in~\cite{gross11recovering}. 

\subsubsection{Construct an approximate dual certificate via golfing scheme}

The approximate dual certificate $\{ \BY_{i} := \A_i^*(\blambda) \}_{i=1}^r$ satisfying Lemma~\ref{lemma:suffcond} is constructed via a sequence of random matrices, by following the philosophy of golfing scheme. The constructed sequence $\{\BY_{i,p}\}_{p=1}^P$ would approach $\bh_i\bx_i^*$ on $T_i$ exponentially fast while keeping $\BY_{i,p}$ ``small" on $\TB_{i}$ at the same time.

\begin{enumerate}
\item Initialize $\BY_{i,0} = \bzero_{K_i\times N_i}$ for all $1\leq i\leq r$ and 
\begin{equation*}
\blambda_0 : = \sum_{j=1}^r\A_{j,1} (\BS_{j,1}\bh_j\bx_j^*) \in \CC^L.
\end{equation*}
\item For $p$ from $1$ to $P$ (where $P$ will be specified later in Lemma~\ref{lem:Y-hx}), we define the following recursive formula:
\begin{eqnarray}
\blambda_{p-1} & : = & \sum_{j=1}^r\A_{j,p} \lp \BS_{j,p}( \bh_j\bx_j^* - \PP_{T_j}(\BY_{j,p-1}))\rp, \label{eq:lambda}\\
\BY_{i,p}  & : = & \BY_{i,p-1} + \A_{i,p}^*(\blambda_{p-1}), \quad 1\leq i\leq r.\label{eq:gs1} 
\end{eqnarray}
\item $\BY_{i,p}$ denotes the result after $p$-th iteration  and let $\BY_i := \BY_{i,P}$, i.e., the final outcome for each $i$, and $\blambda :=\blambda_P.$ 
\end{enumerate}
Denote $\BW_{i,p}$ as the difference between $\BY_{i,p}$ and $\bh_i\bx_i^*$ on $T_i$, i.e.,
\begin{equation}\label{def:W}
\BW_{i,p} = \bh_i\bx_i^* - \PP_{T_i}(\BY_{i,p})\in T_i, \quad \BW_{i,0} = \bh_i\bx_i^*
\end{equation}
and~\eqref{eq:lambda} can be rewritten into the following form:
\begin{equation*}
\blambda_{p-1} =  \sum_{j=1}^r\A_{j,p} (\BS_{j,p} \BW_{j,p-1}).
\end{equation*}
Moreover, $\BW_{i,p}$ yields the following relation:
\begin{equation}\label{eq:ww}
\BW_{i,p} = \BW_{i,p-1} - \sum_{j=1}^r \PP_{T_i}\A_{i,p}^*\A_{j,p} (\BS_{j,p}\BW_{j,p-1})
\end{equation}
from~\eqref{eq:gs1} and~\eqref{def:W}.
\begin{remark}
Here we give an intuitional reason why $\blambda_p$ is constructed as~\eqref{eq:lambda}.
An important observation here is that each $\A_{i,p}^*(\blambda_{p-1})$ is an {\em unbiased} estimator of $\BW_{i,p-1}$, i.e.,
\begin{equation}\label{eq:Alambda}
\E(\A_{i,p}^*(\blambda_{p-1})) = \sum_{j=1}^r\E( \A_{i,p}^*  \A_{j,p} (\BS_{j,p} \BW_{j,p-1}) )  = \BW_{i,p-1}
\end{equation}
where $\E(\A_{i,p}^*\A_{j,p}(\BS_{j,p}\BW_{j,p-1})) = \bzero$ for all $j\neq i$ due to the independence between $\A_{j,p}$ and $\A_{i,p}$ and $\E(\A_{i,p}^*\A_{i,p}(\BS_{i,p}\BW_{i,p-1})) = \BW_{i,p-1}$.
Remember that $\{\BW_{j,p-1}\}_{j=1}^r$ are independent of $\{ \A_{i,p}\}_{i=1}^r$, which follows from the construction of sequences in~\eqref{eq:lambda} and~\eqref{eq:gs1}. More precisely, the expectation above should be treated as the conditional expectation of $\A_{i,p}^*(\blambda_{p-1})$ given $\{\BW_{j,p-1}\}_{j=1}^r$ are known.
\end{remark}

\subsubsection{$\|\PP_{T_i}(\BY_i) - \bh_i\bx_i^*\|_F$ decays exponentially fast}

\begin{lemma}\label{lem:Y-hx}
Conditioned on~\eqref{eq:lcisop} and~\eqref{eq:mix}, the golfing scheme~\eqref{eq:lambda} and~\eqref{eq:gs1} generate a sequence of $\{\BY_{i,p}\}_{p=1}^P$ such that
\begin{equation*}
 \|\BW_{i,p}\|_F = \|\PP_{T_i}(\BY_{i,p}) - \bh_i\bx_i^*\|_F \leq 2^{-p}
\end{equation*}
hold simultaneously for all $1\leq i\leq r.$
In particular, if $P \geq \log_2 (5r\gamma)$, 
\begin{equation*}
\|\PP_{T_i}(\BY_{i}) - \bh_i\bx_i^*  \|_F \leq 2^{-\log_2 (5r\gamma)} \leq \frac{1}{5r\gamma}
\end{equation*}
where $\BY_{i} := \BY_{i,P}$.
In other words, the first condition in~\eqref{eq:suffcond2} holds. 
\end{lemma}

\begin{proof}
Directly following from~\eqref{eq:ww} leads to
\begin{equation}\label{eq:wp2}
\begin{split}
\BW_{i,p} & =  \BW_{i,p-1} - \PP_{T_i} \A^*_{i,p} \A_{i,p}(\BS_{i,p}\BW_{i,p-1})  - \sum_{j\neq i}^r \PP_{T_i} \A^*_{i,p} \A_{j,p}(\BS_{j,p}\BW_{j,p-1}).
\end{split}
\end{equation}
Note that $\BW_{j,p-1}\in T_j$ and thus $\BW_{j,p-1} = \PP_{T_j}(\BW_{j,p-1})$.
By applying~\eqref{eq:lcisop} and~\eqref{eq:mix},
\begin{equation*}
\|\BW_{i,p}\|_F \leq \frac{1}{4} \|\BW_{i,p-1}\|_F + \frac{1}{4r}\sum_{j\neq i} \|\BW_{j,p-1}\|_F, \quad 1\leq i\leq r.
\end{equation*}
It is easy to see that
\begin{equation*}
\max_{1\leq i\leq r}\| \BW_{i,p}\|_F \leq \frac{1}{2} \max_{1\leq i\leq r} \|\BW_{i,p-1}\|_F.
\end{equation*}
Recall that $\| \BW_{i,0}\|_F = \|\bh_i\bx_i^*\|_F = 1$ for all $1\leq i\leq r$ and by the induction above, we prove that
\begin{equation*}
\|\BW_{i,p}\|_F \leq 2^{-p}, \quad 1\leq p\leq P, \quad 1\leq i\leq r.
\end{equation*} 
\end{proof}

\subsubsection{Proof of $\|\PP_{\TB_i}(\BY_{i})\| \leq \frac{1}{2}$}

In the previous section, we have already shown that $\PP_{T_i}(\BY_{i,p})$ approaches $\bh_i\bx_i^*$ exponentially fast with respect to $p$.
The only missing piece of the proof is to show that $\|\PP_{\TB_i}(\BY_{i,{P}})\|$ is bounded by $\frac{1}{2}$ for all $1\leq i\leq r$, i.e., the second condition in~\eqref{cond:suff2} holds. 
Following directly from~\eqref{eq:lambda} and~\eqref{eq:gs1}, 
\begin{equation*}
\BY_i :=\BY_{i,P} = \sum_{p=1}^{P} \A_{i,p}^* (\blambda_{p-1}).
\end{equation*}
There holds
\begin{align*}
\|\PP_{\TB_i}(\BY_{i})\| 
& = \left\| \PP_{\TB_i}\left( \sum_{p=1}^{P} (\A_{i,p}^*(\blambda_{p-1}) - \BW_{i,p-1}) \right) \right\|  \leq \sum_{p=1}^{P} \| \A_{i,p}^*(\blambda_{p-1}) - \BW_{i,p-1}\|,
\end{align*}
where  $\PP_{\TB_i}(\BW_{i,p-1}) = \bzero$.
It suffices to demonstrate that $\| \A_{i,p}^*(\blambda_{p-1}) -  \BW_{i,p-1}\| \leq 2^{-p-1} $  for $1\leq p\leq P $  in order to justify $\|\PP_{\TB_i}(\BY_{i})\| \leq \frac{1}{2}$ since
\begin{equation*}
\|\PP_{T_i}(\BY_{i})\| \leq \sum_{p=1}^{P} 2^{-p-1} < \frac{1}{2}.
\end{equation*}
Before moving to the proof, we first define the quantity $\mu_p$ which will be useful in the proof,

\begin{equation}\label{def:mup}
\mu_p := \frac{Q}{\sqrt{L}}\max_{1\leq i\leq r, l\in\Gamma_{p+1}} \| \BW_{i,p}^*\BS_{i,p+1} \bb_{i,l}\|.
\end{equation}
In particular, $\mu_0 \leq \mu_h$ because of
\begin{equation*}
\mu_0 = \frac{Q}{\sqrt{L}} \max_{i,l\in \Gamma_1}\| \bx_i\bh_i^*\BS_{i,1} \bb_{i,l}\| =   \frac{Q}{\sqrt{L}} \max_{i,l\in \Gamma_1}\| \bh_i^*\BS_{i,1} \bb_{i,l}\| \leq \mu_h.
\end{equation*} 
and the definition of $\mu_h$ in~\eqref{def:muh}. Also we define $\bw_{i,l}$ as
\begin{equation}\label{def:w}
\bw_{i,l} := \BW^*_{i,p-1}\BS_{i,p} \bb_{i,l}\in\CC^{N_i}, \quad l\in\Gamma_p, \quad 1\leq i\leq r
\end{equation}
and there holds
\begin{equation}\label{def:mup-1}
\max_{1\leq i\leq r,l\in\Gamma_p} \|\bw_{i,l}\| \leq \frac{\sqrt{L}}{Q}\mu_{p-1}.
\end{equation}
\begin{remark}
The definition of $\mu_p$ is a little complicated but the idea is simple. Since we have already shown in Lemma~\ref{lem:Y-hx} that $\BW_{i,p}\in T_i$ is very close to $\bh_i\bx_i^*$ for large $p$, $\mu_p$ can be viewed as a measure of the incoherence between $\BW_{i,p}$ (an approximation of $\bh_i\bx_i^*$) and $\{\bb_{i,l}\}_{l\in\Gamma_{p+1}}$. We would like to have ``small" $\mu_p$, i.e., $\mu_p\leq \|\BW_{i,p}\|\mu_h \leq 2^{-p}\mu_h$ which guarantees that $\A_{i,p}^*(\blambda_{p-1})$ concentrates well around $\BW_{i,p-1}$ for all $i$ and $p.$ This insight leads us to the following lemma.
\end{remark}

\begin{lemma}
\label{lem:normbound}
Let $\mu_p$ be defined in~\eqref{def:mup}  and $\BW_{i,p}$ satisfy
\begin{equation*}
\mu_p \leq 2^{-p}\mu_h, \quad \|\BW_{i,p}\|_F \leq 2^{-p}, \quad 1\leq p\leq P,  1\leq i\leq r.
\end{equation*}
If $Q\geq C_{\alpha+\log(r)}r\max\{\mum^2K, \mu^2_h N \}\log^2 L$, then
\begin{equation*}
\| \A^*_{i,p} (\blambda_{p-1}) - \BW_{i,p-1} \| \leq 2^{-p-1}, 
\end{equation*}
simultaneously for $(p,i)$ with probability at least $1 - L^{-\alpha + 1}$. Thus,  the second condition in~\eqref{eq:suffcond2},
\begin{equation*}
\|\PP_{\TB_i}(\BY_{i})\| \leq \frac{1}{2}
\end{equation*}
holds simultaneously for all $1\leq i\leq r$.
\end{lemma}
\begin{remark} 
The assumption $\mu_p \leq 2^{-p}\mu_h$ is justified in Lemma~\ref{lem:mup-half}.
\end{remark}
\begin{proof}
It is shown in~\eqref{eq:Alambda} that
\begin{equation*}
\E\left(  \A_{i,p}^*(\blambda_{p-1})  - \BW_{i,p-1} \right)  = 0.
\end{equation*}
For any fixed $1\leq i\leq r$, we rewrite $ \A_{i,p}^*(\blambda_{p-1})  - \BW_{i,p-1}$ into the sum of rank-1 matrices with mean $\bzero$ by~\eqref{eq:lambda} and~\eqref{def:AP},
\begin{equation}\label{eq:dualsum}
\begin{split}
\A_{i,p}^*(\blambda_{p-1}) -  \BW_{i,p-1} 
 = \sum_{l\in\Gamma_p} \Big(\bb_{i,l}\bb_{i,l}^* \BS_{i,p} \BW_{i,p-1}\left(\ba_{i,l}\ba_{i,l}^* - \I_{N_i}\right)  + \sum_{j\neq i}\bb_{i,l}\bb_{j,l}^* \BS_{j,p} \BW_{j,p-1}\ba_{j,l}\ba_{i,l}^*\Big).
\end{split}
\end{equation}
Denote $\CZ_l$ by
\begin{equation}\label{eq:dualcz}
\CZ_l := \bb_{i,l}\bw_{i,l}^*\left(\ba_{i,l}\ba_{i,l}^* - \I_{N_i}\right) 
+ \sum_{j\neq i}\bb_{i,l}\bw_{j,l}^*\ba_{j,l}\ba_{i,l}^*\in\CC^{K_i\times N_i}
\end{equation}
where $\bw_{j,l}$ is defined in~\eqref{def:w}.
The goal is to bound the operator norm of~\eqref{eq:dualsum}, i.e, $\|\sum_{l\in\Gamma_p} \CZ_l\|$, by $2^{-p-1}$. An important fact here is that $\mu_{p-1}$ is independent of all $\ba_{i,l}$ with $l\in\Gamma_p$ because $\mu_{p-1}$ is a function of $\{\ba_{i,k}\}_{k\in\Gamma_{s}, s< p}.$ 
Following from~\eqref{def:mup} and the assumption $\mu_p \leq 2^{-p}\mu_h$, we have
\begin{equation}\label{def:bw}
\|\bw_{i,l}\| \leq \frac{\sqrt{L}}{Q}\mu_{p-1} \leq \frac{\sqrt{L}}{Q} 2^{-p+1}\mu_h, \quad \forall l\in\Gamma_p.
\end{equation}
The proof is more or less a routine: estimate $\|\CZ_l\|_{\psi_1}$, $\sigma^2$ and apply~\eqref{thm:bern}.
For any fixed $l\in\Gamma_{p},$
\begin{align*}
 \|\CZ_l\| &
 \leq \| \bb_{i,l}\bw_{i,l}^*\left(\ba_{i,l}\ba_{i,l}^* - \I_{N_i}\right) \|
+ \sum_{j\neq i}\| \bb_{i,l}\bw_{j,l}^*\ba_{j,l}\ba_{i,l}^* \|  \\
& \leq \frac{\mum\sqrt{K_i}}{\sqrt{L}} \left[ \|\bw_{i,l}^*\left(\ba_{i,l}\ba_{i,l}^* - \I_{N_i}\right) \|
+ \sum_{j\neq i}\|\bw_{j,l}^*\ba_{j,l}\ba_{i,l}^* \| \right].\\
\end{align*}
Note that for $j\neq i$, $\bw_{j,l}^*\ba_{j,l} \sim \mathcal{N}(0, \|\bw_{j,l}\|^2)$ and $\| \ba_{i,l} \|^2 \sim \chi^2_{N_i}$. There holds
\begin{align*}
\|(|\bw_{j,l}^*\ba_{j,l}|\cdot \|\ba_{i,l}\| )\|_{\psi_1} 
& \leq C\sqrt{N_i} \| \bw_{j,l} \| 
\leq C\frac{2^{-p+1}\mu_h \sqrt{LN_i}}{Q}, \\
\|\bw_{i,l}^* (\ba_{i,l}\ba_{i,l}^*  - \I_{N_i})\|_{\psi_1} 
& \leq C\sqrt{N_i}\|\bw_{i,l}\| \leq C\frac{2^{-p+1} \mu_h \sqrt{LN_i}}{Q} 
\end{align*}
follow from~\eqref{lem:11JB},~\eqref{def:bw} and Lemma~\ref{lemma:psi}. Taking the sum over $j$, from $1$ to $r$, gives 
\begin{align*}
\|\CZ_l\|_{\psi_1} & \leq  C\frac{2^{-p+1}r\mum \mu_h\sqrt{K_i N_i}}{Q}  \leq
C\frac{2^{-p+1}r\max\{ \mum^2K, \mu^2_h N\}}{Q}
, \quad l\in\Gamma_p.
\end{align*}
Thus we have $R : = \max_{l\in\Gamma_p} \|\CZ_l\|_{\psi_1} \leq C\frac{2^{-p+1}r\max\{ \mum^2K, \mu^2_h N\}}{Q}.$
Now let's  move on to the estimation of $\sigma^2$. From~\eqref{eq:dualcz}, we have
\begin{eqnarray*}
\CZ_l^* & = & (\ba_{i,l}\ba_{i,l}^* - \I_{N_i})\bw_{i,l}\bb_{i,l}^* + \sum_{j\neq i}\ba_{i,l}\ba_{j,l}^* \bw_{j,l}\bb_{i,l}^*.
\end{eqnarray*}
The corresponding $\CZ_l^*\CZ_l$ and $\CZ_l\CZ_l^*$ have quite complicated expressions. However, all the cross terms have zero expectation, which simplifies $\E(\CZ_l^*\CZ_l)$ and $\E(\CZ_l\CZ_l^*)$ a lot. 
\begin{align*}
\E(\CZ_l^*\CZ_l ) 
& =  \E\Big(\|\bb_{i,l}\|^2 (\ba_{i,l}\ba_{i,l}^* - \I_{N_i})\bw_{i,l} \bw_{i,l}^* (\ba_{i,l}\ba_{i,l}^* - \I_{N_i})  + \|\bb_{i,l}\|^2 \sum_{j\neq i} |\bw_{j,l}^*\ba_{j,l}|^2 \ba_{i,l}\ba_{i,l}^*\Big) \\ 
& =  \|\bb_{i,l}\|^2 \left(\sum_{j=1}^r\|\bw_{j,l}\|^2\right)\I_{N_i} + \|\bb_{i,l}\|^2\bar{\bw}_{i,l}\bar{\bw}_{i,l}^*,
\end{align*}
which follows from~\eqref{lem:12JB}.
\begin{align*}
\E(\CZ_l\CZ_l^*) 
& = \E\Big( \|(\ba_{i,l}\ba_{i,l}^* - \I_{N_i})\bw_{i,l}\|^2 \bb_{i,l}\bb_{i,l}^* 
+ \sum_{j\neq i}\|\ba_{i,l}\|^2 |\lag \bw_{j,l}, \ba_{j,l}\rag|^2 \bb_{i,l}\bb_{i,l}^* \Big)\\
& = N_i\sum_{j=1}^r\|\bw_{j,l}\|^2 \bb_{i,l}\bb_{i,l}^* + \|\bw_{i,l}\|^2 \bb_{i,l}\bb_{i,l}^* 
\end{align*}
which follows from~\eqref{lem:9JB} and $\E \|\ba_{i,l}\|^2 =N_i.$

\begin{align*}
 \left\|\sum_{l\in\Gamma_p}  \E(\CZ_l^*\CZ_l)\right\|  & \leq  2 \sum_{l\in\Gamma_p}\left[\|\bb_{i,l}\|^2 \left(\sum_{j=1}^r\|\bw_{j,l}\|^2\right)\right] \\
& \quad  \leq \frac{2\mum^2K_i}{L} \sum_{j=1}^r \sum_{l\in\Gamma_p} \Tr(\BW_{j,p-1}^*\BS_{j,p} \bb_{j,l}\bb_{j,l}^*\BS_{j,p}  \BW_{j,p-1})  \\
& \quad \leq \frac{2\mum^2K_i}{L} \sum_{j=1}^r \|\BW_{j,{p-1}}\BW_{j,p-1}^*\|_*\|\BS_{j,p}\|  \\
& \quad \leq \frac{2\mum^2K_i}{L} \frac{4L}{3Q} \sum_{j=1}^r\|\BW_{j,{p-1}}\|_F^2 \\
&\quad  \leq C\frac{4^{-p+1}r\mum^2K_i}{Q}
\end{align*}
where the last inequality follows from $\|\BW_{i,p-1}\|_F \leq 2^{-p+1}$ and $\|\cdot\|_*$ is the dual norm of $\|\cdot\|$.
\begin{align*}
 \left\|\sum_{l\in\Gamma_p}  \E(\CZ_l\CZ_l^*)\right\| & =  \left\| \sum_{l\in\Gamma_p} \left[ N_i\sum_{j=1}^r\|\bw_{j,l}\|^2 \bb_{i,l}\bb_{i,l}^* + \|\bw_{i,l}\|^2 \bb_{i,l}\bb_{i,l}^*\right] \right\| \\
&  \leq \max_{j,l} \|\bw_{j,l}\|^2 \cdot \left\| \sum_{l\in\Gamma_p} \left[ rN_i\bb_{i,l}\bb_{i,l}^* + \bb_{i,l}\bb_{i,l}^*\right] \right\| \\
&  \leq  \frac{\mu^2_{p-1} L}{Q^2} \cdot 2rN_i \|\BT_{i,p}\|  =  \frac{5r\mu^2_{p-1}N_i}{2Q} \\
&  \leq  C\frac{4^{-p+1}r\mu^2_h N_i}{Q}
\end{align*}
where $\|\bw_{i,l}\| \leq \frac{\sqrt{L}\mu_{p-1}}{Q} \leq \frac{2^{-p+1}\sqrt{L}\mu_h}{Q}$. Finally we have an upper bound of $\sigma^2$:
\begin{equation*}
\sigma^2 \leq C\frac{4^{-p+1}r \max\{\mum^2K, \mu^2_h N\}}{Q}.
\end{equation*}
By using Bernstein inequality~\eqref{thm:bern} with $t = \alpha\log L$ and $\log\left(\frac{\sqrt{Q}R}{\sigma}\right) \leq C_1\log L$, we have
\begin{align*}
\left\|\sum_{l\in\Gamma_p} \CZ_l \right\| 
& \leq C_0 2^{-p+1} \max\Big\{\sqrt{ \alpha\frac{r\max\{\mu^2_{\max}K, \mu^2_h N\}}{Q} \log L }, \alpha \frac{r\max\{\mum^2K, \muh^2N\}\log^2 L}{Q} \Big\}
\end{align*}
To let $\|\sum_{l\in\Gamma_p} \CZ_l \| \leq 2^{-p+1}$ hold with probability at least $1 - L^{-\alpha}$, it suffices to let $Q \geq C_{\alpha} r\max\{\mu^2_{\max}K, \mu^2_h N\}\log^2 L.$
Then we take the union bound over all $1\leq p\leq P$ and $1\leq i\leq r$, i.e., totally $rP$ events and then
\begin{equation*}
\| \A^*_{i,p} (\blambda_{p-1}) - \BW_{i,p-1} \| \leq 2^{-p-1}
\end{equation*}
holds simultaneously for all $(p,i)$ with probability at least $1 - rPL^{-\alpha} \geq 1 - rL^{-\alpha + 1}$. To compensate the loss of probability from the union bound, we can choose $\alpha' = \alpha + \log r$, which gives $Q\geq  C_{\alpha+\log(r)}r \max\{\mum^2K, \mu^2_h N\}\log^2L$.

\end{proof}

\subsubsection{Proof of $\mu_p \leq \frac{1}{2}\mu_{p-1}$}

Recall that $\mu_p$ is defined in~\eqref{def:mup} as  $\mu_{p} = \frac{Q}{\sqrt{L}}\max_{1\leq i\leq r, l\in\Gamma_{p+1}} (\|\bb_{i,l}^*\BS_{i,p+1}\BW_{i,p} \| )$.
The goal is to show that $\mu_p \leq \frac{1}{2}\mu_{p-1}$ and thus $\mu_p \leq 2^{-p}\mu_h$ hold with high probability. 

\begin{lemma}\label{lem:mup-half}
Under the assumption of~\eqref{cond:i},~\eqref{cond:iso} and~\eqref{cond:ST} and that $\ba_{i,l}\sim\mathcal{N}(\bzero, \I_{N_i})$ independently for $1\leq i\leq r$ then 
\begin{equation*}
\mu_p \leq \frac{1}{2} \mu_{p-1},
\end{equation*}
with probability at least $1 - L^{-\alpha+1}$ if $Q\geq C_{\alpha+\log(r)}r^2 \max\{\mum^2K, \mu^2_h N\}\log^2 L$.
\end{lemma}

\begin{proof}
In order to show that $\mu_p \leq \frac{1}{2}\mu_{p-1}$, it is equivalent to prove
\begin{equation}\label{eq:muhalf-goal}
\frac{Q}{\sqrt{L}} \| \bb_{i,l}^*\BS_{i,p+1}\BW_{i,p} \| \leq \frac{1}{2} \mu_{p-1}
\end{equation}
for all $l\in\Gamma_{p+1}$ and $1\leq i\leq r$. 

From now on, we fix $i$ as well as $l\in \Gamma_{p+1}$, and show that $\frac{Q}{\sqrt{L}} \| \bb_{i,l}^*\BS_{i,p+1}\BW_{i,p} \| \leq \frac{1}{2} \mu_{p-1}$ holds with high probability. Then taking the union bound over $(i,l)$ completes the proof.
Following from~\eqref{eq:wp2} and~\eqref{def:w}, there holds
\begin{align*}
-\BW_{i,p} 
& = \underbrace{\PP_{T_i}\left(  \sum_{k\in\Gamma_{p}} \bb_{i,k}\bw^*_{i,k}(\ba_{i,k}\ba_{i,k}^* - \I_{N_i}) \right)}_{\Pi_1} 
+ \underbrace{\sum_{j\neq i}\PP_{T_i}\left( \sum_{k\in\Gamma_{p}} \bb_{i,k}\bw^*_{j,k}\ba_{j,k}\ba_{i,k}^*\right)}_{\Pi_2}.
\end{align*}
Obviously,~\eqref{eq:muhalf-goal} follows directly from the following two inequalities, 
\begin{equation}\label{eq:pi12}
 \| \bb_{i,l}^*\BS_{i,p+1}\Pi_1 \| \leq \frac{\sqrt{L}}{4Q} \mu_{p-1}, \quad  \| \bb_{i,l}^*\BS_{i,p+1}\Pi_2 \| \leq \frac{\sqrt{L}}{4Q} \mu_{p-1}.
\end{equation}

\paragraph{Step 1: Proof of $\| \bb_{i,l}^*\BS_{i,p+1}\Pi_1 \| \leq \frac{\sqrt{L}\mu_{p-1}}{4Q}$}
For a fixed $l\in\Gamma_{p+1}$,
\begin{align*}
\bb_{i,l}^*\BS_{i,p+1}\Pi_1 
& =  \sum_{k\in\Gamma_{p}} \bb_{i,l}^*\BS_{i,p+1}\Big[  \bh_i\bh_i^*\bb_{i,k} \bw^*_{i,k} (\ba_{i,k}\ba_{i,k}^* - \I_{N_i})   +  (\I_{K_i} - \bh_i\bh_i^*)\bb_{i,k}\bw^*_{i,k} (\ba_{i,k}\ba_{i,k}^* - \I_{N_i}) \bx_i\bx_i^*\Big] 
\end{align*}
where $\PP_{T_i}$ has an explicit form in~\eqref{cond:proj}.
Define
\begin{equation}\label{def:muhalf-bz}
\bz_k  :=  (\ba_{i,k}\ba_{i,k}^* - \I_{N_i}) \bw_{i,k}\bb_{i,k}^* \bh_i\bh_i^* \BS_{i,p+1}\bb_{i,l} \in\CC^{N_i}
\end{equation}
and
\begin{equation}\label{def:muhalf-z}
z_k: = \bb_{i,l}^*\BS_{i,p+1} (\I_{K_i} - \bh_i\bh_i^*)\bb_{i,k} \bw_{i,k}^* (\ba_{i,k}\ba_{i,k}^* - \I_{N_i})\bx_i.
\end{equation}
There holds
\begin{equation}\label{eq:bzz}
\|\bb^*_{i,l}\BS_{i,p+1}\Pi_1\| \leq \left\|\sum_{k\in\Gamma_p} \bz_k \right\| + \left|\sum_{k\in\Gamma_p} z_k \right|.
\end{equation}
Our goal now is to bound both $\|\sum_{k\in\Gamma_p} \bz_k \|$ and $|\sum_{k\in\Gamma_p} z_k |$ by $\frac{\sqrt{L}\mu_{p-1}}{8Q}$.
First we take a look at $\sum_{k\in\Gamma_{p}} \bz_k$. For each $k$, 
\begin{align*}
\|\bz_k\|_{\psi_1} 
& = |\bb_{i,l}^* \BS_{i,p+1}\bh_i| \cdot |\lag\bh_i, \bb_{i,k} \rag|\cdot \| (\ba_{i,k}\ba_{i,k}^* - \I_{N_i}) \bw_{i,k} \|_{\psi_1}  \\
& \leq C\frac{\sqrt{L}\mu_h}{Q}  \frac{\mu_h}{\sqrt{L}} \sqrt{N_i}\|\bw_{i,k}\| = C\frac{\mu^2_h\sqrt{N_i} \|\bw_{i,k}\|}{Q}
\end{align*}
which follows from~\eqref{cond:i} and $\| (\ba_{i,k}\ba_{i,k}^* - \I_{N_i}) \bw_{i,k} \|_{\psi_1} \leq C\sqrt{N_i}\|\bw_{i,k}\|$ in~\eqref{lem:11JB}. 
The expectation of $\E(\bz_k^*\bz_k)$ and $\E(\bz_k\bz_k^*)$ can be easily computed, 
\begin{align*}
\E(\bz_k^*\bz_k) 
& = |\bb_{i,l}^*\BS_{i,p+1}\bh_i|^2 |\bh_i^*\bb_{i,k}|^2 \E[\bw_{i,k}^*(\ba_{i,k}\ba_{i,k}^*- \I_{N_i})^2 \bw_{i,k}] \\
& =  (N_i + 1) |\bb_{i,l}^*\BS_{i,p+1}\bh_i|^2 |\bh_i^*\bb_{i,k}|^2 \|\bw_{i,k}\|^2, \\
\E(\bz_k\bz_k^*) 
& = |\bb_{i,l}^*\BS_{i,p+1}\bh_i|^2 |\bh_i^*\bb_{i,k}|^2 \\
& \qquad \E [(\ba_{i,k}\ba_{i,k}^*- \I_{N_i}) \bw_{i,k}\bw_{i,k}^*(\ba_{i,k}\ba_{i,k}^*- \I_{N_i})] \\
& = |\bb_{i,l}^*\BS_{i,p+1}\bh_i|^2 |\bh_i^*\bb_{i,k}|^2 ( \|\bw_{i,k}\|^2 \I_{N_i} + \bar{\bw}_{i,k}\bar{\bw}_{i,k}^*),
\end{align*}
which 
follows from~\eqref{lem:9JB} and~\eqref{lem:12JB}.

\begin{align*}
 \left\|\sum_{k\in\Gamma_{p}}\E(\bz_k^*\bz_k)\right\| & \leq 2N_i  |\bb_{i,l}^*\BS_{i,p+1}\bh_i|^2 \max_{k\in\Gamma_p} \{  \|\bw_{i,k}\|^2 \} \sum_{k\in\Gamma_p} |\bh_i^*\bb_{i,k}|^2 \\
& \leq \frac{2N_i L\mu^2_h}{Q^2} \max_{k\in\Gamma_p} \|\bw_{i,k}\|^2 \|\BT_{i,p}\| \\
& \leq \frac{2N_iL\mu^2_h}{Q^2} \frac{5Q}{4L} \max_{k\in\Gamma_p} \|\bw_{i,k}\|^2 \\
& = \frac{5\mu^2_hN_i \max_{k\in\Gamma_p} \|\bw_{i,k}\|^2}{2Q}.
\end{align*}
The estimation of $\|\sum_{k\in\Gamma_p} \E(\bz_k\bz_k^*)\|$ is quite similar to that of $\|\sum_{k\in\Gamma_{p}}\E(\bz_k^*\bz_k)\|$ and thus we give the result directly without going to the details,
\begin{equation*}
\left\|\sum_{k\in\Gamma_p} \E(\bz_k\bz_k^*) \right\| \leq\frac{5\mu^2_h \max_{k\in\Gamma_p} \|\bw_{i,k}\|^2}{2Q}.
\end{equation*}
Therefore, 
\begin{equation*}
R := \max_{k\in\Gamma_p} \|\bz_k\|_{\psi_1} \leq C \frac{\mu^2_h \sqrt{N}}{Q} \max_{k\in\Gamma_p} \|\bw_{i,k}\|
\end{equation*}
and similarly, we have 
\begin{equation*}
\sigma^2 \leq C\frac{\mu^2_h N\max_{k\in\Gamma_p} \|\bw_{i,k}\|^2}{Q}.
\end{equation*}
Then we just apply~\eqref{thm:bern} with $t = \alpha\log L$ and $\log(\sqrt{Q}R/\sigma) \leq C_1\log L$ to estimate $\|\sum_{k\in\Gamma_p} \bz_k\|$, 
\begin{align}\label{ineq:bz}
\begin{split}
\left\|\sum_{k\in\Gamma_p} \bz_k\right\|
& \leq C \max_{k\in\Gamma_p} \|\bw_{i,k}\|^2   \max\left\{ \sqrt{\frac{\alpha \mu^2_h N}{Q} \log L}, 
\frac{\alpha \mu^2_h \sqrt{N}}{Q}\log^2 L) \right\}.
\end{split}
\end{align}
Note that $\max_{k\in\Gamma_p} \|\bw_{i,k}\| \leq \frac{\sqrt{L}\mu_{p-1}}{Q}$ in~\eqref{def:mup-1} and thus it suffices to let $Q \geq C_{\alpha}\mu^2_h N \log^2 L$  to ensure that  $\|\sum_{k\in\Gamma_p} \bz_k\| \leq \frac{\sqrt{L}\mu_{p-1}}{8Q}$ holds with probability at least $1 - L^{-\alpha}.$

\bigskip

Concerning $z_k$ in~\eqref{def:muhalf-z}, we first estimate $\|z_k\|_{\psi_1}$:
\begin{align*}
\|z_k\|_{\psi_1} 
& = |\bb_{i,l}^*\BS_{i,p+1} (\I_{K_i} - \bh_i\bh_i^*)\bb_{i,k}| \| \bw_{i,k}^* (\ba_{i,k}\ba_{i,k}^* - \I_{N_i})\bx_i\|_{\psi_1} \\
& = \|\bb_{i,l}\| \|\BS_{i,p+1}\| \|\bb_{i,k}\| \| \bw_{i,k}^* (\ba_{i,k}\ba_{i,k}^* - \I_{N_i})\bx_i\|_{\psi_1} \\
& \leq C\frac{\mum^2K_i}{L}\cdot \frac{4L}{3Q} \|\bw_{i,k}\|  \leq C\frac{\mum^2K_i}{Q} \|\bw_{i,k}\|
\end{align*}
where $\| \bw_{i,k}^* (\ba_{i,k}\ba_{i,k}^* - \I_{N_i})\bx_i\|_{\psi_1}  \leq C\|\bw_{i,k}\| $ in~\eqref{lem:13JB}. Thus $R: = \max\{\| \bz_k \|_{\psi_1}\} \leq C\frac{\mum^2K_i}{Q} \max_{k\in\Gamma_p}\|\bw_{i,k}\|.$
Furthermore,
\begin{align*}
\E |z_k|^2 & = |\bb_{i,l}^*\BS_{i,p+1} (\I_{K_i} - \bh_i\bh_i^*)\bb_{i,k}|^2 \E\left[ \bw_{i,k}^* (\ba_{i,k}\ba_{i,k}^* - \I_{N_i})\bx_i\bx_i^*(\ba_{i,k}\ba_{i,k}^* - \I_{N_i}) \bw_{i,k}\right]\\
& = |\bb_{i,l}^*\BS_{i,p+1} (\I_{K_i} - \bh_i\bh_i^*)\bb_{i,k}|^2 \bw_{i,k}^*(\I_{N_i} + \bx_i\bx_i^*)\bw_{i,k} 
\end{align*}
which 
follows from~\eqref{lem:12JB}.
The variance $\sum_{k\in\Gamma_p} z_k$ is bounded by
\begin{align*}
\sigma^2 
& \leq \bb_{i,l}^* \BS_{i,p+1} (\I_{K_i} - \bh_i\bh_i) \BT_{i,p} (\I_{K_i} - \bh_i\bh_i^*) \BS_{i,p+1}\bb_{i,l}  \cdot \max_{k\in\Gamma_p} \bw_{i,k}^*(\I_{N_i} + \bx_i\bx_i^*)\bw_{i,k} \\
& \leq \|\bb_{i,l}\|^2 \|\BS_{i,p+1}\|^2 \|\BT_{i,p}\|\max_{k\in\Gamma_p} \bw_{i,k}^*(\I_{N_i} + \bx_i\bx_i^*)\bw_{i,k} \\
& \leq 2\frac{\mum^2K_i}{L} \frac{16L^2}{9Q^2} \frac{5Q}{4L} \max_{k\in\Gamma_p}\|\bw_{i,k}\|^2 \\
& =  \frac{40\mum^2K_i}{9Q} \max_{k\in\Gamma_p}\|\bw_{i,k}\|^2.
\end{align*}
Similar to what we have done in~\eqref{ineq:bz},
\begin{align}\label{ineq:z}
\begin{split}
\left|\sum_{k\in\Gamma_p} z_k\right|
& \leq C \max_{k\in\Gamma_p} \|\bw_{i,k}\|^2   \max\left\{ \sqrt{\frac{\alpha \mum^2 K}{Q} \log L}, 
\frac{\alpha \mum^2K}{Q}\log^2 L \right\}.
\end{split}
\end{align}
Note that $\max_{k\in\Gamma_p} \|\bw_{i,k}\| \leq \frac{\sqrt{L}\mu_{p-1}}{Q}$ and thus $Q \geq C_{\alpha}\mum^2 K \log^2 L$  guarantees that  $|\sum_{k\in\Gamma_p} z_k| \leq \frac{\sqrt{L}\mu_{p-1}}{8Q}$ holds with probability at least $1 - L^{-\alpha}.$
Combining~\eqref{ineq:bz} and~\eqref{ineq:z} gives 
\begin{align}\label{ineq:pi1}
\begin{split}
& \Pr\left( \| \bb_{i,l}^*\BS_{i,p}\Pi_1 \| \geq \frac{\sqrt{L}\mu_{p-1}}{4Q} \right) 
 \leq \Pr\left(\|\sum_{k\in\Gamma_p}\bz_k\| \geq \frac{\sqrt{L}\mu_{p-1}}{8Q}\right) + \Pr\left(|\sum_{k\in\Gamma_p}z_k| \geq \frac{\sqrt{L}\mu_{p-1}}{8Q}\right) 
\leq 2L^{-\alpha},
\end{split}
\end{align}
if $Q\geq C_{\alpha}\max\{ \mum^2K, \mu^2_h N \}\log^2L.$

\paragraph{Step 2: Proof of $\| \bb_{i,l}^*\BS_{i,p+1}\Pi_2 \| \leq \frac{\sqrt{L}\mu_{p-1}}{4Q}$}
For any fixed $l\in\Gamma_{p+1}$, 
\begin{eqnarray*}
\bb_{i,l}^*\BS_{i,p+1} \Pi_2 
 = \bb_{i,l} ^*\BS_{i,p+1} \sum_{j\neq i} \PP_{T_i}\left( \sum_{k\in\Gamma_{p}} \bb_{i,k}\bw_{j,k}^*\ba_{j,k}\ba_{i,k}^*\right).
\end{eqnarray*}
Now we rewrite $\bb_{i,l}^*\BS_{i,p+1} \Pi_2 $ into
\begin{equation*}
\bb_{i,l}^*\BS_{i,p+1} \Pi_2  =  \sum_{j\neq i}\left(\sum_{k\in\Gamma_p} \bz_{j,k}^* + z_{j,k} \bx_i^*\right) 
\end{equation*}
where
\begin{align}
\bz_{j,k} & :=   \bb_{i,l}^*\BS_{i,p+1} \bh_i\bh_i^* \bb_{i,k}\bw_{j,k}^* \ba_{j,k}\ba_{i,k},  \\
z_{j,k} & := \bb_{i,l}^*\BS_{i,p+1}  (\I_{K_i} - \bh_i\bh_i^*)\bb_{i,k}\bw_{j,k}^* \ba_{j,k}\ba_{i,k}^* \bx_i. \label{eq:muhalf-zjk}
\end{align}
By the triangle inequality,
\begin{equation}\label{eq:Pi2}
\| \bb_{i,l}^*\BS_{i,p+1}\Pi_2\| \leq \sum_{j\neq i, j\leq r}\left[\|\sum_{k\in\Gamma_p} \bz_{j,k}\| + |\sum_{k\in\Gamma_p} z_{j,k} | \right].
\end{equation}
To bound $\| \bb_{i,l}^*\BS_{i,p+1}\Pi_2\|$ by $\frac{\sqrt{L}\mu_{p-1}}{4Q}$, it suffices to prove that for all $1\leq j\leq r$, 
\begin{equation}\label{eq:zjk-bd}
\left\|\sum_{k\in\Gamma_p} \bz_{j,k}\right\| \leq \frac{\sqrt{L}\mu_{p-1}}{8rQ}, \quad |\sum_{k\in\Gamma_p} z_{j,k} | \leq \frac{\sqrt{L}\mu_{p-1}}{8rQ}.
\end{equation}
For $\sum_{k\in\Gamma_p} \bz_{j,k}$, 
\begin{align*}
\|\bz_{j,k}\|_{\psi_1} 
& \leq |\bb_{i,l}^*\BS_{i,p+1}\bh_i||\bh_i^* \bb_{i,k}| (|\bw_{j,k}^* \ba_{j,k}| \cdot \|\ba_{i,k}^*\|)_{\psi_1} \\
& \leq C\frac{\sqrt{L}\mu_h}{Q} \frac{\mu_h}{\sqrt{L}} \sqrt{N_i}\|\bw_{j,k}\| \\
& \leq C\frac{\muh^2 \sqrt{N_i} \max_{k\in\Gamma_p} \|\bw_{j,k}\|}{Q}.
\end{align*}
where $(|\bw_{j,k}^* \ba_{j,k}| \cdot \|\ba_{i,k}^*\|)_{\psi_1} \leq C\sqrt{N_i}\|\bw_{j,k}\|$ follows from Lemma~\ref{lemma:psi}.
Now we move on to the estimation of $\sigma^2.$

\begin{align*}
\left\| \sum_{k\in\Gamma_p} \E \bz_{j,k}^*\bz_{j,k} \right\| & = \sum_{k\in\Gamma_p} \Big( | \bb_{i,l}^*\BS_{i,p+1} \bh_i\bh_i^* \bb_{i,k} |^2 \E\left( |\bw_{j,k}^* \ba_{j,k}|^2 \|\ba_{i,k}\|^2\right) \Big) \\
& =  N_i\sum_{k\in\Gamma_p} | \bb_{i,l}^*\BS_{i,p+1} \bh_i|^2|\bh_i^* \bb_{i,k} |^2 | \|\bw_{j,k}\|^2 \\
& \leq N_i \frac{L\mu^2_h}{Q^2} \max_{k\in\Gamma_p} \|\bw_{j,k}\|^2 \|\BT_{i,p}\| \\
& \leq \frac{5\mu^2_h N_i \max_{k\in\Gamma_p} \|\bw_{j,k}\|^2}{4Q}
\end{align*}
and similarly, 
\begin{equation*}
\left\| \sum_{k\in\Gamma_p} \E \bz_{j,k}\bz_{j,k}^* \right\| 
 \leq  \frac{5\mu^2_h\max_{k\in\Gamma_p} \|\bw_{j,k}\|^2}{4Q}.
\end{equation*}
Thus $\sigma^2 \leq C\frac{\mu^2_hN_i\max_{k\in\Gamma_p} \|\bw_{j,k}\|^2}{Q}$.
By applying Bernstein inequality~\eqref{thm:bern}, we have
\begin{align*}
\begin{split}
\left\|\sum_{k\in\Gamma_p} \bz_{j,k}\right\| 
& \leq 
C \max_{k\in\Gamma_p} \|\bw_{j,k}\|^2 \max\left\{ \sqrt{\frac{\alpha \muh^2 N}{Q} \log L}, 
\frac{\alpha \muh^2N}{Q}\log^2 L \right\}
\end{split}
\end{align*}
where $\max_{k\in\Gamma_p} \|\bw_{j,k}\| \leq \frac{\sqrt{L}\mu_{p-1}}{Q}.$

Choosing $Q\geq C_{\alpha}r^2\mu^2_h N\log^2L$ leads to
\begin{equation}\label{eq:bz-jk}
\left\|\sum_{k\in\Gamma_p} \bz_{j,k}\right\| \leq   \frac{\sqrt{L}\mu_{p-1}}{8rQ}
\end{equation}
with probability at least $1 - L^{-\alpha}$ for a fixed $j:1\leq j\leq r.$

\bigskip

For $\sum_{k\in\Gamma_p} z_{j,k}$ defined in~\eqref{eq:muhalf-zjk} and fixed $j$,
\begin{align*}
R: & = \max_{k\in\Gamma_p} |z_{j,k}| \leq \max_{k\in\Gamma_p}\| \bb_{i,l}^*\BS_{i,p+1}  (\I_{K_i} - \bh_i\bh_i^*)\bb_{i,k} \|  \cdot \max_{k\in\Gamma_p}\|\bw_{j,k}^* \ba_{j,k}\ba_{i,k}^* \bx_i\|_{\psi_1} \\
& \leq  C\frac{\mu^2_{\max}K_i}{L} \frac{4L}{3Q} \max_{k\in\Gamma_p}\|\bw_{j,k}\| = C\frac{\mu^2_{\max}K_i\max_{k\in\Gamma_p}\|\bw_{j,k}\|}{Q} \\
\end{align*} 
where $\|\bw_{j,k}^* \ba_{j,k}\ba_{i,k}^* \bx_i\|_{\psi_1}  \leq C\|\bw_{j,k}\|$ follows from Lemma~\ref{lemma:psi}.
Now we proceed to compute the  variance by
\begin{align*}
\sigma^2 & : = \sum_{k\in\Gamma_p}\E |z_{j,k}|^2  =  \sum_{k\in\Gamma_p} |\bb_{i,l}^*\BS_{i,p+1}  (\I_{K_i} - \bh_i\bh_i^*)\bb_{i,k}|^2 \|\bw_{j,k}\|^2 \\
& \leq \max_{k\in\Gamma_p}\|\bw_{j,k}\|^2  \|\BS_{i,p+1}\|^2 \|\BT_{i,p}\| \|\bb_{i,l}\|^2 \\
& \leq \max_{k\in\Gamma_p}\|\bw_{j,k}\|^2 \frac{16L^2}{9Q^2} \frac{5Q}{4L} \frac{\mum^2K_i}{L} \\
& \leq C\frac{\max_{k\in\Gamma_p}\|\bw_{j,k}\|^2 \mum^2K_i}{Q}.
\end{align*}
Then we apply the Bernstein inequality to get an upper bound of $|\sum_k z_{j,k}|$ for fixed $j$,
\begin{align*}
\left|\sum_{k\in\Gamma_p} z_{j,k}\right| 
& \leq C \max_{k\in\Gamma_p} \|\bw_{j,k}\|^2 \max\Big\{ \sqrt{\frac{\alpha \mum^2 K}{Q} \log L}, \frac{\alpha \mum^2K}{Q}\log^2 L) \Big\} \leq \frac{\sqrt{L}\mu_{p-1}}{8rQ}
\end{align*}
with probability $1 - L^{-\alpha}$ if $Q\geq C_{\alpha} r^2\mu^2_{\max}K \log^2 L$. 
Thus, combined with~\eqref{eq:bz-jk}, we have proven that for fixed $j$, 
\begin{equation*}
\left\|\sum_{k\in\Gamma_p} \bz_{j,k} \right\| + \left|\sum_{k\in\Gamma_p}z_{j,k}\right| \leq \frac{\sqrt{L}\mu_{p-1}}{4rQ}
\end{equation*}
holds with probability at least $1 - 2L^{-\alpha}.$ By taking the union bound over $1\leq j\leq r$ and using~\eqref{eq:Pi2}, we can conclude that 
\begin{equation*}
\|\bb_{i,l}^*\BS_{i,p+1}\Pi_2\| \leq \frac{\sqrt{L}\mu_{p-1}}{4Q}
\end{equation*}
with probability $1 - rL^{-\alpha}$ if $Q\geq C_{\alpha} r^2\mu^2_{\max}K \log^2 L$.

\paragraph{Final step: Proof of~\eqref{eq:muhalf-goal}}

To sum up, we have already shown that for fixed $i$ and $l\in\Gamma_p$, 
\begin{align*}
& \frac{Q}{\sqrt{L}} \| \bb_{i,l}^*\BS_{i,p+1}\BW_{i,p} \| 
\leq  \left\|\sum_{k\in\Gamma_p} \bz_k \right\| + \left|\sum_{k\in\Gamma_p} z_k \right| + \sum_{j\neq i}\left[\left\|\sum_{k\in\Gamma_p} \bz_{j,k}\right\| + \left|\sum_{k\in\Gamma_p} z_{j,k} \right| \right] \leq \frac{1}{2}\mu_{p-1}
\end{align*}
with probability at least $1 - (r+ 2)L^{-\alpha}$  if $Q\geq C_{\alpha}r^2 \max\{\mu^2_{\max}K, \mu^2_h N\}\log^2L$. Then we take union bound over all $1\leq i\leq r$ and $l\in\Gamma_p$ and $1\leq p\leq P$ and obtain
\begin{align*}
& \Pr\left(\frac{Q}{\sqrt{L}}\max_{i,l,p} \| \bb_{i,l}^*\BS_{i,p+1}\BW_{i,p} \|  \leq \frac{1}{2}\mu_{p-1} \right)  \geq 1 - r(r+2)PQL^{-\alpha} = 1 - r(r+2)L^{-\alpha+1}.
\end{align*}
If we choose a slightly larger $\alpha$ as $\tilde{\alpha} = \alpha + 2\log r$, i.e., $Q\geq C_{\alpha+2\log(r)}r^2 \max\{\mu^2_{\max}K, \mu^2_h N\}\log^2L$, then $\mu_{p} \leq \frac{1}{2}\mu_{p-1}$ holds for all $p$ with probability at least $1 - L^{-\alpha+1}.$ 
\end{proof}

\subsection{Proof of Theorem~\ref{thm:main} \label{mainproof}}

We now assemble the various intermediate and auxiliary results to establish Theorem~\ref{thm:main}.
We recall that Theorem~\ref{thm:main} follows immediately from Lemma~\ref{lemma:suffcond}, which in turn
hinges on the validity of the conditions~\eqref{cond:suffcond} and~\eqref{cond:suff2}. Let us focus
on condition~\eqref{cond:suffcond} first, i.e., we need to show that
\begin{gather}
\max_{1\leq i\leq r} \, \|\PP_{T_i}\A_i^*\A_i \PP_{T_i} - \PP_{T_i}\| \leq \frac{1}{4}, \label{cond:suffcond1} \\
  \max_{1\leq j\neq k \leq r}\, \|\PP_{T_j}\A_j^*\A_k\PP_{T_k}\|  \leq \frac{1}{4r},  \label{cond:suffcond2} \\
\max_{1\leq i\leq r} \, \|\A_i\| \leq \gamma. \label{cond:suffcond3}
\end{gather}
Under the assumptions of Theorem~\ref{thm:main}, Proposition~\ref{prop:lcisop} ensures that 
condition~\eqref{cond:suffcond1} holds with probability at least $1-L^{-\alpha+1}$ if
$Q\geq C_{\alpha+\log(r)}\max\{\mum^2K, \mu^2_h N\}\log^2 L$ where $K := \max K_i$ and $N := \max N_i.$ 
Moving on to the incoherence condition~\eqref{cond:suffcond2}, Proposition~\ref{prop:mix} implies that this condition 
holds with probability at least $1 - L^{-\alpha + 1}$ if $Q\geq C_{\alpha+\log(r)}r^2 \max\{\mum^2K, \mu^2_h N\}\log^2 L$. Furthermore, $\gamma$ in condition~\eqref{cond:suffcond3} is bounded by $\sqrt{N(\log NL/2) + \alpha \log L}$ with probability $1 - rL^{-\alpha}$ according to Lemma 1 in~\cite{RR12}.
We now turn our attention to condition~\eqref{cond:suff2}. Under the assumption that properties~\eqref{eq:lcisop}
and~\eqref{eq:mix} hold, Lemma~\ref{lem:Y-hx} implies the first part of condition~\eqref{cond:suff2}. 
The two properties \eqref{eq:lcisop} and~\eqref{eq:mix} have been established in Propositions~\ref{prop:lcisop}
and~\ref{prop:mix}, respectively. The second part of the approximate dual certificate condition
in~\eqref{cond:suff2} is established in Lemma~\ref{lem:normbound} with the aid of Lemma~\ref{lem:mup-half}, with probability at least $1 - 2L^{-\alpha+1}$ if $Q\geq C_{\alpha+\log(r)}r^2 \max\{ \mum^2 K, \mu^2_h N \}\log^2L$.

By ``summing up'' all the probabilities of failure in each substep,
\begin{equation*}
\Pr( \hBX_i = \BX_i, \forall 1\leq i\leq r)  \geq 1 - 5L^{-\alpha + 1}
\end{equation*}
if $Q\geq C_{\alpha+\log(r)}r^2 \max\{ \mum^2 K, \mu^2_h N \}\log^2 L$. Since $L = PQ$ and $P$ is chosen to be greater than $\log_2(5r\gamma)$, it suffices to let $L$ yield:
\begin{equation*}
L \geq C_{\alpha+\log(r)}r^2 \max\{ \mum^2 K, \mu^2_h N \}\log^2 L \log\gamma
\end{equation*}
with $\gamma \leq \sqrt{N\log(NL/2) + \alpha\log L}.$
Thus, the sufficient conditions stated in Lemma~\ref{lemma:suffcond} are fulfilled with probability at least
$1 - {\cal O}(L^{-\alpha+1})$, hence Theorem~\ref{thm:main} follows now directly from Lemma~\ref{lemma:suffcond}.

\subsection{Stability theory -- Proof of Theorem~\ref{thm:noise}\label{s:stability}}

Since we do not assume $\{\BX_i\}_{i=1}^r$  are of the same size, notation will be an issue during the discussion. We introduce a few notations in order to make the derivations easier.
Recall $\sum_{i=1}^r\A_i(\BZ_i)$ is actually a linear mapping from $\CC^{K_1\times N_1} \oplus\cdots \oplus \CC^{K_r\times N_r}$ to $\CC^L.$ This linear operator can be easily written into matrix form:
define $\BPhi := \left[\BPhi_1 | \cdots |\BPhi_r\right]$ with $\BPhi_i\in \CC^{L\times K_iN_i } $ and $\BPhi \in \CC^{ L\times (\sum_{i=1}^r K_iN_i)}$ as
\begin{equation*}
\BPhi_i\VEC(\BZ_i) := \VEC(\A_i(\BZ_i)); \BPhi
\begin{bmatrix}
\VEC(\BZ_1) \\
\vdots \\
 \VEC(\BZ_r) 
\end{bmatrix}
: =  \VEC(\sum_{i=1}^r\A_i(\BZ_i))
\end{equation*}
where $\BZ_i \in\CC^{K_i\times N_i}.$ The operation ``$\VEC$" vectorizes a matrix into a column vector. $\BPhi$ and $\BPhi_i$ are well-defined here. 
It could be be shown by slightly modifying the proof of Lemma 2 in~\cite{RR12} that
\begin{equation*}
\BPhi\BPhi^* = \sum_{i=1}^r\BPhi_i\BPhi_i^*\in\CC^{L\times L}
\end{equation*}
is well conditioned, which means the largest and smallest eigenvalues of $\BPHI\BPHI^*$, denoted by $\lambda_{\max}^2$ and $\lambda_{\min}^2$ respectively, are of the same scale. More precisely,
\begin{equation}\label{eq:BPh-cm}
0.48 \mu^2_{\min}\frac{\sum_{i=1}^rK_iN_i}{L} \leq \lambda_{\min}^2 \leq \lambda_{\max}^2 \leq 4.5\mu^2_{\max}\frac{\sum_{i=1}^r K_iN_i}{L}
\end{equation}
with probability at least $1 - O(L^{-\alpha + 1})$ if $\sum_{i=1}^r K_iN_i \geq \frac{C_{\alpha}}{\mu^2_{\min}}L\log^2L$ with $\mumin^2$ defined in~\eqref{def:mumax}. Note that $\sum_{i=1}^rK_i N_i$ is usually much larger than $L$ in applications. 

\bigskip

Let $\BE_i = \hBX_i - \BX_{i}\in\CC^{K_i\times N_i}, 1\leq i\leq r$ be the difference between $\hBX_i$ and $\BX_{i}$.  Define 
\begin{equation*}
\be_i := \VEC(\BE_i), \quad \be : = 
\begin{bmatrix}
\be_1 \\
\vdots \\
\be_r 
\end{bmatrix}
\in\CC^{(\sum_{i=1}^rK_iN_i)\times 1},
\end{equation*}
where $\be$ is a long vector consisting of all $\be_i$, $1\leq i\leq r$. 
We also consider $\be$ being projected on $\Ran(\BPhi^*)$, denoted by $\be_{\BPhi}$, 
\begin{equation*}
\be_{\BPhi} : = \BPhi^*(\BPhi\BPhi^*)^{-1}\BPhi \be
\end{equation*}
where $\BPhi\be = \sum_{i=1}^r\BPhi_i\be_i  =  \sum_{i=1}^r\A_i(\BE_i)$.  From~\eqref{cvx:noise},  we know that 
\begin{align}\label{eq:2eta}
\begin{split}
\|\BPhi \be\|_F 
& = \left\|\sum_{i=1}^r \A_i(\BE_i)\right\|_F\\
&  \leq \left\|\sum_{i=1}^r \A_i(\hat{\BX}_i ) - \hby\right\|_F + \left\|\sum_{i=1}^r \A_i(\BX_i) - \hby\right\|_F \leq 2\eta
\end{split}
\end{align}
since both $\{\hBX_i\}_{i=1}^r$ and $\{\BX_i\}_{i=1}^r$ are inside the feasible set.
Similarly, define $\be_{\BPHIB} := \be - \be_{\BPHI}\in\Null(\BPHI)$ and denote $\BH_i\in\CC^{K_i\times N_i}$ and $\BJ_i\in\CC^{K_i\times N_i}$, $1\leq i\leq r,$ as matrices satisfying
\begin{equation}\label{def:HJ}
\be_{\BPHIB} := 
\begin{bmatrix}
\VEC(\BH_1) \\
\vdots \\
\VEC(\BH_r) 
\end{bmatrix},
\quad
\be_{\BPHI} := 
\begin{bmatrix}
\VEC(\BJ_1) \\
\vdots \\
\VEC(\BJ_r) 
\end{bmatrix}
\end{equation}
where $\sum_{i=1}^r\A_i(\BH_i)= \BPHI\be_{\BPHIB}= \bzero$ and $\BH_i + \BJ_i = \BE_i$ follows from the definition of $\BH_i$ and $\BJ_i.$  

Define $\BP_{T_i}$ as the projection matrix from $\VEC(\BZ)$ to $\VEC(\PP_{T_i}(\BZ))$, as
\begin{equation*}
\BP_{T_i}\VEC(\BZ) = \VEC(\PP_{T_i}(\BZ)), \quad \BP_{T_i} \in\CC^{(K_iN_i)\times (K_iN_i)}
\end{equation*}
and 
\begin{align*}
\BP_T : = 
\begin{bmatrix}
\BP_{\BT_1} & \cdots & \bzero \\
\vdots & \ddots &\vdots \\
\bzero & \cdots & \BP_{\BT_r}
\end{bmatrix}, 
\BP_{\TB} : = 
\begin{bmatrix}
\I_{K_1N_1} - \BP_{\BT_1} &\cdots & \bzero \\
\vdots & \ddots & \vdots \\
\bzero & \cdots &\I_{K_rN_r}  -  \BP_{\BT_r}
\end{bmatrix}.
\end{align*}
Actually the definitions above immediately give the following equations:\begin{align}\label{eq:orth}
\begin{split}
\BP_T \be = 
\begin{bmatrix}
\BP_{T_1} \be_1 \\
\vdots \\
\BP_{T_r} \be_r 
\end{bmatrix}, \quad
\BP_T \be_{\BPHIB} = 
\begin{bmatrix}
\VEC(\BH_{1,T_1}) \\
\vdots \\
\VEC(\BH_{r,T_r}) 
\end{bmatrix}, \quad \BP_{\TB} \be_{\BPHIB} = 
\begin{bmatrix}
\VEC(\BH_{1,\TB_1}) \\
\vdots \\
\VEC(\BH_{r,\TB_r}) 
\end{bmatrix}.
\end{split}
\end{align}

We will prove that if the observation $\hby$ is contaminated by noise, the minimizer $\hat{\BX}_i$ to the convex program~\eqref{cvx:noise} yields, 
\begin{equation*}
\|\be\|  \leq C \frac{r\lambda_{\max}\sqrt{\max\{K, N\}}}{\lambda_{\min} (1 - \beta - 2r\gamma\alpha)}\eta.
\end{equation*}

\begin{proof}
The proof basically follows similar arguments as~\cite{RR12,CanPlan10}.
First we decompose $\be$ into several linear subspaces. By using orthogonality and Pythagorean Theorem,
\begin{equation}\label{eq:pty}
\|\be\|_F^2 = \| \be_{\BPhi} \|^2 + \|\BP_T\be_{\BPHIB} \|_F^2 + \|\BP_{\TB}\be_{\BPHIB} \|_F^2.
\end{equation}
Following from~\eqref{eq:orth},~\eqref{lowbd} and~\eqref{upperbd} gives an estimate of the second term in~\eqref{eq:pty},
\begin{align*}
\|\BP_T \be_{\BPHIB}\|^2_F  & = \sum_{i=1}^r\|\BH_{i,T_i}\|_F^2  \leq 2 \left\| \sum_{i=1}^r \A_i(\BH_{i, T_i}) \right\|_F^2 \\
& = 2 \left\| \sum_{i=1}^r \A_i(\BH_{i, \TB_i}) \right\|_F^2 \leq 2\gamma^2\left(\sum_{i=1}^r \| \BH_{i, \TB_i} \|_F\right)^2  \\
&  \leq 2r\gamma^2 \sum_{i=1}^r \| \BH_{i, \TB_i} \|_F^2  \leq 2r \gamma^2 \|  \BP_{\TB}\be_{\BPHIB}\|_F^2\\
&  \leq  2r\lambda_{\max}^2  \|  \BP_{\TB}\be_{\BPHIB}\|_F^2
\end{align*}
where $\max \|\A_i\| \leq \gamma$, $\lambda^2_{\max}$ is largest eigenvalue of $\BPHI\BPHI^*$ and obviously $\gamma \leq \lambda_{\max}$. The second equality holds since $\sum_{i=1}^r\A_i(\BH_i) = \bzero.$
For the third term in~\eqref{eq:pty}, by reversing the arguments in the proof of Lemma~\ref{lemma:suffcond}, we have
\begin{align*}
\|  \BP_{\TB}\be_{\BPHIB}\|_F 
& = \sqrt{ \sum_{i=1}^r\| \BH_{i, \TB_i}\|^2_F} \leq  \sum_{i=1}^r\| \BH_{i, \TB_i}\|_F\\
 & \leq \frac{1}{1 - \beta - 2r\gamma\alpha} \sum_{i=1}^r \lag \BH_{i}, \bh_i\bx_i^*\rag + \| \BH_{i, \TB_i} \|_* \\
& \leq \frac{1}{1 - \beta - 2r\gamma \alpha} \sum_{i=1}^r \left[ \|\BX_{i} + \BH_i \|_* - \|\BX_{i}\|_* \right] \\
& \leq \frac{1}{1 - \beta - 2r\gamma \alpha} \sum_{i=1}^r \left[ \|\BX_{i} + \BH_i \|_* - \|\hBX_{i}\|_* \right]
\end{align*}
where the first equality comes from~\eqref{eq:orth}, the third inequality is due to Lemma~\ref{lem:suff-1st} and the last inequality follows from $\sum_{i=1}^r\|\hBX_i\|_*  \leq \sum_{i=1}^r \|\BX_{i}\|_*$ in~\eqref{inq:noiseobj}. From the definition of $\BH_i$ and $\BJ_i$ in~\eqref{def:HJ}, $\hBX_i = \BX_{i} + \BE_i = \BX_{i} + \BH_i + \BJ_i$ and triangle inequality gives,
\begin{align*}
\|\BP_{\TB} \be_{\BPHIB}\|_F 
& \leq \frac{1}{1 - \beta - 2r\gamma \alpha} \sum_{i=1}^r \|\BJ_i\|_*   \leq \frac{\sqrt{\max\{K, N\}}}{1 - \beta - 2r\gamma \alpha} \sum_{i=1}^r \|\BJ_i\|_F.
\end{align*}
In other words,
\begin{align}\label{ineq:P2}
\begin{split}
\|\BP_{\TB}\be_{\BPHIB}\|^2_F  
& \leq \frac{r\max\{K, N\}}{(1 - \beta - 2r\gamma \alpha)^2} \sum_{i=1}^r \|\BJ_i\|_F^2   \leq \frac{r\max\{K, N\}}{(1 - \beta - 2r\gamma \alpha)^2}   \|\be_{\BPhi}\|_F^2
\end{split}
\end{align}
where $\|\be_{\BPhi}\|_F^2 = \sum_{i=1}^r \|\BJ_i\|_F^2$ follows from~\eqref{def:HJ}.

By combining all those estimations together, i.e., $\|\BP_T \be_{\BPHIB}\|^2_F \leq 4r\lambda_{\max}^2  \|  \BP_{\TB}\be_{\BPHIB}\|_F^2$,~\eqref{ineq:P2} and~\eqref{eq:pty}, we arrive at
\begin{eqnarray*}
\|\be\|_F^2 & \leq  \|\be_{\BPhi}\|_F^2 + (2r\lambda_{\max}^2 + 1)\|\BP_{\TB}\be_{\BPHIB}\|_F^2 \leq C \frac{r^2\lambda_{\max}^2\max\{K, N\}}{(1 - \beta - 2r\gamma\alpha)^2} \|\be_{\BPHI}\|_F^2.
\end{eqnarray*}
Note that $\be_{\BPhi} : = \BPhi^*(\BPhi\BPhi^*)^{-1}\BPhi \be$, 
\begin{equation*}
\|\be_{\BPHI}\|_F \leq \frac{1}{\lambda_{\min} } \|\BPHI\be\|_F
\end{equation*}
where $\lambda_{\min}^2$ is the smallest eigenvalue of $\BPHI\BPHI^*$.
By applying $\|\BPHI\be\| \leq 2\eta$ in~\eqref{eq:2eta}, we have
\begin{align*}
\|\be\|_F & \leq 
C \frac{r\lambda_{\max}\sqrt{\max\{K, N\}}}{\lambda_{\min} (1 - \beta - 2r\gamma\alpha)} \|\BPHI\be\|_F \leq C \frac{r\lambda_{\max}\sqrt{\max\{K, N\}}}{\lambda_{\min} (1 - \beta - 2r\gamma\alpha)}\eta.
\end{align*}
In particular, if we choose $\alpha = (5r\gamma)^{-1}$ and $\beta = \frac{1}{2}$ according to Lemma~\ref{lemma:suffcond}, then $\frac{1}{1- \beta - 2r\gamma \alpha} = 10$.
This completes the proof of Theorem~\ref{thm:noise}.
\end{proof}

\section{Conclusion \label{s:conclusion}}

We have developed a theoretical and numerical framework for simultaneously blindly deconvolve and demix multiple transmitted signals from just one received signal.
The reconstruction of the transmitted signals and the impulse responses can be accomplished by solving a semidefinite program. 
Our findings  are of interest for a variety of applications, in particular for the area of multiuser wireless communications. Our theory provides a bound for
the number of measurements needed to guarantee successful recovery. While this bound scales quadratically in the number of unknown signals, it seems that our theory is somewhat pessimistic. Indeed, numerical experiments indicate, surprisingly, that the proposed algorithm succeeds already even if the number of measurements is fairly close to the theoretical limit with respect to the number of degrees of freedom. It would be very desirable to develop a theory that can explain this remarkable phenomenon.

Hence, this paper does not only provide answers, but it also triggers several follow-up questions. Some key questions are: (i)~Can we derive a theoretical bound that scales linearly in $r$, rather than
quadratic in $r$ as our current theory? (ii)~Is it possible to develop satisfactory theoretical bounds for deterministic matrices $\BA_i$?
 (iii)~Do there exist faster numerical algorithms that do not need to resort to solving a semidefinite program 
 (say in the style of the phase retrieval Wirtinger-Flow algorithm~\cite{CLS14}) with provable performance guarantees?\footnote{After completion of this manuscript we have developed such algorithms for both $r=1$ and $r > 1$, see~\cite{LLSW16} and~\cite{LS17} respectively. However, the sampling complexity in~\cite{LS17} is still sub-optimal, i.e., $L \gtrsim r^2(K+N)$ instead of $L\gtrsim r(K+N)$.}
(iv)~Can we develop a theoretical framework where the signals $\bx_i$ belong to some non-linear subspace, e.g.\ for sparse $\bx_i$? 
(v)~How do the relevant parameters change when we have multiple (but less than $r$) receive signals? Answers to these questions
could be particularly relevant in connection with the future Internet-of-Things.

\section*{Acknowledgement}

The authors want to thank Holger Boche and Benjamin Friedlander for insightful discussions on the topic of the paper. The authors are also grateful to the anonymous referees for their careful reading of this paper and many suggestions, which helped to improve the manuscript. 

\section{Appendix}
\label{s:appendix}

\subsection{Some useful auxiliary results}

The key concentration inequality we use throughout our paper comes from Proposition 2 in~\cite{KolVal11,VK11}.

\begin{theorem}\label{BernGaussian}
Consider a finite sequence of $\CZ_l$ of independent centered random matrices with dimension $M_1\times M_2$. Assume that  $\|\CZ_l\|_{\psi_1} \leq R$ where the norm $\|\cdot\|_{\psi_1}$ of a matrix is defined as
\begin{equation}\label{def:psi-1}
\|\BZ\|_{\psi_1} := \inf_{u \geq 0} \{ \E[ \exp(\|\BZ\|/u)] \leq 2 \}.
\end{equation}
and introduce the random matrix 
\begin{equation*}
\BS = \sum_{l=1}^Q \CZ_l.
\end{equation*}
Compute the variance parameter
\begin{equation}\label{sigmasq}
\sigma^2 = \max\Big\{ \| \sum_{l=1}^Q \E(\CZ_l\CZ_l^*)\|, \| \sum_{l=1}^Q \E(\CZ_l^* \CZ_l)\| \Big\},
\end{equation}
then for all $t \geq 0$, we have the tail bound on
the operator norm of $\BS$, 
\begin{align}\label{thm:bern}
\begin{split}
\|\BS\| &
 \leq C_0 \max\Big\{ \sigma \sqrt{t + \log(M_1 + M_2)},  R\log\left( \frac{\sqrt{Q}R}{\sigma}\right)(t + \log(M_1 + M_2)) \Big\}
\end{split}
\end{align}
with probability at least $1 - e^{t}$ where $C_0$ is an absolute constant. 
\end{theorem}

For convenience we collect some results used throughout the proofs. 
Before we proceed, we note that there is a quantity equivalent  to $\|\cdot\|_{\psi_1}$ defined in~\eqref{def:psi-1}, i.e., 
\begin{equation}\label{def:psi-2}
c_1 \sup_{q\geq 1}q^{-1} (\E |z|^q)^{1/q} \leq \|z\|_{\psi_1} \leq c_2  \sup_{q\geq 1}q^{-1} (\E |z|^q)^{1/q},
\end{equation}
where $c_1$ and $c_2$ are two universal positive constants, see Section 5.2.4 in~\cite{Ver10}. Therefore, $\sup_{q\geq 1}q^{-1} (\E |z|^q)^{1/q}$ will be used to quantify $\|z\|_{\psi_1}$ in this section since it is easier to use in explicit calculations.

\begin{lemma} \label{lem:6JB}
Let $z$ be a random variable which obeys $\Pr\{ |z| > u \} \leq a e^{-b u }$, then
\begin{equation*}
\|z\|_{\psi_1} \leq (1 + a)/b
\end{equation*}
which is proven in Lemma 2.2.1 in~\cite{VW96}. Moreover, it is easy to verify that for a scalar $\lambda\in \CC$
\begin{equation*}
\|\lambda z\|_{\psi_1}  = |\lambda| \|z\|_{\psi_1}.
\end{equation*}
For another independent random variable $\bw$ with an exponential tail
\begin{equation}\label{lem:zw}
\|z + w\|_{\psi_1} \leq C(\|z\|_{\psi_1} + \|w\|_{\psi_1})
\end{equation}
for some universal contant $C$.
\end{lemma}
\begin{proof}
We only prove~\eqref{lem:zw} by using the equivalent quantity introduced in~\eqref{def:psi-2}. Suppose that both $z$ and $w$ yield~\eqref{def:psi-2}, there holds
\begin{eqnarray*}
\|z+w\|_{\psi_1} & \leq & c_2 \sup_{q\geq 1}q^{-1}(\E |z + w|^q)^{1/q} \\
& \leq & c_2 \sup_{q\geq 1}q^{-1} \left[(\E |z|)^{1/q} + (\E |w|)^{1/q}\right] \\
& \leq & c_1c_2 (\|z\|_{\psi_1} + \|w\|_{\psi_1}),
\end{eqnarray*}
where the second inequality follows from triangle inequality on $L^p$ spaces.
\end{proof}

\begin{lemma} \label{lem:multiple1}
Let $\bu\in\RR^n \sim \mathcal{N}(\bzero, \I_n)$, then $\|\bu\|^2 \sim \chi^2_{n}$ and
\begin{equation}
\label{lem:8JB}
\| \|\bu\|^2 \|_{\psi_1} = \| \lag\bu, \bu\rag \|_{\psi_1}  \leq 2 n.
\end{equation}
Furthermore,
\begin{equation}
\label{lem:9JB}
\E\ls (\bu\bu^* - \I_n)^2 \rs = (n + 1)\I_n.
\end{equation}
\end{lemma}

\begin{lemma}[Lemma 10-13 in~\cite{RR12}] \label{lem:multiple2}
Let $\bu\in\RR^n\sim\mathcal{N}(\bzero, \I_n)$ and $\bq\in\CC^n$ be any deterministic vector, then the following properties hold
\begin{equation}
|\lag \bu, \bq\rag|^2 \sim \|\bq\|^2 \chi^2_1,
\label{lem:7JBa}
\end{equation}
\begin{equation}
\| |\lag \bu, \bq\rag|^2  \|_{\psi_1} \leq C \|\bq\|^2,
\label{lem:7JB}
\end{equation}
\begin{equation}
\label{lem:10JB}
\| |\lag \bu, \bq\rag |^2 - \|\bq\|^2 \|_{\psi_1} \leq C \|\bq\|^2,
\end{equation}
\begin{equation}
\label{lem:11JB}
\| (\bu\bu^* - \I_n)\bq\|_{\psi_1} \leq C\sqrt{n}\|\bq\|,
\end{equation}
\begin{equation}
\label{lem:12JB}
\E\ls (\bu\bu^* - \I_n)\bq\bq^* (\bu\bu^* - \I_n)\rs = \|\bq\|^2 \I_n + \bar{\bq}\bar{\bq}^*.
\end{equation}
Let $\bp\in\CC^n$ be another deterministic vector, then
\begin{equation}\label{lem:13JB}
\| \lag \bu, \bq \rag \lag \bp, \bu\rag  - \lag \bq, \bp\rag\|_{\psi_1} \leq \|\bq\|\|\bp\|.
\end{equation}
\end{lemma}

\begin{proof}
~\eqref{lem:7JBa} to~\eqref{lem:11JB} and~\eqref{lem:13JB} directly follow from Lemma 10-13 in~\cite{RR12}, except for small differences in the constants. 
We only prove~\eqref{lem:12JB} here.
\begin{equation*}
\E\ls (\bu\bu^* - \I_n)\bq\bq^* (\bu\bu^* - \I_n)\rs = \E\ls |\lag \bu, \bq\rag|^2 \bu\bu^*\rs - \bq\bq^*.
\end{equation*}
For each $(i,j)$-th entry of $R_{ij} =  |\lag \bu, \bq\rag|^2 u_i u_j = \bq^* \ls u_iu_j \bu\bu^*\rs \bq$. 
\begin{equation*}
\E\ls u_iu_j \bu\bu^*\rs 
= 
\begin{cases}
\BE_{ij} + \BE_{ji} & i\neq j \\
\I_n + \BE_{ii} & i = j
\end{cases}
\end{equation*}
where $\BE_{ij}$ is an $n\times n$ matrix with  the $(i,j)$-th entry equal to 1 and the others being $0.$ The expectation of $R_{ij}$
\begin{equation*}
\E R_{ij} = 
\begin{cases}
q_i^*q_j + q_j^*q_i & i \neq j\\
\|\bq\|^2 + |q_i|^2 & i = j
\end{cases}
\end{equation*}
and
\begin{align*}
\E\ls |\lag \bu, \bq\rag|^2 \bu\bu^*\rs - \bq\bq^* & = \|\bq\|^2 \I_n + \bq\bq^* + \bar{\bq}\bar{\bq}^* - \bq\bq^* \\
& = \|\bq\|^2 \I_n + \bar{\bq}\bar{\bq}^*
\end{align*}
where $\bar{\bq}$ is the complex conjugate of $\bq$.
\end{proof}

\begin{lemma}\label{lemma:psi}
Assume $\bu\sim\mathcal{N}(\bzero,\I_{n})$ and $\bv\sim\mathcal{N}(\bzero,\I_{m})$ are two independent Gaussian random vectors, then
\begin{equation*}
\left\| \|\bu\|^2 + \|\bv\|^2\right\|_{\psi_1}  \leq n+m
\end{equation*}
and
\begin{equation*}
\left\| \|\bu\| \cdot \|\bv\|\right\|_{\psi_1} \leq C\sqrt{mn}.
\end{equation*}

\end{lemma}
\begin{proof}
Let us start with the first one. 
\begin{equation*}
\left\| \|\bu\|^2 + \|\bv\|^2\right\|_{\psi_1}  \leq  \| \|\bu\|^2 \|_{\psi_1} +  \| \|\bu\|^2 \|_{\psi_1} \leq n + m,
\end{equation*}
which directly follows from~\eqref{lem:zw} and~\eqref{lem:8JB}.
Following from independence, 
\begin{align*}
\left\| \|\bu\| \cdot \|\bv\|\right\|_{\psi_1}  & \leq c_2  \sup_{q\geq 1} q^{-1} (\E \|\bu\|^q \|\bv\|^q)^{1/q}  \leq c_2\sup_{q} q^{-1} (\E \|\bu\|^q )^{1/q} (\E \|\bv\|^q )^{1/q}.
\end{align*}
Let $t = q/2$, 

\begin{align*}
 \left\| \|\bu\| \cdot \|\bv\|\right\|_{\psi_1}  & \leq c_2 \sup_{t\geq 1} \frac{1}{2t} (\E \|\bu\|^{2t} )^{1/2t} (\E \|\bv\|^{2t} )^{1/2t}  \\
& \leq \frac{c_2}{2} \left( \sup_{t\geq 1} \frac{1}{t} (\E \|\bu\|^{2t} )^{1/t} \right)^{1/2}  \left( \sup_{t\geq 1} \frac{1}{t} (\E \|\bv\|^{2t} )^{1/t} \right)^{1/2} \\
&  \leq \frac{c_1c_2}{2} \sqrt{\| \bu \|_{\psi_1}  \cdot \| \bv \|_{\psi_1}} \leq  C\sqrt{mn},
\end{align*}
where $\|\bu\|^2\sim\chi^2_n$ and $\|\bv\|^2 \sim \chi^2_m$ and  $\|\bu\|_{\psi_1}$ and $\|\bv\|_{\psi_1}$ are given by~\eqref{lem:8JB}.
\end{proof}

\subsection{A useful fact about the ``low-frequency" DFT matrix}
\label{sub:fourier}
Suppose that $\BB$ is a ``low-frequency'' Fourier matrix, i.e., 
\begin{equation*}
\BB = \frac{1}{\sqrt{L}} (e^{-2\pi i lk/L})_{l,k} \in \CC^{L\times K},
\end{equation*}
where $1\leq k\leq K$ and $1\leq l\leq L$ with $K \le L$.
Assume there exists a $Q$ such that $L = QP$ with $Q \geq K$. We choose $\Gamma_p = \{p, P + p, \cdots , (Q - 1)P + p\}$ with $1\leq p < P$ such that $|\Gamma_p| = Q$, $\bigcup_{1\leq p\leq P} \Gamma_p = \{1,\cdots, L\}$ and they are mutually disjoint. Let $\BB_p$ be the $Q\times K$ matrix by choosing its rows from those of $\BB$ with indices in $\Gamma_p$. Then we can rewrite $\BB_p$ as 
\begin{equation*}
\BB_{p} = \frac{1}{\sqrt{L}}(e^{-2\pi i (tP - P +  p)k/(PQ)})_{1\leq t\leq Q, 1\leq k\leq K} \in \CC^{Q\times K},
\end{equation*}
and it actually equals
\begin{equation*}
\BB_{p} = \frac{1}{\sqrt{L}} (e^{-2\pi i tk/Q} e^{2\pi i (P - p)k/(PQ)} )_{1\leq t\leq Q, 1\leq k\leq K}\in \CC^{Q\times K}.
\end{equation*}
Therefore
\begin{equation*}
\BB_{p} = \sqrt{\frac{Q}{L}}\BF_{Q}  \diag (e^{2\pi i (P - p)/(PQ)}, \cdots, e^{2\pi i K (P - p)/(PQ)}  )
\end{equation*}
where $\BF_Q$ is the first $K$ columns of a $Q\times Q$ DFT matrix with $\BF_Q^*\BF_Q = \I_K$. There holds
\begin{equation*}
\sum_{l\in\Gamma_p} \bb_l\bb_l^* = \BB_p^*\BB_p = \frac{Q}{L} \I_K
\end{equation*}
where $\bb_l$ is the $l$-th column of $\BB^*.$

\bibliographystyle{abbrv}

%





\end{document}